\documentclass[11pt,onecolumn, letterpaper]{IEEEtran} 

\usepackage{amsmath,amstext,amsfonts,amssymb,url,psfrag,epsfig}
\usepackage{hyperref}
\usepackage[usenames]{color}

\newcommand{\blist}{    \begin{list}{$\bullet$}{\topsep 0.0in \partopsep 0.0in                 \itemsep 0.05in \parsep                 0.0in \leftmargin 0.3in}} 
\newcommand{\elist}{\end{list}}

\newcommand{\bear}{\begin{eqnarray}}
\newcommand{\ear}{\end{eqnarray}}
\newcommand{\bears}{\begin{eqnarray*}}
\newcommand{\ears}{\end{eqnarray*}}

\newtheorem{theorem}            {Theorem}
\newtheorem{lemma}              {Lemma}
\newtheorem{definition}         {Definition}
\newtheorem{corollary}          {Corollary}
\renewenvironment{proof}{\noindent{\bf Proof:}\\\,\,}{$\hfill\bullet$\\[0\baselineskip]}
\newenvironment{proofa}[1]{\noindent{\bf Proof (of {#1}):}\\\,\,}{$\hfill\bullet$\\[0\baselineskip]}

\newcommand{\DEF}{\ \triangleq \ } 

\newcommand{\PX}[1]{\Pr\left[{#1}\right] }  
\newcommand{\PCX}[2]{\PX{ \left.\! {#1} \right| {#2} }} 
\newcommand{\EX}[1]{E\left[{#1}\right]} 
\newcommand{\ECX}[2]{\EX{ \left.\! {#1} \right| {#2} }} 
\newcommand{\SPX}[2]{\Pr\thinspace^{(#1)}\left[{#2}\right] }  
\newcommand{\SPCX}[3]{\SPX{#1}{ \left.\! {#2} \right| {#3} }} 

\newcommand{\CKLD}[3]{D\left( \left.\! {#1} \right\|  \left.\! {#2} \right| {#3}  \right ) }
\newcommand{\KLD}[2]{D\left( \left.\! {#1} \right\|{#2} \right)}
\newcommand{\CX}{{{\emph C}}}
\newcommand{\DX}{{\emph D_{\max}}}

\newcommand{\PCAP}{\tilde \CX} 
\newcommand{\IND}[1]  {{\mathbb{I}}_{\left\{ {#1} \right\}}} 

\newcommand{\Eb}[0] {{E_\textrm{b}}}  
\newcommand{\Emd}[0] {{E_{\textrm{md}}}} 
\newcommand{\Emdr}[0] {{E_{\textrm{md}}}} 
\newcommand{\Emdre}[0] {{E_{\textrm{md}}^{{ \it df}}}} 
\newcommand{\Efa}[0] {{E_{\textrm{fa}}}} 
\newcommand{\Efal}[0] {{E_{\textrm{fa}}^{\textrm{l}}}} 
\newcommand{\Efau}[0] {{E_{\textrm{fa}}^{\textrm{u}}}} 

\newcommand{\FEb}[0] {{E_{\textrm{b}}^{{\it f}}}} 
\newcommand{\FEbr}[0] {{E_{\textrm{bits}}^{{\it f}}}} 
\newcommand{\FEmd}[0] {{E_{\textrm{md}}^{ {\it f}}}} 
\newcommand{\FEmdr}[1] {{E_{\textrm{md}}^{{\it f}}({#1})}} 
\newcommand{\FEfa}[0] {{E_{\textrm{fa}}^{{ \it f}}}} 

\newcommand{\Esp}[0] {E_{\textrm{sp}}}
\newcommand{\Ee}[0] {E}

\newcommand{\SC}[0]{ {{\cal Q} }}  
\newcommand{\SCe}[0]{ {{\cal Q}}}  
\newcommand{\SCf}[0]{ {{\cal Q}}}  

\newcommand{\blx}[0]  {{ \tau}} 
\newcommand{\BLX}[0]  {{n}}    
\newcommand{\ts}[0]  {\blx_{\delta}}
\newcommand{\inx}[0]  {k}

\newcommand{\mes}[0]{ {\it M }}   
\newcommand{\ames}[0]  { \textrm{M}'}
\newcommand{\amesX}[0]  {{\cal M}'}
\newcommand{\partit}[0]  {{\it \Theta}}

\newcommand{\RX}[0] {{R}}          
\newcommand{\MX}[0] {{\cal M}}   
\newcommand{\FX}[1]{ {\cal J}\left({#1}\right) }  

\newcommand{\h}[1]{ {\mathcal H}(\mes|#1)}
\newcommand{\HXz}[0]{{\mathcal H}(\mes)}
\newcommand{\HX}[1]{\h{Y^{{#1}}}}
\newcommand{\htwo}[1]{ {\mathcal H}(\mes_2|#1)}
\newcommand{\ha}[1]{ {\mathcal H}(\mes_1|#1)}
\newcommand{\HXa}[1]{\ha{Y^{{#1}}}}
\newcommand{\MI}[2]{ {\mathcal I} \left( {#1}; {#2} \right)}
\newcommand{\CMI}[3]{ {\mathcal I} \left( {#1}; {#2} \left\vert {#3} \right. \right)}

\newcommand{\Ah}[1]{ {\mathcal H}(\ames|#1)}
\newcommand{\AHX}[1]{\Ah{Y^{{#1}}}}
\newcommand{\APe}[1] {{\it P}_{e}^{\ames}\{ {#1}\}}

\newcommand{\DEC}[1]{{\cal G}({#1})}

\newcommand{\bx}{{\bar x}}

\newcommand{\rvb}{{\bf b}}

\newcommand{\brx}{{\bar \rvx}}

\newcommand{\era}{{\textrm{erasure}}}
\newcommand{\Pe}[0] {{P_e}}  
\newcommand{\Pera}[0] {{P_{\era}}}  

\newcommand{\slack}{{\epsilon}}

\newcommand{\xr}{{x_r}}
\newcommand{\xfu}{{x_{f_u}}}
\newcommand{\xfl}{{x_{f_l}}}
\newcommand{\xa}{{x_a}}
\newcommand{\xd}{{x_d}}

\newcommand{\psed}{\textrm{red-alert exponent}}

\newcommand{\DIST}[1] { {\cal P}\left( {{#1}} \right)}

\newcommand{\fix}{{\bar \rvx'}}
\newcommand{\ctype}[2] { {\sf V} ({#1},{#2}) }
\newcommand{\typey}[1] { {\sf Q}_{ \left( {#1} \right) }}
\newcommand{\shell}[2] { {\sf T}_{{#1}} \left({#2}\right)}
\newcommand{\shelly}[1] { {\sf T}_{{#1}} }
\newcommand{\shellx}[1] { {\sf T}_{{#1}} }
\newcommand{\Dshell}[1] { {\cal S}_{{#1}, V}^{(\BLX)}} 
\newcommand{\DshellU}[1] { {\cal S}_{{#1}}^{(\BLX)}} 
\newcommand{\PYIID}[1] { \mathbb{P}_Y^{\BLX} \left( {{#1}} \right)} 
\newcommand{\uep}[0] {{\it UEP}}

\newcommand{\lcb}{\left\{}
\newcommand{\rcb}{\right\}}
\renewcommand{\(}{\left( }
\renewcommand{\)}{\right)}

\newcommand{\rvx}{{\bf x}}

\newenvironment{intexp}{\vspace{.2cm}\noindent{\textbf{Intuitive interpretation:~}}}{\\}
\newenvironment{optstr}{\vspace{.2cm}\noindent{\textbf{Optimal strategy:}}}{\\}
\newenvironment{strl}  {\vspace{.2cm}\noindent{\textbf{Strategy to reach lower bound}}}{\\}

\begin{document}

\title{Unequal Error Protection:\\
An Information Theoretic Perspective\thanks{This research is supported by
DARPA ITMANET project and an AFOSR grant FA9550-06-0156.
Initial part of this paper was submitted to IEEE
International Symposium on Information Theory, 2008.}}
\author{Shashi Borade \quad  Bar\i\c{s} Nakibo\u{g}lu \quad Lizhong Zheng \\
EECS, Massachusetts Institute of Technology\\ \{ spb , \ nakib ,\ lizhong \}\ @mit.edu}
\date{\today}

\maketitle
\begin{abstract} 
An information theoretic framework for unequal error protection is developed in terms of the exponential error bounds. The fundamental difference between the \emph{bit-wise} and \emph{message-wise} unequal error protection (\uep) is demonstrated, for fixed length block codes on DMCs without feedback. Effect of feedback is investigated via variable length block codes. It is shown that, feedback  results in a significant improvement in both   \emph{bit-wise} and \emph{message-wise} \uep~(except the single message case for missed detection). The distinction between false-alarm and missed-detection formalizations for \emph{message-wise} \uep~ is also considered. All results presented are at rates close to capacity. \end{abstract}

\section{Introduction}
Classical theoretical framework for communication \cite{shannon} assumes that all information is equally important. In this framework, the communication system aims to provide a uniform error protection to all messages: any particular message being mistaken as any other is viewed to be equally costly. With such uniformity assumptions, reliability of a communication scheme is measured by either the average or the worst case probability of error, over all possible messages to be transmitted. In information theory literature, a communication scheme is said to be {\it reliable} if this error probability can be made small. Communication schemes designed with this framework turn out to be optimal in sending any source over any channel, provided that long enough codes can be employed. This homogeneous view of information motivates the universal interface of ``bits'' between any source and any channel \cite{shannon}, and is often viewed as  Shannon's most significant contribution.

In many communication scenarios, such as wireless networks, interactive systems, and control applications, where uniformly good error protection becomes a luxury; providing such a protection to the entire information might be wasteful, if not infeasible. Instead, it is more efficient here to protect a crucial part of information better than the rest. For example,
\begin{itemize}
\item In a wireless network, control signals like channel state, power control, and scheduling information are often more important than the payload data, and should be protected more carefully. Thus  even though the final objective is delivering the payload data, the physical layer should provide a better protection to such protocol information. Similarly for the Internet, packet headers are more important for delivering the packet and need better protection to ensure that the actual data gets through. 
\item Another example is transmission of a multiple resolution source code. The coarse resolution needs a better protection than the fine resolution so that the user at least obtains some crude reconstruction after bad noise realizations. 
\item Controlling unstable plants over noisy communication link  \cite{sahai1} and compressing unstable sources \cite{sahai2} provide more examples where different parts of information need different reliability.
\end{itemize}
In contrast with the classical homogeneous view, these examples demonstrate the heterogeneous nature of information. Furthermore the practical need for unequal error protection (\uep) due to this heterogeneity demonstrated in these examples is the reason why we need to go beyond the conventional content-blind information processing.

Consider a message set ${\cal M}=\{1, 2,3, \ldots, 2^k\}$ for a block code. Note that members of this set, i.e.  ``messages'', can also be represented by length $k$ strings of information bits, $\rvb=[b_1, b_2, \ldots b_k]$. A block code is composed of an encoder which maps the messages, $M \in  {\cal M}$ into channel inputs and a decoder which maps channel outputs to decoded message, $\hat{\mes} \in {\cal M}$. An error event for a block code is $\{\hat{\mes} \neq \mes\}$.  In most information theory texts, when an error occurs, the entire bit sequence $\rvb$ is rejected. That is, errors in decoding the message and in decoding the information bits are treated similarly. We avoid this,  and try to figure out what can be achieved by analyzing the errors of different subsets of bits separately.

In the existing formulations of unequal error protection codes \cite{trott} in coding theory, the information bits are partitioned  into subsets, and the decoding errors in different subsets of bits are viewed as different kinds of errors. For example, one might want to provide a better protection to one subset of bits by ensuring that errors in these bits are less probable than the other bits. We call such problems as ``bit-wise \uep''. Previous examples of packet headers, multiple resolution codes, etc. belong to this category of \uep.

However, in some situations, instead of \emph{bits} one might want  to provide a better protection to a subset of \emph{messages}. For example, one might consider embedding a special message in a normal $k$-bit code, i.e., transmitting one of $2^k+1$ messages, where the extra message has a special meaning and requires a smaller error probability. Note that the error event for the special message is not associated to error in any particular bit or set of bits. Instead, it corresponds to a particular bit-sequence (\emph{i.e}. message) being decoded as some other bit-sequence. Borrowing from hypothesis testing, we can define two kinds of errors corresponding to a special message.

\begin{itemize}
\item
\emph{Missed-detection} of a message $i$ occurs when transmitted message $\mes$ is $i$ and decoded message  $\hat{\mes}$ is some other message $j\neq i$. Consider a special message indicating some system emergency which is too costly to be missed. Clearly, such  special messages demand a small missed detection probability.  Missed detection probability of a message is  simply the conditional error probability after its transmission.
\item
\emph{False-alarm} of a message $i$ occurs when transmitted message $\mes$ is some other message $j\neq i$ and decoded message $\hat{M}$ is $i$. Consider the reboot message for a remote-controlled system such as a robot or a satellite or the ``disconnect'' message to a cell-phone. Its false-alarm could cause unnecessary shutdowns and other system troubles. Such special messages demand small false alarm probability. 
\end{itemize}

 We call such problems as ``message-wise \uep''. In conventional framework, every bit is as important as every other bit and every message is as important as every other message. In short in conventional framework it is assumed that all the information is ``created equal''. In such a  framework there is no reason to distinguish between  bit-wise or message wise error probabilities because message-wise error probability is larger than bit-wise error probability by an insignificant factor, in terms of exponents. However, in the \uep~ setting, it is necessary to differentiate between  message-errors and bit-errors. We will see that in many situations, error probability of special bits and messages have behave  very differently.

The main contribution of this paper is a set of results, identifying the performance limits and optimal coding strategies, for a variety of \uep~ scenarios. We focus on a few simplified notions of \uep, most with immediate practical applications, and try to illustrate the main insights for them. One can imagine using these \uep~ strategies for embedding protocol information within the actual data. By eliminating a separate control channel, this can enhance the overall bandwidth and/or energy efficiency.

For conceptual clarity, this article focuses exclusively on situations where the data rate is essentially equal to the channel capacity. These situation can be motivated by the scenarios where data rate is a crucial system resource that can not be compromised. In these situations, no positive error exponent in the conventional sense can be achieved. That is, if we aim to protect the entire information uniformly well, neither bit-wise nor message-wise error probabilities can decay exponentially fast with increasing code length. We ask the question then ``can we make the error probability of a particular bit, or  a particular message, decay exponentially fast with block length?''

 When we break away from the conventional framework and start to provide better protection to against certain kinds of errors, there is no reason to restrict ourselves by assuming that those errors are erroneous decoding of some particular \emph{bits} or missed detections or false alarms  associated with  some particular messages. A general formulation of \uep~ could be an arbitrary combination of protection demands against some specific kinds of errors. In this general definition of  \uep, bit-wise \uep~ and message-wise \uep~ are simply two particular ways of specifying which kinds of errors are too costly compared to others.

In the following, we start by specifying the channel model and giving some basic definitions in Section \ref{sec:model}. Then  in section \ref{sec:nofeed} we discuss bit-wise \uep~ and message-wise \uep~ for block codes without feedback. Theorem \ref{thm:bit} shows that  for data-rates approaching capacity, even a single bit cannot achieve  any positive error exponent. Thus in bit-wise \uep, the data-rate must back off from capacity for achieving any positive error exponent even for a single bit. On the contrary, in message-wise \uep, positive error exponents can be achieved even at capacity.  We first  consider the case when there is only one special message and show that,   Theorem \ref{thm:md}, optimal (missed-detection) error exponent for the special message is equal to the \emph{red-alert exponent}, which is defined in section \ref{sec:nofeedsm}.   We then consider situations where an exponentially large number of messages are special and each special message  demands  a positive (missed detection) error exponent. (This situation has previously been analyzed before in \cite{csiszar1}, and a result closely related to our has been reported there.)  Theorem \ref{thm:mdmany}  shows a surprising result that these special messages can achieve  the same exponent as if all the other (non-special) messages  were absent.  In other words, a capacity achieving code and an error  exponent-optimal code below capacity can coexist without hurting  each other. These results also shed some new light on the structure  of capacity achieving codes.

Insights from the block codes without feedback becomes useful in Section \ref{sec:feedback} where we investigate similar problems for variable  length block codes with feedback. Feedback together with variable decoding  time creates some fundamental connections between bit-wise \uep~ and message-wise \uep. Now even for bit-wise \uep, a  positive error exponent can be achieved at capacity. Theorem \ref{thm:bitf} shows that a single special bit can achieve the same exponent as a single special message---the $\psed$. As the number of special bits increases, the achievable exponent for them decays linearly with their rate as shown in Theorem \ref{thm:l1}. Then Theorem \ref{thm:many} generalizes this result to the case when there are multiple levels of specialty---most special, second-most special and so on. It uses a strategy similar to onion-peeling  and achieves error exponents which are successively refinable over multiple layers. For  single special message case, however, Theorem \ref{thm:mdf} shows that feedback does not improve the optimal missed detection exponent. The case of exponentially many  messages is resolved in Theorem \ref{thm:msru}.  Evidently  many special messages cannot achieve an exponent higher than that of a single special message, i.e. $\psed$.  However it turns out that the special messages can reach $\psed$ at rates below a certain threshold, as if all the other special messages were absent. Furthermore for the rates above the very same threshold, special messages reach the corresponding value of Burnashev's exponent, as if all the ordinary messages were absent.

Section \ref{sec:fa} then addresses message-wise \uep~ situations where special messages demand small probability of false-alarms instead of missed-detections. It considers the case of fixed length block codes with out feedback as well as variable length block codes with  feedback. This discussion for false-alarms was postponed from earlier sections to avoid confusion with the missed-detection results in earlier sections. Some future directions are discussed briefly in Section \ref{sec:summary}.

After discussing each theorem, we will provide a brief description of the optimal strategy, but refrain from detailed technical discussions.  Proofs can be found in later sections. In section \ref{sec:proofs} and section \ref{sec:proofsf}  we will present  the proofs of the results in Section \ref{sec:nofeed},  on block codes without feedback, and Section \ref{sec:feedback}, on variable length  block codes with feedback, respectively. Lastly in Section \ref{sec:proofsfa} we discuss the proofs for the false-alarm results  of Section \ref{sec:fa}. Before going into the presentation of our work let us give a very brief overview of the previous work on the problem, in different fields.

\subsection{Previous Work and Contribution}\label{sec:prev}
 The simplest method of unequal error protection is to allocate different channels for different types of data. For example, many wireless systems allocate a separate ``control channel'', often with short codes with low rate and low spectral efficiency, to transmit control signals with high reliability. The well known Gray code, assigning similar bit strings to close by constellation points, can be viewed as \uep: even if there is some error in identifying the transmitted symbol, there is a good chance that some of the bits are correctly received. But clearly this approach is far from addressing the problem in any effective way.

 The first systematic consideration of problem in coding theory was within the  frame work of linear codes. In \cite{wolf}, Masnick and Wolf suggested techniques which  protects different parts (bits) of the message against different number of channel errors (channel symbol conversions).  This frame work has extensively studied over the years in \cite{uep1}, \cite{uep2}, \cite{uep3}, \cite{uep4}, \cite{uep5}, \cite{uep6}, \cite{uep7} and in many others. Later issue is addressed within frame work of Low Density Parity Check (LDPC) codes too \cite{vas}, \cite{pou}, \cite{pou2}, \cite{rah},  \cite{ext1}, and \cite{ext2}.

``Priority encoded transmission'' (PET) was suggested by Albenese et.al. \cite{pet} as an alternative model of the problem, with packet erasures. In this approach guarantees are given not in terms of channel errors but packet erasures. Coding and modulation issues are addressed simultaneously in \cite{uep}. For wireless channels, \cite{diggavi} analyzes this problem in terms of diversity-multiplexing trade-offs.

In contrast with above mentioned work, we pose and address the problem within the information theoretic frame work. We work with  the error probabilities and refrain from making assumptions about the particular block code used while proving our converse results. This is the main difference between our approach and the prevailing  approach within the coding theory community.

In \cite{bas}, Bassalygo {\it et. al.}  considered the error correcting codes whose messages are composed of two group of bits, each of which requires different level of protection against channel errors and provided inner and outer bounds to the achievable performance, in terms of hamming distances and rates. Unlike other works within coding theory frame work,  they do not make any assumption about the code. Thus their results can indeed be reinterpreted in our  framework as a result for bit wise \uep, on binary symmetric channels.

 Some of the the \uep~ problems have already been investigated within the framework of information theory too. Csisz\'ar   studied message wise \uep~ with many messages in  \cite{csiszar1}. Moreover  results in  \cite{csiszar1} are not restricted to the rates close to capacity, like ours. Also messages wise \uep~ with single special message was dealt with in \cite{kud} by Kudryashov. In \cite{kud}, an \uep~ code with  single special message is used as a subcode within a variable delay communication scheme. The scheme proposed in \cite{kud} for the single special message case  is a key building block in many of the results in section \ref{sec:feedback}.  However the optimality of the scheme was not proved in \cite{kud}. We show that it is indeed optimal.

The main contribution of the current work is the proposed frame work for \uep~ problems within information theory.  In addition to the particular results presented  on different problems and the contrasts demonstrated between different scenarios,  we believe the proof techniques used in subsections\footnote{The key idea in  subsection \ref{sec:proofthmsix}  is a generalization of the approach presented in \cite{burna2}.}   \ref{sec:proofthmone}, \ref{sec:proofthmsix} and \ref{subsec:fmdconverse} are novel and they are promising for the future work in the field.

\section{Channel Model and Notation}
\label{sec:model}
\subsection{DMC's  and Block Codes}
We consider a discrete memoryless channel (DMC) $W_{Y|X}$, with input alphabet ${\cal X}=\{1,2,\ldots,|{\cal X}|\}$ and output alphabet ${\cal Y}=\{1,2,\ldots,|{\cal Y}|\}$. The conditional distribution of output letter $Y$ when the channel input letter $X$ equals  $i\in{\cal X}$ is denoted by $W_{Y|X}(\cdot|i)$.
\begin{equation*}
 \PCX{Y=j}{X=i}=W_{Y |X}(j|i)  \qquad \forall i\in{\cal X},\ \forall j\in{\cal Y}.
\end{equation*}
We assume that all the entries of the channel transition matrix are positive, that is, every output letter is reachable from every input letter. This assumption is indeed a crucial one. Many of the results  we present in this paper  change when there are zero-probability  transitions.

A  length $\BLX$  block code without feedback  with message set ${\cal M}= \{1,2,\ldots,|{\cal M}|\}$  is  composed of two mappings, encoder mapping and  decoder mapping. Encoder mapping assigns a length $\BLX$ codeword,\footnote{Unless mentioned otherwise, small letters (e.g. $x$) denote a particular value of the corresponding random variable denoted in capital letters (e.g. $X$).} 
\begin{equation*}
\bx^{\BLX}(k)\stackrel{\Delta}=(\bx_1(k),\bx_2(k)\cdots,\bx_\BLX(k)) \qquad \forall k \in {\cal M}  
\end{equation*}
where $\bx_t(k)$ denotes the input at time $t$ for message $k$. Decoder mapping, $\hat{\mes}$, assigns a message to each possible channel output sequence, i.e. $\hat{\mes}: {\cal Y}^{\BLX} \rightarrow {\cal M}$.

At time zero, the transmitter  is given the message $\mes$, which is chosen from ${\cal M}$ according to a uniform distribution. In the following $\BLX$ time units, it sends the corresponding codeword. After observing $Y^{\BLX}$, receiver  decodes a message. The error probability $\Pe$ and rate $\RX$ of the code is given by
\begin{equation*}
\Pe\DEF \PX{\hat{\mes} \neq \mes}\qquad\textrm{and} \qquad \RX\DEF \tfrac{\ln |\MX|}{\BLX}.
\end{equation*}

\subsection{Different Kinds of Errors}
\label{sec:modele}
 While discussing message-wise \uep, we  consider the conditional error probability for a particular message $i\in{\cal M}$,
\begin{equation*}
\PCX{\hat{\mes} \neq i}{\mes=i}.
\end{equation*}
Recall that this  is the same as the missed detection probability for message $i$.

On the other hand when we are talking about bit-wise \uep, we  consider message sets that are of the form ${\cal M}={\cal M}_1 \times {\cal M}_2$. In such cases message $\mes$ is composed of two submessages, $\mes=(\mes_1,\mes_2)$.  First submessage $\mes_1$ corresponds to the high-priority bits while second submessage $\mes_2$ corresponds to the low-priority bits. The uniform choice of $\mes$ from ${\cal M}$,  implies the uniform and independent choice of $\mes_1$ and $\mes_2$  from ${\cal M}_1$ and ${\cal M}_2$ respectively. Error probability of a submessage  $\mes_j$ is  given by
\begin{equation*}
\PX{\hat{\mes}_j \neq \mes_j}  \qquad j=1,2
\end{equation*}
Note that the overall message $\mes$ is decoded incorrectly when either $\mes_1$ or $\mes_2$ or both are decoded incorrectly. The goal of bit-wise \uep~ is to achieve best possible  $\PX{\hat{\mes}_1 \neq \mes_1}$ while ensuring a reasonably small  $\Pe=\PX{\hat{\mes}\neq \mes}$.

\subsection{Reliable Code Sequences}
We focus  on systems where reliable communication is achieved in order to find exponentially tight bounds for error probabilities of special parts of information. We use the notion of code-sequences to simplify our discussion. 

 A sequence of codes indexed by their  block-lengths is called \emph{reliable} if
\begin{equation*}
  \lim_{\BLX \rightarrow \infty} \Pe^{(\BLX)}=0
\end{equation*}
For any reliable code-sequence $\SC$, the rate $\RX_{\SC}$ is given by
\begin{equation*}
  \RX_{\SC} \DEF \liminf_{\BLX \rightarrow \infty} \tfrac{\ln |{\MX}^{(\BLX)}|}{\BLX}
\end{equation*}
The (conventional) error  exponent  of  a reliable sequence is
then
\begin{equation*}
 E_{\SC} \DEF \liminf_{\BLX \rightarrow \infty} \tfrac{- \ln \Pe^{(\BLX)}}{\BLX}
\end{equation*}
Thus the number of messages in ${\cal Q}$ is\footnote{The $\doteq$ sign denotes equality in the exponential sense. For a sequence $a^{(n)}$, \begin{equation*}  a^{(n)} \doteq e^{n F}  \Leftrightarrow    F=\liminf_{n \rightarrow \infty} \frac{\ln a^{(n)}}{n} \end{equation*}} $\doteq e^{\BLX \RX_{\SC}}$ and their average error probability decays like $e^{-\BLX E_{\SC}}$ with block length. Now we can define error exponent $E(\RX)$ in the conventional sense, which is equivalent to  the ones given in \cite{gallager_book}, \cite{sgb}, \cite{CK}, \cite{forney1}, \cite{forney}. 

\begin{definition}
   For any $\RX \leq \CX$ the error exponent $E(\RX)$ is defined as
\begin{equation*}
\Ee(\RX)\stackrel{\Delta}=\sup_{\substack{\SC: \RX_{\SC}\geq R }} E_{\SC}
\end{equation*}
\end{definition}

As mentioned previously, we are interested in \uep~ when operating at capacity. We already know, \cite{sgb}, that $ E(\CX) =0$, i.e. the overall error probability cannot decay exponentially at capacity. In the following sections, we show how certain parts of information can still achieve a positive exponent at capacity. In doing that, we are focusing only on  the reliable sequences  whose rates are equal to $\CX$. We call such reliable code sequences as \emph{capacity-achieving sequences}.

Through out the text we denote Kullback-Leibler (KL)  divergence between two distributions $\alpha_X(\cdot)$ and $\beta_X(\cdot)$ as $\KLD{\alpha_{X}(\cdot)}{\beta_{X}(\cdot)}$.
\begin{equation*}
  \KLD{\alpha_{X}(\cdot)}{\beta_{X}(\cdot)} = \sum_{i \in {\cal X}} \alpha_X(i) \ln \tfrac{\alpha_X(i)}{\beta_X(i)}
\end{equation*}
Similarly  conditional KL divergence between $V_{Y|X} (\cdot|\cdot)$ and $W_{Y|X}(\cdot|\cdot)$ under $P_X(\cdot)$ is given by
\begin{equation*}
 \CKLD{V_{Y|X}(\cdot|X)}{W_{Y|X}(\cdot|X)}{P_X} =\sum_{i \in {\cal X}} P_X(i) \sum_{j \in {\cal Y}} V_{Y|X}(j|i) \ln \tfrac{V_{Y|X}(j|i)}{W_{Y|X}(j|i)}
\end{equation*}
The output distribution that achieves the capacity is denoted by $P_Y^{*}$ and a corresponding input distribution is denoted by $P_{X}^{*}$.

\section{\uep~ at Capacity: Block Codes without Feedback}
\label{sec:nofeed}
\subsection{Special bit}
We first address the situation where one particular  {bit} (say the first)  out of the total $\log_2 |{\cal M}|$ bits is a special bit---it needs a much better error protection than the overall information. The error probability of the special bit is required to decay as fast as possible while ensuring reliable communication at capacity, for the overall code.  The single special bit is denoted by  $\mes_1$  where ${\cal M}_1=\{0,1\}$ and over all message $\mes$ is of the form $\mes=(\mes_1,\mes_2)$ where ${\cal M}={\cal M}_1 \times {\cal M}_2$. The optimal error exponent $\Eb$ for the special bit is then defined as follows\footnote{Appendix A discusses a different but equivalent type of definition and shows why it is equivalent to this one. These two types of definitions are equivalent for all the \uep~ exponents discussed in this paper.}.

\begin{definition}
   For a capacity-achieving sequence $\SC$ with message sets ${\cal M}^{(\BLX)}={\cal M}_1 \times {\cal M}_2^{(\BLX)}$ where ${\cal M}_1=\{0,1\}$, the special bit error exponent is defined as
\begin{equation*}
\Eb_{,\SC} \DEF \liminf_{\BLX \rightarrow \infty} \tfrac{ -\ln \SPX{\BLX}{\hat{\mes}_1 \neq\mes_1}}{\BLX}
\end{equation*}
Then $\Eb$ is defined as  $\Eb \DEF \sup_{\SC} \Eb_{,\SC}$.
\end{definition}

Thus if  $\SPX{\BLX}{\hat{\mes}_1 \neq \mes_1} \doteq \exp(-\BLX \Eb_{,\SC})$ for a reliable sequence $\SC$, then $\Eb$ is the supremum of $\Eb_{,\SC}$ over all capacity-achieving $\SC$'s.

Since $\Ee(\CX)=0$, it is clear that the entire information cannot achieve any positive error exponent at capacity. However, it is not clear whether a single special bit can steal a positive error exponent $\Eb$ at capacity.
\begin{theorem}
\[\Eb=0\] \label{thm:bit}
\end{theorem} 
This implies that, if we want the error probability of the messages to vanish with increasing block length and the error probability of at least one of the bits to decay with a positive exponent with block length, the rate of the code sequence should be strictly smaller than  the capacity.

Proof of the theorem is heavy in calculations, but the main idea behind is the ``blowing up lemma'' \cite{CK}. Conventionally, this lemma is only used for strong converses for various capacity theorems. It is also worth mentioning that the conventional converse techniques like Fano's inequality are not sufficient to prove this result.

\begin{intexp}
Let the shaded balls in Fig. \ref{fig:cluster} denote the minimal  decoding regions of the messages. These decoding regions ensure reliable communication, they are essentially  the typical noise-balls (\cite{cover-book}) around codewords. The decoding regions on the left of the thick line corresponds to $\hat{\mes}_1=1$ and those on the right correspond to the same when $\hat{\mes}_1=0$. Each of these halves includes half of the decoding regions. Intuitively, the blowing up lemma implies that if we try to add slight extra thickness to the left clusters in Figure  \ref{fig:cluster}, it blows up to occupy almost all the output space. This strange phenomenon in high dimensional spaces leaves no room for the right cluster to fit. Infeasibility of adding even slight extra thickness implies zero error exponent the special bit.
  \begin{figure}
   \centering
 \includegraphics[width=2.5in]{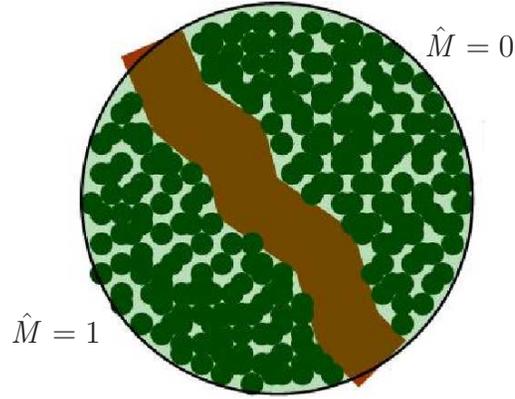}\\
   \setlength{\unitlength}{1cm}
   \begin{picture}(.1,.1)
     \put(-3.5,1.2){{\large  $\hat{\mes}=1$}}
     \put(2,5) {{\large  $\hat{\mes}=0$}}
   \end{picture}
\caption{Splitting the output space into $2$ distant enough clusters. }
  \label{fig:cluster}
  \end{figure}
\end{intexp}

\subsection{Special message}\label{sec:nofeedsm}
Now consider situations where one particular message (say $\mes=1$) out of the $\doteq e^{\BLX \CX}$ total messages is a special message---it needs a superior error protection.  The missed detection probability for this `emergency' message needs to be minimized. The best missed detection exponent $\Emd$ is defined as follows.\footnote{Note that the definition obtained by replacing $\SPCX{\BLX}{\hat{\mes} \neq 1}{\mes =1}$ by $\min_{j} \SPCX{\inx}{\hat{\mes} \neq j}{\mes =j}$ is equivalent to the one given above, since we are taking the supremum over $\SCf$ anyway. In short, the message $j$ with smallest conditional error probability could always be relabeled as message $1$.}

\begin{definition}
  For a capacity-achieving sequence $\SC$, missed detection exponent is defined as
\begin{equation*}
\Emd_{,\SC} \DEF \liminf_{\BLX \rightarrow \infty} \tfrac{ -\ln \SPCX{\BLX}{\hat{\mes} \neq 1}{\mes=1}}{\BLX}.
\end{equation*}
Then $\Emd$ is defined as $\Emd \DEF \sup_{\SC} \Emd_{,\SC}$.
\end{definition}

Compare this with the situation where we aim to protect  all the messages uniformly well. If all the messages demand equally good missed detection exponent, then no positive exponent is achievable at capacity. This follows from  the earlier discussion about $E(\CX)=0$. Below theorem shows the improvement in this exponent if we only demand it for a single message instead of all.

\begin{definition}
   The  parameter $\PCAP$ is defined\footnote{Authors would like to thank  Krishnan Eswaran of UC Berkeley for suggesting this name.} as the \emph{$\psed$} of a channel.
\begin{equation*}
  \PCAP \DEF \max_{i\in {\cal X}} \KLD{P_Y^*(\cdot)}{W_{Y|X}(\cdot|i)}  
\end{equation*}
We will denote the input letter achieving above maximum by $\xr$.
\end{definition}
\begin{theorem}
\[ \Emd=\PCAP.\]\label{thm:md}
\end{theorem}
Recall that Karush-Kuhn-Tucker (KKT) conditions for achieving capacity imply the following expression for capacity, \cite[Theorem 4.5.1]{gallager_book}.
\begin{equation*}
 \CX= \max_{i\in {\cal X}} \KLD{W_{Y|X}(\cdot|i)}{P_Y^*(\cdot)}
\end{equation*}
Note that simply switching the arguments of KL  divergence within the maximization for $\CX$,  gives us the expression for $\PCAP$. The capacity $\CX$ represents the best possible data-rate over a channel, whereas $\psed$ $\PCAP$ represents the best possible protection  achievable for a message at capacity.

It is worth mentioning here  the ``very noisy'' channel in \cite{gallager_book}. In this formulation  \cite{preprint},  the KL divergence is symmetric, which implies $\KLD{P_Y^*(\cdot)}{W_{Y|X}(\cdot|i)}\approx \KLD{W_{Y|X}(\cdot|i)}{P_Y^*(\cdot)}$. Hence the $\psed$ and capacity become roughly equal. For a symmetric channel like BSC, all inputs can be used as $\xr$. Since the $P_Y^*$ is the uniform distribution for these channels, $\PCAP=\KLD{P_Y^{*}(\cdot)}{W_{Y|X}(\cdot|i)}$ for any input letter $i$. This also happens to be the sphere-packing exponent $E_{\textrm{sp}}(0)$  of this channel \cite{sgb} at rate $0$.

\begin{optstr}
Codewords of a capacity achieving code are used for the ordinary messages. Codeword for the special message is a repetition sequence of the input letter $\xr$.  For all the output sequences special message is decoded, except for the output sequences with empirical distribution (type) approximately equal to $P_Y^*$. For the output sequences with empirical distribution  approximately  $P_Y^*$, the decoding scheme of the original capacity achieving code is used.
\end{optstr}
Indeed Kudryashov \cite{kud} had already suggested the encoding scheme described above, as a subcode for his  non-block variable delay coding scheme. However discussion in \cite{kud} does not make any claims about the optimality of this encoding scheme.

\begin{intexp}
Having a large missed detection exponent for the special message corresponds to having a large decoding region for the special message. This ensures that when $\mes=1$, i.e. when the special message is transmitted, probability  of $\hat{\mes} \neq 1 $ is exponentially small. In a sense  $\Emd$ indicates how large the decoding region of the special message could be made, while still filling  $\doteq e^{\BLX \CX}$ typical noise balls in the remaining space. The red region in Fig. \ref{fig:mdfa} denotes such a large region. Note that the actual decoding region of the special message is much larger  than this illustration, because it consists of all output types except the ones close to $P_Y^*$, whereas the ordinary decoding regions only contain the output types close to $P_Y^*$.
\begin{figure}[!ht]
\psfrag{exp(n R)}{}
\psfrag{exp(n C)}{}
\psfrag{\mes=1}{\!$\hat\mes=1$}
\centerline{\psfig{figure=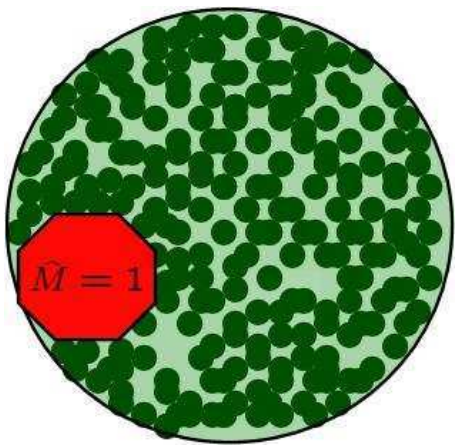,width=6cm}}
\caption{Avoiding missed-detection} 
\label{fig:mdfa} 
\end{figure}
\end{intexp}

Utility of this result is two folds: first, the optimality of such a simple scheme was not obvious before; second, as we will see later protecting a single special message is  a key building block for many other problems when feedback is available.

\subsection{Many special messages}
Now consider the case when  instead of a single special message, exponentially many of the  total $\doteq e^{\BLX \CX}$ messages are special. Let
${\cal M}_s^{(\BLX)}\subseteq{\cal M}^{(\BLX)}$ denote this set of
special messages,
\begin{equation*}
{\cal M}_s^{(\BLX)}=\{1,2,\cdots, \lceil e^{\BLX r}\rceil\}.
\end{equation*}
The best missed detection exponent, achievable simultaneously for all of the special messages, is denoted by $\Emdr(r)$.

\begin{definition}
For a capacity-achieving sequence $\SC$, the missed detection exponent achieved on sequence of subsets ${\cal M}_s$ is defined as 
\begin{equation*}
\Emdr_{,\SC, {\cal M}_s} \DEF \liminf_{\BLX \rightarrow \infty} \tfrac{{\displaystyle-\ln\max_{i\in{\cal M}_s^{(\BLX)}}} \SPX{\BLX}{\hat{\mes}\neq i|\mes=i}}{\displaystyle\BLX}.  
\end{equation*}
Then for a given $r<C$, $\Emdr(r)$ is defined as,
$\Emdr(r)\ \DEF\sup_{\SC, {\cal M}_s} \Emdr_{,\SC,{\cal M}_s}$ 
where maximization is over ${\cal M}_s$'s such that $ \displaystyle{\liminf_{\substack {\BLX\rightarrow\infty}}\frac{\ln|{\cal M}_s^{(\BLX)}|}{\BLX} = r}$.
\end{definition}
This message wise \uep~   problem has already been investigated by Csisz\'ar in his paper on joint source-channel coding \cite{csiszar1}. His analysis allows for multiple sets of special messages each with its own rate and an overall rate that can be smaller than the capacity.\footnote{Authors would like to thank Pulkit Grover of UC  Berkeley for pointing out this closely related work, \cite{csiszar1}}

Essentially, $\Emdr(r)$ is the best value for which missed detection probability of every special message is $\doteq \exp(-\BLX \Emdr(r))$ or smaller. Note that if the only messages in the code are these  $\lceil e^{\BLX r}\rceil $ special messages (instead of $|{\cal M}^{(\BLX)}|\doteq e^{\BLX \CX}$ total messages), their best missed detection exponent equals the classical error exponent $\Ee(r)$ discussed earlier.

\begin{theorem} \label{thm:mdmany}
\[\Emdr(r)=\Ee(r) \quad \forall\ r\in[0,\CX).\]
\end{theorem}

Thus we can communicate reliably at capacity and still protect the special messages as if we are only communicating the special messages. Note that the classical error exponent $\Ee(r)$ is yet unknown for the rates below critical rate (except zero rate). Nonetheless, this theorem says that whatever $\Ee(r)$ can be achieved for  $\lceil e^{\BLX r}\rceil $  messages when they are by themselves in the codebook, can still be achieved when there are $\doteq e^{\BLX \CX}$ additional ordinary messages requiring reliable communication.

  \begin{optstr}
     Start with an optimal code-book for $\lceil e^{\BLX r}\rceil $ messages which achieves the  error exponent $E(r)$. These codewords are used for the special messages. Now the ordinary codewords are added using random coding. The ordinary codewords which land close to a special codeword may be discarded  without  essentially any effect on the rate of communication. 

Decoder uses a two-stage decoding rule, in  first stage of which it  decides whether or not a special message was sent. If the received sequence is close to one or more of the special codewords, receiver decides that a special message was sent else it decides an ordinary message was sent. In the second stage, receiver employs an ML decoding either among the ordinary messages or the among the special messages depending on its decision in the first stage.

The overall missed detection exponent $\Emdr(r)$ is bottle-necked by the second stage errors. It is because the first stage error exponent is essentially the sphere-packing exponent $\Esp(r)$, which is never smaller than the second stage error exponent $E(r)$. 
  \end{optstr}

  \begin{intexp}
      This means that we can start with a code of  $ \lceil e^{\BLX r}\rceil$ messages, where the decoding regions are large enough to provide a missed detection exponent of $E(r)$. Consider the balls around each codeword with sphere-packing radius  (see Fig. \ref{fig:footballs}(a)). For each message, the probability  of going outside its ball decays exponentially with the sphere-packing exponent. Although, these $ \lceil e^{\BLX r}\rceil$ balls fill up most of the output space, there are still some cavities left between them. These small cavities can still accommodate $\doteq e^{\BLX \CX}$ typical noise balls for the ordinary messages (see Fig. \ref{fig:footballs}(b)), which are much smaller than the original $\lceil e^{\BLX r}\rceil$ balls. This is analogous to filling sand particles in a box full of large boulders. This theorem is like saying that the number of sand particles remains unaffected (in terms of the exponent) in spite of the large boulders.
  \end{intexp}

\begin{figure}[!h]
\begin{minipage}{8cm}
\psfrag{b1}{}\psfrag{exp(n R)}{} \psfrag{exp(n C)}{}
\centerline{\psfig{figure=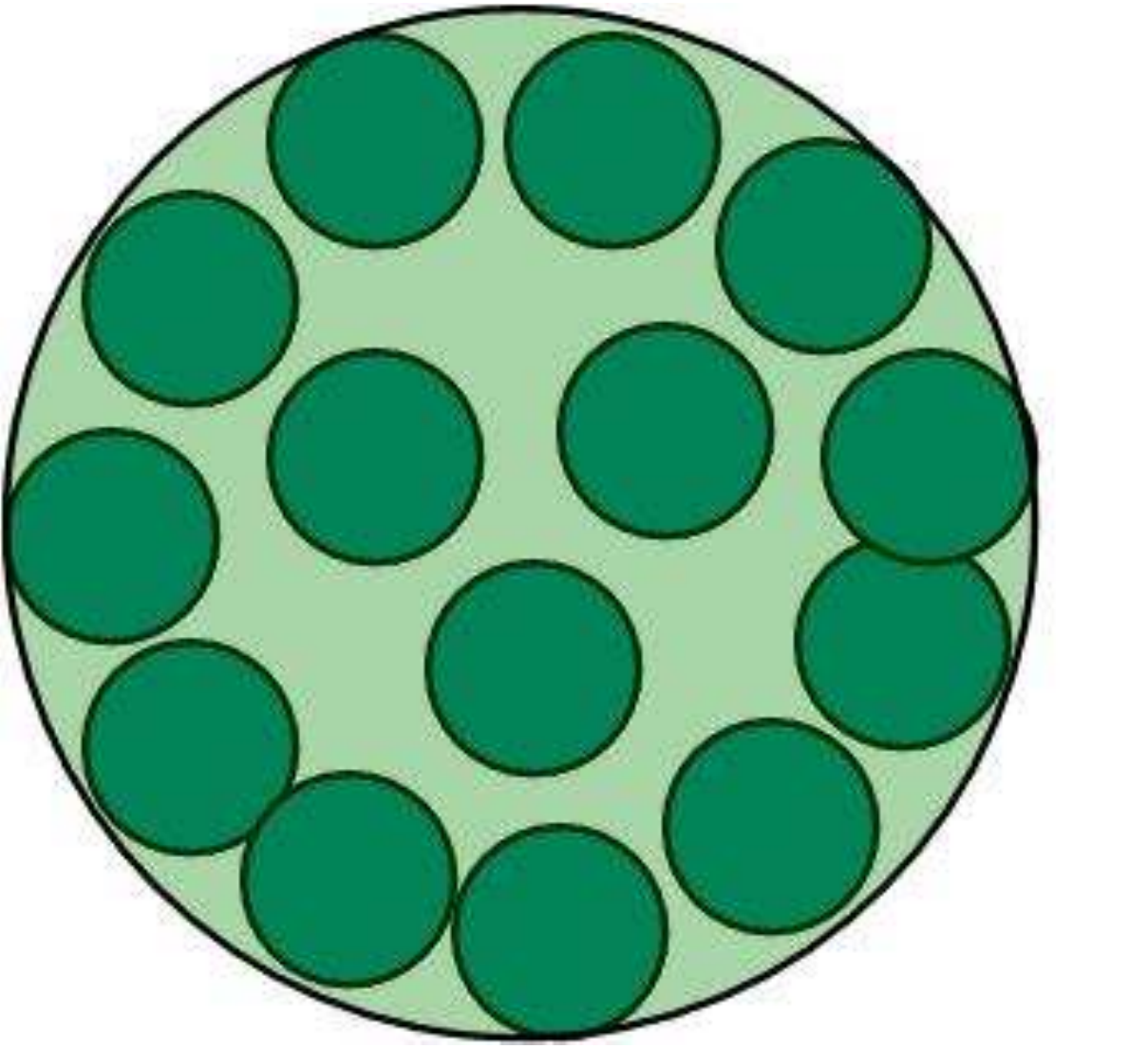,width=5.4cm}}
\centerline{(a) Exponent optimal code}
\end{minipage}
\begin{minipage}{8cm}
\psfrag{exp(n R)}{} \psfrag{exp(n C)}{} \vspace{.1cm}
\centerline{\psfig{figure=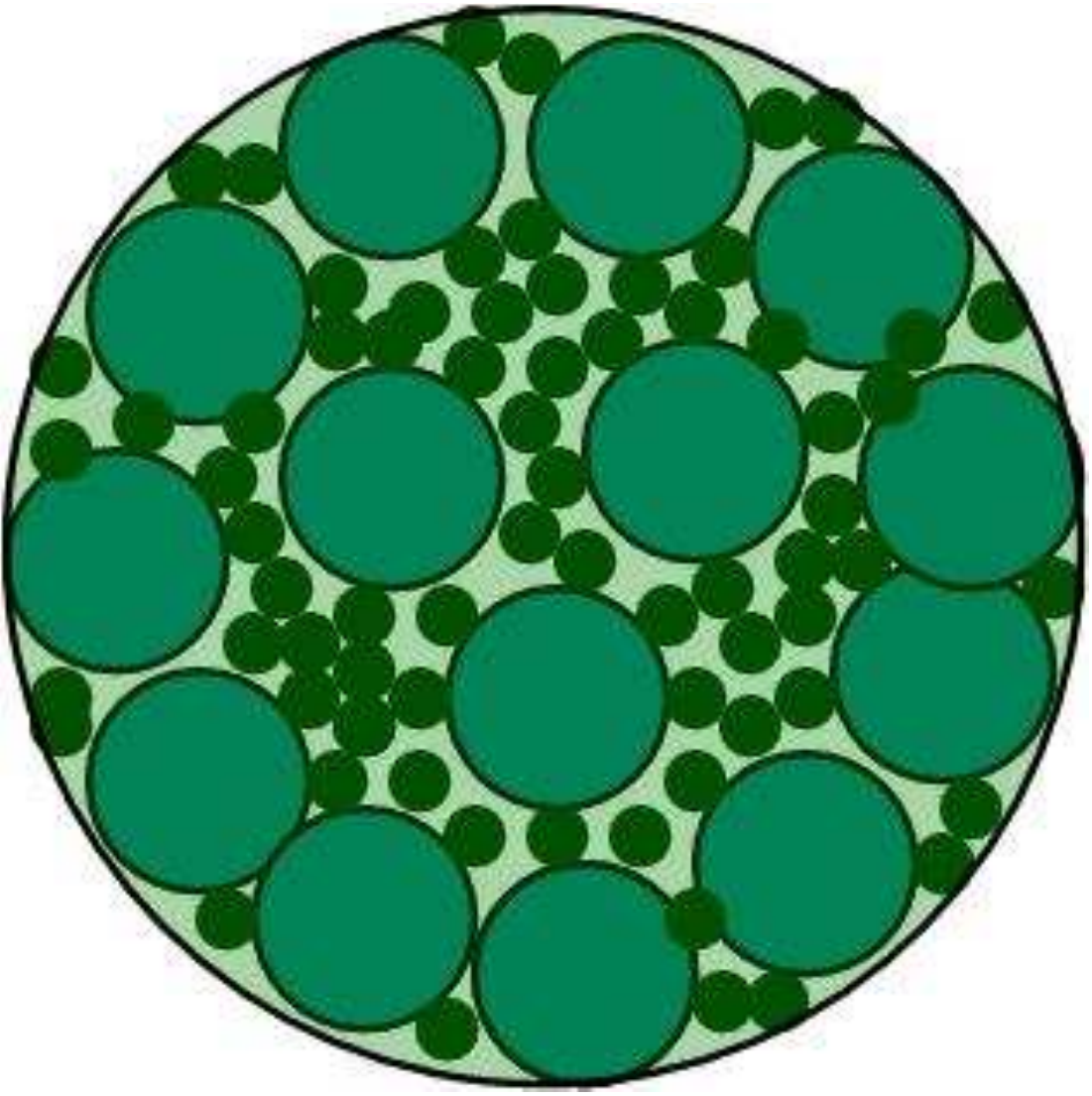,width=5cm}}
\centerline{(b) Achieving capacity}
\end{minipage}
\caption{``There is always room for capacity''}
\label{fig:footballs}
\end{figure}

\subsection{Allowing erasures}

In some situations, a decoder may be allowed declare an erasure when it is not sure about the transmitted message. These erasure events are not counted as errors and are usually followed by a retransmission using a decision feedback protocol like Hybrid-ARQ. This subsection extends the earlier result for $\Emd(r)$ to the cases  when such erasures are allowed.

In decoding with erasures, in addition to the message set ${\cal M}$, the decoder can map the received sequence $Y^{\BLX}$ to a virtual message called ``erasure''. Let $\Pera$ denote the average erasure probability of a code.
\begin{equation*}
  \Pera=\PX{\hat\mes=\era}
\end{equation*}
Previously when there was no erasures, errors were not detected. For errors and erasures decoding, erasures are detected errors, the  rest of the errors are undetected errors and  $\Pe$ denotes the undetected error probability. Thus average and conditional  (undetected) error probabilities are given by
\begin{equation*}
\Pe=\PX{\hat{\mes}\neq \mes, \hat{\mes}\neq \era} \quad\textrm{and}\quad \Pe(i)=  \PCX{\hat{\mes} \neq \mes,\hat{\mes}\neq\era}{\mes=i}
\end{equation*}
An infinite sequence $\SCe$ of block codes with errors and erasures decoding is  \emph{reliable}, if its average error probability and average erasure probability, both vanish with $\BLX$.
\begin{equation*}
 \lim_{\BLX \rightarrow \infty} \Pe^{(\BLX)}=0
 \qquad\textrm{and}\qquad
    \lim_{\BLX \rightarrow \infty} \Pera^{(\BLX)}=0
\end{equation*}
If the erasure probability is small, then average number of retransmissions needed is also small. Hence this condition of vanishingly small $\Pera^{(\BLX)}$ ensures that the effective data-rate of a decision feedback protocol remains unchanged in spite of retransmissions.  We again restrict ourselves to reliable sequences  whose rate equal $\CX$.

We could  redefine all previous exponents for decision-feedback (df) scenarios, i.e.  for reliable codes with erasure decoding. But resulting exponents do not change with  the provision of erasures with vanishing probability for single bit or single message problems, i.e.  decision  feedback protocols such as Hybrid-ARQ does not improve  $\Eb$ or $\Emd$. Thus we only discuss the decision feedback version of $\Emdr(r)$.

\begin{definition}
For a capacity-achieving sequence with erasures, $\SC$, the missed detection exponent achieved on sequence of subsets ${\cal M}_s$ is defined as 
\begin{equation*}
\Emdre_{,\SCe}(r) \DEF \liminf_{\BLX \rightarrow \infty} \tfrac{ {\displaystyle-\ln\max_{i\in{\cal M}_s^{(\BLX)}}} \SPX{\BLX}{\hat{\mes}\neq i,\hat\mes\neq \era|\mes=i}}{\displaystyle \BLX}.
\end{equation*}
Then for a given $r<C$,  $\Emdre_{,\SCe}(r)$ is defined as,
$\Emdre_{,\SCe}(r) \DEF\sup_{\SC, {\cal M}_s} \Emdr_{,\SC,{\cal M}_s}$ 
where maximization is over ${\cal M}_s$'s  such that $ \displaystyle{\liminf_{\substack {\BLX\rightarrow\infty}}\frac{\ln|{\cal M}_s^{(\BLX)}|}{\BLX} = r}$.
\end{definition}
Next theorem shows   allowing erasures increases the missed-detection exponent for $r$ below critical rate, on symmetric channels.

\begin{theorem}
For symmetric channels 
\[\Emdre(r)\ge E_{\textrm{sp}}(r)\quad \forall\ r\in[0,\CX).\]\label{thm:mde}
\end{theorem}

Coding strategy is similar to the no-erasure case. We first start  with an  erasure code for $\lceil e^{\BLX r}\rceil$ messages like the one in \cite{forney2}. Then add randomly generated ordinary codewords to it. Again a two-stage decoding is performed where the first stage decides between the set of ordinary codewords and the set of special codewords using a threshold distance. If this first stage chooses special codewords, the second stage applies the decoding rule in \cite{forney2} amongst special codewords. Otherwise, the second stage uses  the ML decoding among ordinary codewords.

The overall missed detection exponent $\Emdre(r)$ is bottle-necked by the first stage errors. It is because the first-stage error exponent $E_{\textrm{sp}}(r)$ is smaller than the second stage error exponent $E_{\textrm{sp}}(r)+\CX-r$. This is in contrast with the case without erasures.

\section{\uep~ at Capacity: Variable Length Block Codes with Feedback}
\label{sec:feedback} 
In the last section, we analyzed bit wise and message wise \uep~  problems for fixed length block codes (without feedback) operating at capacity. In this section, we will revisit the same problems for variable length block codes with perfect feedback, operating at capacity. Before going into the discussion of the problems, let us recall variable length block codes with feedback briefly.

A variable length block code with feedback, is composed of a coding algorithm and a decoding rule. Decoding rule determines the decoding time and the message that is decoded then.  Possible observations of the receiver can be seen as  leaves of $|{\cal Y}|$-ary tree,  as in \cite{burna2}. In this tree, all nodes at length $1$ from the root denote all $|{\cal Y}|$ possible outputs at time $t=1$. All non-leaf nodes among these split into further $|{\cal Y}|$ branches in the next time $t=2$ and the branching of the non-leaf nodes continue like this ever after. Each node of depth  $t$ in this tree corresponds to a  particular sequence, $y^t$, i.e. a history of outputs until time $t$.  The parent of node $y^t$ is its prefix $y^{t-1}$. Leaves of this tree form a prefix free source code, because decision to stop for decoding has to be a casual event. In other words the event $\{\blx=t \}$ should be measurable in the $\sigma$-field generated by $Y^{t}$. In addition we have $\PX{\tau < \infty}=1$ thus decoding time $\blx$ is  Markov stopping time with respect to receivers observation. The coding algorithm on the other hand assigns an input letter, $X_{t+1}(y^t;i)$, to each message, $i \in {\cal M}$, at each non-leaf node, $y^t$,  of this tree. The encoder stops transmission of a message when a leaf is  reached i.e. when  the decoding is complete.

 Codes we consider are block codes in the sense that transmission of each message (packet) starts only after the transmission of the previous one ends. The error probability and rate of the code are simply  given by
\begin{equation*}
  \Pe=\PX{\hat{\mes} \neq \mes}  \qquad\textrm{and},\qquad \RX=\tfrac{\ln \MX}{\EX{\blx}}
\end{equation*}
A more thorough discussion of variable length block codes with feedback can be found in \cite{burna} and \cite{burna2}.

Earlier discussion in Section \ref{sec:modele} about different kinds of errors is still valid as is but we need to slightly modify our discussion about the reliable sequences. A reliable sequence of variable length block codes with feedback, $\SCf$, is  any countably infinite collection of codes indexed by integers, such that
\begin{equation*}
  \lim_{\inx  \rightarrow \infty} \Pe^{(\inx)}=0
\end{equation*}
In the rate and exponent definitions for reliable sequences, we replace block-length $\BLX$ by the expected decoding time $\EX{\blx}$. Then a capacity achieving sequence with  feedback is a reliable sequence of variable length block codes with feedback whose rate is $\CX$

It is worth noting the importance of our assumption that all the entries of the transition probability matrix, $W_{Y|X}$ are positive. For any channel with a $W_{Y|X}$ which has one or more zero probability transitions, it is possible to have error free  codes operating at capacity, \cite{burna}. Thus all the exponents discussed below are infinite for DMCs with one or more zero probability transitions.

\subsection{Special  bit}
Let us consider a capacity achieving sequence $\SCf$  whose message sets are of the form ${\cal M}^{(\inx)}={\cal M}_1 \times {\cal M}_2^{(\inx)}$ where ${\cal M}_1=\{0, 1\}$. Then the error exponent of the $\mes_1$, \emph{i.e.}, the initial bit, is defined as follows. 

\begin{definition} 
For a capacity achieving sequence with feedback, $\SCf$, with message sets ${\cal M}^{(\inx)}$ of the form ${\cal M}^{(\inx)}={\cal M}_1 \times {\cal M}_2^{(\inx)} $ where ${\cal M}_1=\{0,1\}$, the special bit error exponent is defined as
\begin{equation*}
 \FEb_{,\SCf} \DEF \liminf_{\inx \rightarrow \infty} \tfrac{ -\ln  \SPX{\BLX}{\hat{\mes}_1 \neq \mes_1}}{\EX{\blx^{(\inx)}}}
\end{equation*}
Then $\FEb$ is defined as $\FEb \DEF \sup_{\SCf} \FEb_{,\SCf}$
\end{definition}

\begin{theorem}\label{thm:bitf}
\[\FEb=\PCAP.\]
\end{theorem}
Recall that without feedback, even a single bit could not achieve any positive error exponent at capacity, Theorem \ref{thm:bit}. But feedback together with variable decoding time connects the message wise \uep~ and the bit wise \uep~ and results in a positive exponent for bit wise \uep. Below described strategy show how schemes for protecting a special message can be used to protect a special bit.

\begin{optstr}
  We use  a length  $(\inx+\sqrt{\inx})$ fixed length block code with errors and erasures decoding as a building block for our code.  Transmitter first transmits $\mes_1$ using a short repetition code of length $\sqrt{\inx}$. If the tentative decision about $\mes_1$, $\tilde{\mes}_1$, is correct after this repetition code, transmitter  sends   $\mes_2$  with a length $\inx$  capacity achieving code. If $\tilde{\mes}_1$ is incorrect after the repetition code, transmitter  sends the symbol $\xr$ for $\inx$ time units where $\xr$ is the input letter $i$ maximizing the $\KLD{P_{Y}^{*}(\cdot)}{W_{Y|X}(\cdot|i)}$. If the output sequence in the second phase, $Y_{\sqrt{\inx}+1}^{\sqrt{\inx}+\inx}$, is not a typical sequence of  $P_{Y}^{*}$, an erasure is declared for the block. And the same message is retransmitted by repeating the same strategy afresh. Else receiver  uses an ML decoder to chose $\hat{\mes}_2$ and $\hat{\mes}=(\tilde{\mes}_1, \hat{\mes}_2)$.

 The erasure probability is vanishingly small, as a result the undetected error probability of $\mes_i$ in fixed length erasure code is approximately equal to the  error probability of $\mes_i$ in the variable length block code. Furthermore $\EX{\blx}$ is roughly  $(\inx+\sqrt{\inx})$ despite the retransmissions. A decoding error for $\mes_1$ happens only when $\tilde{\mes}_1 \neq \mes_1$ and the empirical distribution of the output sequence in the second phase is close to $P_{Y}^*$. Note that latter event happens with probability  $\doteq e^{-\PCAP \EX{\blx}}$.
\end{optstr}

\subsection{Many special bits}
We now analyze the situation where instead of a single special bit, there are approximately $\EX{\blx}r/\ln 2$ special bits out of the total $\EX{\blx} \CX / \ln 2$ (approx.)  bits. Hence we consider the capacity achieving sequences with feedback having message sets  of the form ${\cal M}^{(\inx)}={\cal M}_1^{(\inx)} \times {\cal M}_2^{(\inx)}$. Unlike the previous subsection where size of ${\cal M}_1^{(\inx)}$ was fixed, we now allow its size to vary with the index of the code. We restrict ourselves to the cases where $\displaystyle{ \liminf_{\inx \rightarrow \infty } \tfrac{\ln |{\cal M}_1^{(\inx)}|}{\EX{\tau^{(\inx)}}}=r}$. This limit gives us the rate of the special bits. It is worth noting at this point that even when the rate  $r$ of special bits is zero, the number of special bits might not be bounded, i.e. $\displaystyle{\liminf_{ \inx \rightarrow \infty} |{\cal M}_1^{(\inx)}|}$ might be infinite.  The error exponent  $\FEbr_{,\SCf}$ at a given rate $r$ of special bits is  defined  as follows,

\begin{definition}
For any capacity achieving sequence with feedback $\SCf$ with the message sets ${\cal M}^{(\inx)}$ of the form 
${\cal M}^{(\inx)}={\cal M}_1^{(\inx)} \times {\cal M}_2^{(\inx)} $,  $r_{\SCf}$ and  $\FEbr_{,\SCf}$ are defined as
\begin{equation*}
r_{\SCf}\DEF \liminf_{\inx \rightarrow \infty }  \tfrac{\ln |{\cal M}_1^{(\inx)}|}{\EX{\tau^{(\inx)}}}
\qquad
 \FEbr_{,\SCf} \DEF \liminf_{\inx \rightarrow \infty} \tfrac{ -\ln \SPX{\inx}{\hat{\mes}_1\neq \mes_1}}{\EX{\blx^{(\inx)}}}
\end{equation*}
Then $\FEbr(r)$ is defined as $\FEbr(r) \DEF {\displaystyle\sup_{ \SCf: r_{ {\tiny \SCf}}\geq r } \FEbr_{,\SCf}}$
\end{definition}

Next theorem shows how this exponent decays linearly with rate $r$ of the special bits.
\begin{theorem} \label{thm:l1}
 \begin{equation*}
   \FEbr(r) =\left(1-\tfrac{r}{\CX}\right) \PCAP
   \end{equation*}
 \label{thm:rf}
\end{theorem}
\vspace{-.2cm}
Notice that the exponent $\FEbr(0) = \PCAP$, i.e. it is as high as the exponent in the single bit case, in spite of the fact that here the number of bits can be growing to infinity with $\EX{\tau}$. This linear trade off between rate and reliability reminds us of Burnashev's result \cite{burna}.

\begin{optstr}
 Like the single bit case, we  use a fixed length block code with  erasures as our building block. First transmitter sends $\mes_1$ using a capacity achieving code of length $\tfrac{r}{\CX}\inx$. If the tentative decision $\tilde{\mes}_1$ is correct, transmitter sends $\mes_2$
with a capacity achieving code of length $ (1-\tfrac{r}{\CX})\inx$.  Otherwise transmitter sends the channel input $\xr$ for $(1-\frac{r}{\CX})\inx$ time units. If the output sequence in the second phase is not typical with $P_Y^{*}$ an erasure is declared and same strategy is  repeated afresh. Else receiver uses a ML decoder to decide $\hat{\mes}_2$ and decodes the message $\mes$ as $\hat{\mes}= (\tilde{\mes}_1,\hat{\mes}_2)$.   A decoding error for $\mes_1$ happens only when an error happens in the first phase and the output sequence in the second phase is typical with $P_Y^{*}$ when the reject codeword is sent. But the probability of the later event is $\doteq e^{- (1-\tfrac{r}{\CX}) \PCAP \inx}$. The factor of $(1-\tfrac{r}{\CX})$ arises because the relative duration of the second phase to the over all communication block. Similar to the single bit case, erasure probability remains vanishingly small in this case.  Thus not only the expected decoding time of the variable length block code is roughly equal to the block length of the fixed length block code, but also its error probabilities are roughly equal to the corresponding error probabilities associated with the fixed length block code. 
\end{optstr}

\subsection{Multiple layers of priority}
We can generalize this result to the case when there are multiple levels of priority, where the most important layer contains $\EX{\blx}r_1/\ln 2$ bits, the second-most important layer contains $\EX{\blx}r_2/\ln 2$ bits and so on.  For an $L$-layer situation, message set ${\cal M}^{(\inx)}$ is of the form ${\cal M}^{(\inx)}={\cal M}_1^{(\inx)} \times {\cal M}_2^{(\inx)} \times \cdots \times {\cal M}_L^{(\inx)} $. We  assume without loss of generality that  the order of importance of the $\mes_i$'s is $\mes_1 \succ \mes_2 \succ \cdots \succ \mes_L$. Hence we have $\Pe^{\mes_1} \leq \Pe^{\mes_2}\leq \cdots \leq \Pe^{\mes_L}$.

Then for any $L$-layer capacity achieving sequence with feedback, we define the error exponent of the $s^{\mbox{th}}$ layer as
\begin{equation*}
\FEbr_{,s,\SCf} = \liminf_{\inx \rightarrow \infty} \tfrac{- \ln \SPX{\inx}{\hat{\mes}_s \neq \mes_s}}{\EX{\blx^{(\inx)}} }.
\end{equation*}
The achievable error exponent region of the $L$-layered capacity achieving sequences with feedback is the set of all achievable exponent vectors $(\FEbr_{,1,\SCf}, \FEbr_{,2,\SCf},\ldots,\FEbr_{,L-1,\SCf})$. The following theorem determines that region. \vspace{.1cm}

\begin{theorem} \label{thm:many}
Achievable error exponent region of the $L$-layered capacity achieving sequences with feedback, for rate vector $(r_1,r_2,\ldots,r_{L-1})$ is the set of  vectors $(E_1,E_2,\ldots,E_{L-1})$ satisfying,
\begin{equation*}
  E_i \leq   \left(1-\frac{\sum_{j=1}^i r_j}{\CX} \right) \PCAP  \qquad  \forall i \in \{1,2,\ldots, (L-1)\}.
\end{equation*}
\end{theorem}
Note that the least important layer cannot achieve any positive  error exponent because we are communicating at capacity, i.e. $E_L=0$. 

\begin{optstr}
Transmitter first sends  the most important layer, $\mes_1$, using a capacity achieving code of length $\tfrac{r_1}{\CX} \inx$. If it is decoded correctly, then it sends the next layer with a capacity achieving code of length $\tfrac{r_2}{\CX} \inx$. Else it starts sending the input letter $\xr$ for not only $\tfrac{r_2}{\CX} \inx$ time units but also for all remaining $L-2$ phases. Same strategy is repeated  for $\mes_3, \mes_4,\ldots,\mes_L$. 

Once the whole block of channel outputs,  $Y^{\inx}$,  is observed; receivers checks the empirical distribution of the output in all of the phases except the first one. If they are all typical with $P_{Y}^{*}$ receiver uses the tentative decisions to decode,  $\hat{\mes}=(\tilde{\mes}_1, \tilde{\mes}_2, \ldots \tilde{\mes}_{L})$. If one or more of the output sequences are not typical with $P_{Y}^{*}$ an erasure is declared for the whole block and transmission starts from scratch. 
\end{optstr}

For each layer $i$, with the above strategy we can achieve an exponent as if there were only two kinds of bits (as in Theorem \ref{thm:l1})
\begin{itemize}
\item bits in layer $i$ or in  more important layers $k<i$ (i.e. special bits)
\item  bits in less important layers (i.e. ordinary bits).
\end{itemize}
Hence Theorem \ref{thm:many} does not only specify  the optimal performance when there are multiple layers, but also shows that the performance we observed in Theorem  \ref{thm:l1}, is  successively refinable. Figure \ref{fig:layers} shows these simultaneously achievable exponents of Theorem \ref{thm:rf},  for a particular rate vector $(r_1,r_2,\ldots,r_{L-1})$.

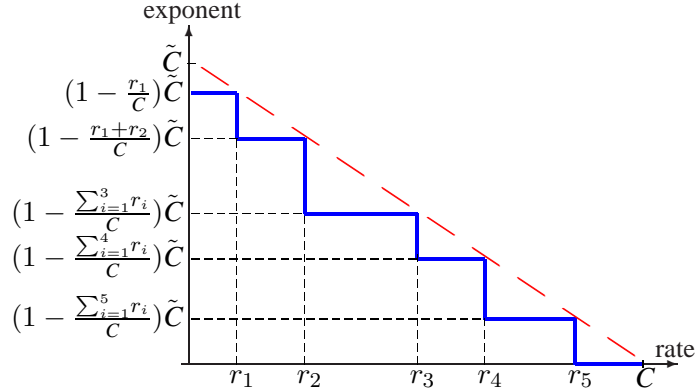
\begin{figure}[!h]
   \setlength{\unitlength}{1cm}
   \centering
   \begin{picture}(7.5,5)(-1.5,-.5)
   \put(0.1,0){\vector(0,1){4.5}}
   \put(0,0){\vector(1,0){6.6}}
   \put(-.5,4.6){\small exponent}
   \put(6.3,.1){\small rate}
   \linethickness{0.35mm}
   \color{blue}
   \put(0.0,3.6){\line(1,0){0.6}}
   \put(0.6,3.6){\line(0,-1){.6}}
   \put(0.6,3.0){\line(1,0){0.9}}
   \put(1.5,3.0){\line(0,-1){1}}
   \put(1.5,2.0){\line(1,0){1.5}}
   \put(3.0,2.0){\line(0,-1){.6}}
   \put(3.0,1.4){\line(1,0){0.9}}
   \put(3.9,1.4){\line(0,-1){.8}}
   \put(3.9,0.6){\line(1,0){1.2}}
   \put(5.1,0.6){\line(0,-1){.6}}
   \put(5.1,0.0){\line(1,0){.9}}
   \color{black}
   \linethickness{0.025mm}
   \multiput(0,3.0)(.17,0){4}{\line(1,0){.1}}
   \multiput(0,2.0)(.17,0){9}{\line(1,0){.1}}
   \multiput(0,1.4)(.17,0){18}{\line(1,0){.1}}
   \multiput(0,0.6)(.17,0){23}{\line(1,0){.1}}
   \multiput(0.6,0)(0,.17){18}{\line(0,1){.1}}
   \multiput(1.5,0)(0,.17){12}{\line(0,1){.1}}
   \multiput(3.0,0)(0,.17){9}{\line(0,1){.1}}
   \multiput(3.9,0)(0,.17){4}{\line(0,1){.1}}

   \linethickness{0.06mm}
   \put(0.5,-.25){$r_1$}
   \put(1.4,-.25){$r_2$}
   \put(2.9,-.25){$r_3$}
   \put(3.8,-.25){$r_4$}
   \put(5.0,-.25){$r_5$}
   \put(5.9,-.30){$\CX$}
   \put(6.0,-0.05){\line(0,1){.1}}

   \put(-0.05,4.0){\line(1,0){.1}}
   \put(-.4,3.9){$\PCAP$} 
   \put(-1.7,3.5){$(1-\tfrac{r_1}{\CX})\PCAP$} 
   \put(-2.2,2.9){$(1-\tfrac{r_1+r_2}{\CX})\PCAP$} 
   \put(-2.4,1.9){$(1-\tfrac{\sum_{i=1}^{3}\!r_i}{\CX})\PCAP$} 
   \put(-2.4,1.3){$(1-\tfrac{\sum_{i=1}^{4}\!r_i}{\CX})\PCAP$} 
   \put(-2.4,0.5){$(1-\tfrac{\sum_{i=1}^{5}\!r_i}{\CX})\PCAP$} 

   \linethickness{0.2mm}
    \color{red}
   \multiput(0,3.95)(0.6,-.4){10}{\line(3,-2){.36}}
   \end{picture}
\caption{Successive refinability for multiple layers of priority, demonstrated on an example with six layers; $\sum_{i=1}^{6}r_i=\CX$.}
\label{fig:layers} 
   \end{figure}

Note that the most important layer can achieve an exponent close to $\PCAP$ if its rate is close to zero. As we move to the layers with
decreasing importance, the achievable error exponent decays gradually.

\subsection{Special message}
Now consider one particular message, say the first one,  which requires small missed-detection probability. Similar to the no-feedback case, define $\FEmd$  as its missed-detection exponent at capacity.

\begin{definition}
  For any capacity achieving  sequence with feedback, $\SCf$, missed detection exponent is defined as
\begin{equation*}
  \FEmd_{,\SCf} \DEF\liminf_{\inx \rightarrow \infty}\tfrac{- \ln \SPCX{\inx}{\hat{\mes} \neq 1}{\mes =1}}{\EX{\tau^{(\inx)}}}.
\end{equation*}
Then $\FEmd$ is defined as $\FEmd  \DEF \sup_{\SCf} \FEmd_{,\SCf}$.
\end{definition}

\begin{theorem}
\[\FEmd=\PCAP.\] \label{thm:mdf}
\end{theorem}
Theorem \ref{thm:md} and \ref{thm:mdf} implies following corollary,
\begin{corollary}
Feedback doesn't improve the missed detection  exponent of a single special message: $\FEmd=\Emd$.  
\end{corollary}
If $\psed$ were defined as the best protection of a special message achievable at capacity, then this result could have been thought of as an analog the ``feedback does not increase capacity'' for the $\psed$. Also note that with feedback, $\FEmd$ for the special message and $\FEb$ for the special bit are equal.

\subsection{Many special messages}
\label{subsec:mdr} 
Now let us consider the problem where the first $\lceil e^{\EX{\tau} r} \rceil$ messages are special, i.e. ${\cal M}_s=\{1,2,\ldots,\lceil e^{\EX{\tau} r} \rceil \}$. Unlike previous problems, now we will also impose a uniform expected delay constraint as follows.

\begin{definition}
For any reliable variable length block code with feedback,
\begin{equation*}
  \Gamma \DEF \tfrac{\max_{i \in {\cal M}}  \ECX{\tau}{\mes=i}}{\EX{\tau}}
\end{equation*}
A reliable sequence with feedback, $\SCf$, is  a uniform delay reliable sequence with feedback if and only if 
$\displaystyle{\lim_{\inx\rightarrow \infty} \Gamma^{(\inx)} =1}$.
\end{definition}

This means that the average $\ECX{\tau}{\mes=i}$  for every message $i$ is essentially equal to $\EX{\tau}$ (if not smaller). This uniformity constraint reflects a system requirement for ensuring a robust delay performance, which is invariant of the transmitted message.\footnote{Optimal exponents in all previous problems  remain unchanged irrespective of this uniform delay constraint.} Let us define the missed-detection exponent $\FEmd (r)$ under this uniform delay constraint.

\begin{definition}
  For any uniform delay capacity achieving  sequence with feedback, $\SCf$,  the missed detection exponent achieved on sequence of subsets ${\cal M}_s$ is defined as 
\begin{equation*}
\FEmd_{,\SC, {\cal M}_s} \DEF \liminf_{\BLX \rightarrow \infty} \tfrac{ {\displaystyle-\ln\max_{i\in{\cal M}_s^{(\inx)}}} \SPX{\inx}{\hat{\mes}\neq i|\mes=i}}{\displaystyle{\EX{\tau^{(\inx)}}}}.
 \end{equation*}
Then for a given $r<C$, we define  $\FEmd(r) \DEF\sup_{\SC, {\cal M}_s} \FEmd_{,\SC,{\cal M}_s}$ where maximization is over ${\cal M}_s$'s such that $\displaystyle{\liminf_{\substack {\inx \rightarrow \infty}}\frac{\ln|{\cal M}_s^{(\inx)}|}{\EX{\blx^{(\inx)}}}=r}$.
\end{definition}

The following theorem shows that the special messages can achieve the minimum of the $\psed$ and the Burnashev's exponent at rate $r$.
\begin{theorem}
\label{thm:msru}
  \begin{equation*}
\FEmd (r) = \min \lcb \PCAP, (1-\tfrac{r}{\CX}) \DX \rcb , \quad \forall\ r<\CX.
  \end{equation*}
where $\displaystyle{\DX \DEF \max_{i,j \in {\cal X}} \KLD{W_{Y|X}(\cdot|i)}{W_{Y|X}(\cdot|j)}}$.
\end{theorem}
For $r \in [0, (1-\tfrac{\PCAP}{\DX})\CX]$ each special message achieves the  best missed detection exponent $\PCAP$ for a single special message, as if the rest of the special messages were absent. For $r \in [(1-\tfrac{\PCAP}{\DX})\CX, \CX)$  special messages achieve the Burnashev's exponent as if the ordinary messages were absent.

The optimal strategy is based on transmitting a special bit first. This result demonstrates, yet another time,  how feedback connects bit-wise \uep~  with message-wise \uep. In the optimal strategy for bit-wise \uep~  with many bits a special message was used, whereas now in message wise \uep~  with many messages a special bit is used. The roles of bits and messages, in two optimal strategies are simply  swapped between the two cases.

\begin{optstr}    
We combine the strategy for achieving $\PCAP~$ for a special bit and the Yamamoto-Itoh strategy for achieving Burnashev's exponent \cite{itoh}. In the first phase, a special bit, $b$, is sent with a repetition code of $\sqrt{\inx}$ symbols. This is the indicator  bit for special messages: it is $1$ when a special message is to be sent and $0$ otherwise.

If $b$ is decoded incorrectly as $\hat b=0$, input letter $x_r$ is sent for the remaining $\inx$ time unit.  If it is decoded correctly as $\hat b=0$, then the ordinary message is sent using a codeword from a capacity achieving code. If the output sequence in the second phase is typical with $P_Y^{*}$ receiver  use an ML decoder to chose one of the ordinary messages,  else an erasure is declared for $(\inx+\sqrt{\inx})$ long block.

If $\hat{b}=1$, then a length $\inx$ two phase code with errors and erasure decoding, like the one given in \cite{itoh}  by  Yamamoto and Itoh, is used to send the message. In the communication phase  a length $\tfrac{r}{\CX}\inx$ capacity achieving code is used to send the message, $\mes$, if $\mes \in {\cal M}_s$.  If $\mes \notin {\cal M}_s$ an arbitrary codeword from the length $\tfrac{r}{\CX}\inx$  capacity achieving code is sent. In the control phase, if $\mes \in {\cal M}_s$ and if it is decoded correctly at the end of communication phase, the accept letter $\xa$ is sent for $(1-\tfrac{r}{\CX})\inx$  time units, else the  reject letter, $\xd$, is sent for $(1 -\tfrac{r}{\CX})\inx$ time units. If the empirical distribution  in the control phase is typical with $W_{Y|X}(\cdot|\xa)$ then special message decoded at the end of the communication phase becomes the final $\hat{\mes}$, else an erasure is  declared for $(\inx+\sqrt{\inx})$ long block.

Whenever an erasure is declared for the whole block, transmitter and receiver applies above strategy  again from  scratch. This scheme is repeated  until a non-erasure decoding is reached.
\end{optstr}

\section{Avoiding False Alarms}\label{sec:fa}
In the previous sections while investigating message wise \uep~  we have only considered the missed detection formulation of the problems. In this section we will focus on an alternative formulation of message wise \uep~  problems based on false alarm probabilities. 
\subsection{Block Codes without Feedback}
We first consider the no-feedback case. When false-alarm of a special message is a  critical event, e.g. the ``reboot'' instruction,  the false alarm probability  $\Pr\!\left[\hat \mes= 1|\mes\neq 1\right]$ for this message should be minimized, rather than the missed detection probability  $\Pr\!\left[\hat{\mes}\neq  1|\mes= 1\right]$.

 Using Bayes' rule and assuming uniformly chosen messages we get,
\begin{align*}
  \PX{\hat \mes= 1|\mes\neq 1}
&=\frac{\PX{\hat \mes= 1,\mes\neq 1}}{\PX{\mes\neq 1}}\\
&=\frac{\sum_{j\neq 1}\PX{\hat \mes= 1|\mes= j}}{( |{\cal M}|-1)}.
\end{align*}
In classical error exponent analysis, \cite{gallager_book}, the error probability for a given message usually means its missed detection probability. However, examples such as the ``reboot'' message necessitate this notion of  false alarm probability.
\begin{definition}
 For a capacity-achieving sequence, $\SC$, such that
\begin{equation*}
 \limsup_{\BLX \rightarrow \infty}   \SPCX{\BLX}{\hat \mes \neq  1}{\mes= 1}=0,
 \end{equation*}
 false alarm exponent is defined as  
\begin{equation*}
\notag \Efa_{,\SC} \DEF \liminf_{\BLX \rightarrow \infty} \tfrac{ -\ln \SPCX{\BLX}{\hat \mes= 1}{\mes\neq 1}}{\BLX}.
\end{equation*}
Then $\Efa$ is defined as $\Efa \DEF\sup_{\SC} \Efa_{,\SC}$.  
\end{definition}

Thus  $\Efa$ is the best exponential decay rate of false alarm probability with $\BLX$. Unfortunately we do not have the exact expression for $\Efa$. However upper bound given below  is sufficient to demonstrate the improvement introduced by feedback and variable decoding time.
\begin{theorem}
\label{thm:fa}
\begin{equation*}
 \Efal \leq \Efa\leq \Efau.
\end{equation*}
The upper and lower bounds to the false alarm exponent are given by
\begin{align*}
  \Efal & \DEF \max_{i \in {\cal X}}  \min_{\substack{ V_{Y|X}:\\ \sum_{j} V_{Y|X}(\cdot|j) P_X^{*}(j) =  W_{Y|X}(\cdot|i)}}  \CKLD{V_{Y|X}(\cdot|X)}{W_{Y|X}(\cdot|X)}{P_X^{*}}\\
  \Efau & \DEF \max_{i \in {\cal X}} \CKLD{W_{Y|X}(\cdot|i)}{W_{Y|X}(\cdot|X)}{P_X^{*}}.
\end{align*}
\end{theorem}
The maximizers of the optimizations for $\Efal$ and $\Efau$ are denoted by  $\xfl$ and $\xfu$
\begin{align*}
\Efal
&=  \min_{\substack{ V_{Y|X}:\\ \sum_{j} V_{Y|X}(\cdot|j) P_X^{*}(j) =  W_{Y|X}(\cdot|\xfl)}}  
\CKLD{V_{Y|X}(\cdot|X)}{W_{Y|X}(\cdot|X)}{P_X^{*}}\\
\Efau
&=   \CKLD{V_{Y|X}(\cdot|\xfu)}{W_{Y|X}(\cdot|X)}{P_X^{*}}.  
\end{align*}

\begin{strl}  
Codeword for the special message $\mes=1$ is a repetition sequence of input letter $\xfl$. Its decoding region is the typical `noise ball' around it, the output sequences whose empirical distribution is approximately equal to  $W_{Y|X}(\cdot|\xfl)$. For the ordinary messages, we use a capacity achieving code-book where all codewords have the same empirical distribution (approx.) $P_X^*$. Then for $y^{\BLX}$ whose empirical distribution is not in the typical `noise ball' around the special codeword,  receiver makes an ML decoding among the  ordinary codewords.
\end{strl}
Note the contrast between this strategy for achieving $\Efal$ and the optimal strategy for achieving $\Emd$. For achieving $\Emd$, output sequences of any  type other than the ones close to $P_Y^*$ were decoded as the special message; whereas for achieving $\Efa$, only the output sequences of types that are close to $W_{Y|X}(\cdot|\xfl)$ are decoded as the special  message. 
\begin{figure}[!h]
\psfrag{exp(nR)}{} \psfrag{exp(n C)}{} \psfrag{\mes=1}{\!$\hat{\mes} =1$}
\centerline{\psfig{figure=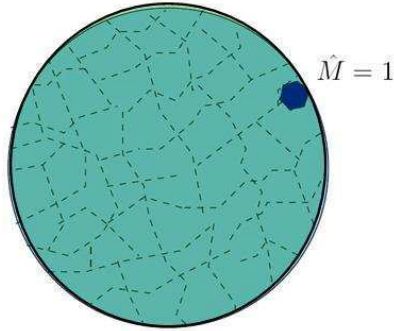,width=6cm}} 
\caption{Avoiding false-alarm} 
\label{fig:fa} 
\vspace{-0cm}
\end{figure}

\begin{intexp}   
A false alarm exponent for the special message corresponds to having the smallest possible decoding region for the special  message. This ensures that when some ordinary message is transmitted, probability of the event $\{\hat{\mes}=1\}$ is exponentially small. We cannot make it too small though, because when the special message is transmitted, the probability of the very same event should be almost one. Hence the decoding region of the special message  should at least contain the typical noise ball around the special codeword. The blue region in Fig. \ref{fig:fa} denotes such a region.
\end{intexp}

Note that $\Efal$ is larger than channel capacity $\CX$ due to the  convexity of KL divergence.
\begin{align*}
 \Efal
&= \max_{i \in {\cal X}} \min_{\substack{ V_{Y|X}:\\ \sum_{j} V_{Y|X}(\cdot|j) P_X^{*}(j) =  W_{Y|X}(\cdot|i)}} \CKLD{V_{Y|X}(\cdot|X)}{W_{Y|X}(\cdot|X)}{P_X^{*}}\\
&> \max_{i \in {\cal X}}  \min_{\substack{ V_{Y|X}:\\ \sum_{j} V_{Y|X}(\cdot|j) P_X^{*}(j) =  W_{Y|X}(\cdot|i)}}  \KLD{\sum_{k} P_X^{*}(k) V_{Y|X}(\cdot|k)}{\sum_{k'} P_X^{*}(k') W_{Y|X}(\cdot|k')}\\
&= \max_{i \in {\cal X}}     \KLD{ W_{Y|X}(\cdot|i)}{  P_Y^{*}(\cdot)}\\
&= \CX
\end{align*}
where   $P_Y^{*}$ denotes the output distribution corresponding to the capacity achieving input distribution $P_X^{*}$ and  the last equality follows from KKT condition for achieving capacity we mentioned previously  \cite[Theorem 4.5.1]{gallager_book}.

Now we can compare our result for a special message with the similar result for classical situation where all messages are treated equally. It turns out that if every message in a capacity-achieving code demands equally good  false-alarm exponent, then this uniform exponent cannot be larger than $\CX$.  This result seems to be  directly connected with the problem of identification via channels \cite{channelid}. We can prove the achievability part of their capacity theorem using an extension of the achievability part of $\Efal$. Perhaps a new converse of their result is also possible using such results. Furthermore we see that reducing the demand of false-alarm exponent to only one message, instead of all, enhances it from $\CX$ to at least $\Efal$.

\subsection{Variable Length Block Codes with Feedback}
Recall that feedback does not improve the missed-detection exponent for a special message. On the contrary, the false-alarm exponent of a special message is improved  when feedback is available and variable decoding time is allowed. We again restrict to uniform delay capacity achieving sequences with feedback, i.e. capacity achieving sequences satisfying ${\displaystyle \lim_{\inx\rightarrow \infty} \Gamma^{(\inx)} =1}$.

\begin{definition}
   For a uniform delay capacity-achieving sequence with feedback, $\SCf$, such that
\begin{equation*}
 \limsup_{\inx \rightarrow \infty}   \SPCX{\inx}{\hat \mes \neq  1}{\mes= 1}=0,
 \end{equation*}
 false alarm exponent is defined as  
\begin{equation}
\notag \FEfa_{,\SCf} \DEF \liminf_{\inx \rightarrow \infty} \tfrac{ -\ln \SPCX{\inx}{\hat \mes= 1}{\mes\neq 1}}{\EX{\blx^{\inx}}}.
\end{equation}
Then $\FEfa$ is defined as $\FEfa \DEF\sup_{\SCf} \FEfa_{,\SCf}$.
\end{definition}

\begin{theorem}
\[\FEfa=\DX.\] \label{thm:faf}
\end{theorem}
Note that $\DX > \Efau$. Thus  feedback strictly improves the false alarm exponent, $\FEfa > \Efa$.

\begin{optstr}
We  use a  strategy similar to the one employed  in proving Theorem \ref{thm:msru} in  subsection \ref{subsec:mdr}. In the first phase, a length  $\sqrt{\inx}$ code is used to convey whether $\mes=1$ or not, using a special bit $b=\IND{\mes=1}$.
  \begin{itemize}
  \item  If $\hat{b}=0$, a length $\inx$ capacity achieving code with $\Emd=\PCAP$ is used. If the decoded message for the length $\inx$ code is $1$, an erasure is declared for $(\inx+\sqrt{\inx})$ long block. Else the decoded message of length $\inx$ code becomes the decoded message for the whole $(\inx+\sqrt{\inx})$ long block.
  \item  If $\hat{b}=1$,
    \begin{itemize}
    \item and $\mes=1$, input symbol $\xa$ is transmitted for $\inx$ time units.
    \item and $\mes\neq1$, input symbol $\xd$ is transmitted for $\inx$ time units.
    \end{itemize}
  If the output sequence, $Y_{\sqrt{\inx}+1}^{\sqrt{\inx}+\inx}$,  is typical with $W_{Y|X}(\cdot|\xa)$ then $\hat{\mes}=1$ else an erasure is declared for $(\inx+\sqrt{\inx})$ long block.
  \end{itemize}
Receiver and transmitter starts from scratch if an erasure is declared at the end of second phase. 
\end{optstr}

Note that, this strategy simultaneously achieves the optimal missed-detection exponent $\PCAP$ and the optimal false-alarm exponent $\DX$ for this special message.

\section{Future directions}
\label{sec:summary}

In this paper we have restricted our investigation of \uep~ problems to data rates that are essentially equal to the channel capacity. Scenarios we have analyzed provides us with a rich class of problems when we consider data rates below capacity. 

Most of the \uep~ problems has a coding theoretic version. In these coding theoretic versions deterministic guarantees, in terms of Hamming distances, are demanded instead of the probabilistic guarantees, in terms of error exponents.  As we have mentioned in section \ref{sec:prev}, coding theoretic versions of bit-wise \uep~ problems have been studied for the case of linear codes extensively. But it seems coding theoretic versions of both message-wise \uep~ problems and bit-wise \uep~ problem for non-linear codes are scarcely investigated \cite{bas}, \cite{bd09}.

Throughout this paper, we focused on the channel coding component of communication. However, often times, the final objective is to communicate a source within some distortion constraint. Message-wise \uep~ problem itself has first come up within this framework \cite{csiszar1}.  But the source we are trying to convey can itself be heterogeneous, in the sense that  some part of its output may demand a smaller distortion than other parts. Understanding optimal methods for communicating such sources over noisy channels present many novel joint-source channel coding problems.

At times the final objective of communication is achieving some coordination between various agents \cite{paul}. In these scenarios channel is used for both communicating data and achieving coordination. A new class of problem lends itself to us when we try to figure out the  tradeoffs between error exponents of the coordination and data?

We can also actively use \uep~ in network protocols. For example, a relay can forward some partial information even if it cannot decode everything. This partial information could be characterized in terms of special bits as well as special messages. Another example is two-way communication, where \uep~ can be used for more reliable feedback and synchronization.

Information theoretic understanding of \uep~ also gives rise to some network optimization problems. With \uep, the interface to physical layer is no longer bits. Instead, it is a collection of various levels of error protection. The achievable channel resources of reliability and rate need to be efficiently divided amongst these levels, which gives rise to  many resource allocation problems.

\section{Block Codes without Feedback: Proofs}\label{sec:proofs}
In the following sections, we use the following standard notation for entropy, conditional entropy and mutual information,
\begin{align*}
  H(P_X)
&=~~\sum_{j\in {\cal X}} P_X(j) \ln \tfrac{1}{P_X(j)}\\
H(W_{Y|X}|P_X)
&=\sum_{j\in {\cal X}, k \in {\cal Y}} P_X(j)  W_{Y|X}(k|j) \ln \tfrac{1}{W_{Y|X}(k|j)}\\
I(P,W)
&=\sum_{j\in {\cal X}, k \in {\cal Y}} P_X(j)  W_{Y|X}(k|j) \ln  \tfrac{W_{Y|X}(k|j)}{\sum_{i \in {\cal X}}W_{Y|X}(k|i) P_X(i)}.
\end{align*}
In addition we denote the decoding region of a message $i \in {\cal M}$ by $\DEC{i}$, i.e.
\begin{equation*}
  \DEC{i}\DEF \{y^{\BLX}: \hat{\mes}(y^{\BLX})=i\}.
\end{equation*}
 \subsection{Proof of Theorem \ref{thm:bit}}\label{sec:proofthmone}
\begin{proof}
We first show that any capacity achieving sequence $\SC$ with $\Eb_{,\SC}$ can be used to construct another  capacity achieving sequence, $\SC'$ with $\Eb_{,\SC'}=\tfrac{\Eb_{,\SC}}{2}$,  all members of which are fixed composition codes. Then we show that $\Eb_{,\SC'}=0$ for any capacity achieving sequence, $\SC'$ which only includes fixed composition codes.  

Consider a capacity achieving sequence, $\SC$ with message sets ${\cal M}^{(\BLX)}= {\cal M}_1 \times {\cal M}_2^{(\BLX)}$, where ${\cal M}_1=\{0,1\}$. As a result of Markov inequality, at least $\tfrac{4}{5}|{\cal M}^{(\BLX)}|$ of the messages in ${\cal M}^{(\BLX)}$  satisfy,
\begin{equation}
\label{eq:uniformer1}
  \PCX{\hat{\mes_1}\neq \mes_1}{\mes=i}  \leq 5 \PX{\hat{\mes_1}\neq \mes_1}.
\end{equation}
Similarly at least   $\tfrac{4}{5}|{\cal M}^{(\BLX)}|$ of the messages in ${\cal M}^{(\BLX)}$ satisfy,
\begin{equation}
\label{eq:uniformer2}
  \PCX{\hat{\mes}\neq \mes}{\mes=i}\leq 5  \PX{\hat{\mes}\neq \mes}.
\end{equation}

Thus at least  $\tfrac{3}{5}|{\cal M}^{(\BLX)}|$ of the messages in ${\cal M}^{(\BLX)}$ satisfy both (\ref{eq:uniformer1}) and (\ref{eq:uniformer2}). Consequently at least $\tfrac{1}{10}|{\cal M}^{(\BLX)}|$  messages are of the form $(0, \mes_2)$ and satisfy equations (\ref{eq:uniformer1}) and (\ref{eq:uniformer2}).  If we group them  according to their empirical distribution at least one of the groups will have more than $\tfrac{|{\cal M}^{(\BLX)}|}{10 (\BLX+1)^{|{\cal X}|}}$ messages because the number of different empirical distributions for elements of ${\cal X}^\BLX$ is less than  $(\BLX+1)^{|{\cal X}|}$. We keep the first $\tfrac{|{\cal M}^{(\BLX)}|}{10 (\BLX+1)^{|{\cal X}|}}$  codewords of this most populous type, denote them by $\fix_A(\cdot)$ and throw away all of other codeword corresponding to the messages of the form $(0,\mes_2)$. We  do the same for the messages of the form $\mes=(1, \mes_2)$ and denote corresponding codewords by $\fix_B(\cdot)$. 

Thus we  have a length $\BLX$ code  with message set ${\cal M}'$ of the form  ${\cal M}'={\cal M}_1 \times {\cal M}_2'$ where  ${\cal M}_1=\{0,1\}$ and $|{\cal M}_2'|=\tfrac{|{\cal M}_2'|}{{10 (\BLX+1)^{|{\cal X}|}}}$. Furthermore,
\begin{equation*}
  \PCX{\hat{\mes}_1'\neq \mes_1'}{\mes'=i}  \leq 5 \PX{\hat{\mes}_1\neq \mes_1}  \qquad 
  \PCX{\hat{\mes}'\neq \mes'}{\mes'=i}\leq 5  \PX{\hat{\mes}\neq \mes} \quad \forall i \in {\cal M}'.
\end{equation*}

 Now let us  consider following  $2 \BLX$ long block code with message set ${\cal M}''={\cal M}_1 \times {\cal M}_2''\times {\cal M}_3''$  where ${\cal M}_2''={\cal M}_3''={\cal M}_2'$.  If $\mes''=(0,\mes_2'',\mes_3'')$ then $\brx(\mes'')= \fix_A(\mes_2'') \fix_B(\mes_3'')$. If $\mes''=(1,\mes_2'',\mes_3'')$ then $\brx(\mes'')= \fix_B(\mes_2'') \fix_A(\mes_3'')$. Decoder of this new length $2\BLX$ code  uses the decoder of the original length $\BLX$ code first on  $y^{\BLX}$ and then on $y_{\BLX+1}^{2\BLX}$. If the concatenation of length $\BLX$ codewords corresponding to the decoded halves,  is a codeword for an  $i \in {\cal M}''$ then $\hat{\mes}''=i$. Else an arbitrary message is decoded. One can easily see that the error probability of the length $2\BLX$ code is  less than the twice  the error probability of the length $\BLX$ code, i.e.
\begin{align*}
  \PCX{\hat{\mes}''\neq \mes''}{\mes''}
& \leq 1 - (1-  \PCX{\hat{\mes}'\neq \mes'}{\mes'=\mes_2''}) (1-  \PCX{\hat{\mes}'\neq \mes'}{\mes'=\mes_3''})\\
& \leq 2  \PX{\hat{\mes}'\neq \mes'}.
  \end{align*}
 Furthermore bit error probability of the new code is also at most twice the bit error probability of the length $\BLX$ code, i.e.
\begin{align*}    
\PCX{\hat{\mes}_1''\neq \mes_1''}{\mes_1''}
& \leq 1 - (1-  \PCX{\hat{\mes}_1'\neq \mes_1'}{\mes_1'=\mes_1''}) (1-  \PCX{\hat{\mes}_1'\neq \mes_1'}{\mes_1'=\mes_1''})\\
& \leq 2  \PX{\hat{\mes}_1' \neq \mes_1'}
  \end{align*}
Thus using these codes one can obtain a capacity achieving sequence $\SC'$ with $\Eb_{,\SC'} =\tfrac{\Eb_{,\SC}}{2}$ all members of which are fixed composition codes.

In the following discussion we focus on capacity achieving sequences, $\SC$'s  which are composed of fixed composition codes only.  We will show that $\Eb_{,\SC}=0$ for all capacity achieving  $\SC$'s  with fixed composition codes.  Consequently the discussion above implies that $\Eb=0$. 

We  call the empirical distribution of a given output sequence, $y^{\BLX}$, conditioned on the code word, $\brx(i)$,  the  conditional type of $y^{\BLX}$ given the message $i$ and denote it by $\ctype{y^{\BLX}}{i}$. Furthermore we call the set of $y^{\BLX}$'s whose conditional type with message $i$ is $V$,  the $V$-shell of $i$ and denote it by $\shell{V}{i}$.  Similarly we  denote the  set of output  sequences $y^{\BLX}$ with the empirical distribution $U_Y$, by $\shelly{U_Y}$.

 We  denote the empirical distribution of the codewords of the $\BLX^{th}$ code of the sequence by $P_X^{(\BLX)}$ and the corresponding output distribution by $P_{Y}^{(\BLX)}$, i.e. %
\begin{equation*}
  P_Y^{(\BLX)}(\cdot)=\sum_{i \in {\cal X}} W_{Y|X}(\cdot|i) P_X^{(\BLX)}(i).
\end{equation*}
We simply use  $P_X$ and $P_Y$ whenever the value of  $\BLX$ is unambiguous from the context.  Furthermore $\PYIID{\cdot}$ stands for the probability measure on ${\cal Y}^{\BLX}$ such that
 \begin{equation*}
   \PYIID{y^{\BLX}} =\prod_{k=1}^{\BLX} P_Y (y_k).
 \end{equation*}
$\Dshell{0}$ is  the set of $y^{\BLX}$'s for which  $\hat{\mes}_1=0$ and $\ctype{y^{\BLX}}{\hat{\mes}(y^{\BLX})}=V$.
\begin{equation}
\Dshell{0} \DEF \{y^{\BLX} : \ctype{y^{\BLX}}{\hat{\mes}(y^{\BLX})}=V \mbox{ and } \hat{\mes} (y^{\BLX})=(0,j) \mbox{ for some } j \in {\cal M}_2 \}
\end{equation}
In other words, $\Dshell{0}$ is the set of $y^{\BLX}$'s such that $y^{\BLX} \in \shell{V}{ \hat{\mes}(y^{\BLX})}$ and decoded value of the first bit is zero. Note that since for each $y^{\BLX} \in {\cal Y}^{\BLX}$ there is a unique $\hat{\mes}(y^{\BLX})$ and  for each $y^{\BLX} \in {\cal Y}^{\BLX}$ and message $i \in {\cal M}$ there is unique $\ctype{y^{\BLX}}{i}$; each $y^{\BLX}$ belongs to a unique $\Dshell{0}$  or  $\Dshell{1}$, i.e. $\Dshell{0}$'s and $\Dshell{1}$'s are disjoint sets that collectively cover the set ${\cal Y}^{\BLX}$.  

Let us define the typical neighborhood of $W_{Y|X}$ as $[W]$
\begin{equation}
\label{eq:wbradef}
  [W] \DEF\{V_{Y|X}:  |V_{Y|X}(j|i)P_X^{(\BLX)}(i) - W_{Y|X}(j|i)P_X^{(\BLX)}(i)| \leq \sqrt[4]{1/\BLX} \quad \forall i,j \}
\end{equation}
Let us denote the union of all $\Dshell{0}$'s for typical $V$'s by $\DshellU{0} =\displaystyle{  \bigcup_{\substack{V \in [W]}} \Dshell{0}}$. We will establish the following inequality later.  Let us assume for the moment that it holds.
\begin{equation}
  \label{eq:aux-lem}
  \PYIID{\DshellU{0}} \geq e^{\BLX(R^{(\BLX)} -(\CX +\epsilon_{\BLX}))} \left( \tfrac{1}{2}- \tfrac{|{\cal X}||{\cal Y}|}{8 \sqrt{\BLX}} - \Pe \right)
\end{equation}
where $ \displaystyle{\lim_{\BLX \rightarrow \infty} \epsilon_{\BLX}=0}$.

As a result of bound given in (\ref{eq:aux-lem}) and the  blowing up lemma \cite[Ch. 1, Lemma 5.4]{CK}, we can conclude that for any capacity achieving sequence $\SC$, there exists a sequence of $(\ell_{\BLX}, \eta_{\BLX})$ pairs satisfying $\displaystyle{\lim_{\BLX \rightarrow \infty } \eta_{\BLX}=1}$  and $\displaystyle{\lim_{\BLX \rightarrow \infty} \tfrac{\ell_{\BLX}}{\BLX}=0}$ such that
\begin{equation*}
\PYIID{\Gamma^{\ell_{\BLX}} (\DshellU{0})} \geq \eta_{\BLX}  
\end{equation*}
where $\Gamma^{\ell_{\BLX}}(A)$ is the  set of all $y^{\BLX}$'s  which differs from an element of  $A$ in at most $\ell_{\BLX}$ places.   Clearly one can repeat the same argument for $\Gamma^{\ell_{\BLX}} ( \DshellU{1})$ to get,
\begin{equation*}
\PYIID{\Gamma^{\ell_{\BLX}} (\DshellU{1})} \geq \eta_{\BLX}.  
\end{equation*}
Consequently, 
\begin{align*}
\PYIID{\Gamma^{\ell_{\BLX}} (\DshellU{0})  \bigcap  \Gamma^{\ell_{\BLX}} (\DshellU{1})}
&=
\PYIID{\Gamma^{\ell_{\BLX}} (\DshellU{0})}+ \PYIID{\Gamma^{\ell_{\BLX}} (\DshellU{1})}
-\PYIID{\Gamma^{\ell_{\BLX}} (\DshellU{0})  \bigcup  \Gamma^{\ell_{\BLX}} (\DshellU{1})}
\\
\PYIID{\Gamma^{\ell_{\BLX}} (\DshellU{0})  \bigcap  \Gamma^{\ell_{\BLX}} (\DshellU{1})} 
&\geq 2\eta_{\BLX}-1.
\end{align*}
Note that if  $y^{\BLX} \in \Gamma^{\ell_{\BLX}} (\DshellU{1})$, then there exist  at least one element  $\tilde{y}^{\BLX} \in \shelly{P_Y}$ which differs from $y^{\BLX}$  in at most $(|{\cal Y}||{\cal X}| \BLX^{3/4}+ \ell_{\BLX} )$ places.\footnote{Because of the integer constraints  $\shelly{P_Y}$ might actually be an empty set. If so we can  make a similar argument for the $U_Y^*$ which minimizes $\sum_{j} |U_Y (j) - P_Y (j)|$. However this technicality is inconsequential.} Thus we can upper bound its probability by,
\begin{equation*}
 y^{\BLX} \in \Gamma^{\ell_{\BLX}} (\DshellU{1}) \Rightarrow \PYIID{y^{\BLX}} \leq e^{-\BLX H(P_Y) - (|{\cal Y}||{\cal X}|\BLX^{3/4}+ \ell_{\BLX} ) \ln \lambda}
\end{equation*}
where $\lambda=\min_{i,j} W_{Y|X}(j|i)$. Thus we have 
\begin{equation}
\label{eq:worstsize}
|\Gamma^{\ell_{\BLX}} (\DshellU{0})  \bigcap  \Gamma^{\ell_{\BLX}} (\DshellU{1})| \geq (2\eta_{\BLX}-1)  e^{\BLX H(P_Y) + (|{\cal Y}||{\cal X}| {\BLX^{3/4}}+ \ell_{\BLX} ) \ln \lambda}. 
\end{equation}
Note that for any $y^{\BLX} \in \Gamma^{\ell_{\BLX}} (\DshellU{0})  \bigcap  \Gamma^{\ell_{\BLX}} (\DshellU{1})$,  there exist a $\tilde{y}^{\BLX} \in \shell{W}{i}$ for an $i$ of the form $i=(0,\mes_2)$ which differs from $y^{\BLX}$ in at most $ (|{\cal Y}||{\cal X}| {\BLX^{3/4}}+ \ell_{\BLX} )$ places.\footnote{Integer constraints here are inconsequential too.} Consequently
\begin{equation}
\label{eq:yprobx}
  \PCX{y^{\BLX}}{\mes=i}\geq e^{-\BLX H(W_{Y|X}|P_X)+ (|{\cal Y}||{\cal X}|{\BLX^{3/4}}+ \ell_{\BLX} ) \ln \lambda}.
\end{equation}
Since ${\cal M}_2=\tfrac{e^{\BLX R^{(\BLX)}}}{2}$ using equation (\ref{eq:yprobx}) we can lower bound the probability of $y^{\BLX}$ under the hypothesis $\mes_1=0$ as follows,
\begin{align}
  \notag
\PCX{y^{\BLX}}{\mes_1=0}
&=\sum_{j \in {\cal M}_2}  \PCX{y^{\BLX}}{\mes=(0,j)} \PCX{\mes=(0,j)}{\mes_1=0}\\
&\geq 2 e^{-\BLX (H(W_{Y|X}|P_X)+R^{(\BLX)})+ (|{\cal Y}||{\cal X}|{\BLX^{3/4}}+ \ell_{\BLX} ) \ln \lambda}. 
\label{eq:worstprob1}
\end{align}
Clearly same holds for $\mes_1=1$ too, thus
\begin{equation}
\label{eq:worstprob2}
  \PCX{y^{\BLX}}{\mes_1=1}\geq 2 e^{-\BLX (H(W_{Y|X}|P_X)+R^{(\BLX)})+ (|{\cal Y}||{\cal X}|{\BLX^{3/4}}+ \ell_{\BLX} ) \ln \lambda} .
\end{equation}
Consequently,
\begin{align} 
\notag
\PX{\hat{\mes_1} \neq \mes_1}
&\geq 
\sum_{y^{\BLX}} \tfrac{1}{2} \min (\PCX{y^{\BLX}}{\mes_1=0}, \PCX{y^{\BLX}}{\mes_1=1})\\ 
\notag
&\mathop{\geq}\limits^{(a)} 
 \sum_{y^{\BLX} \in \Gamma^{\ell_{\BLX}} (\DshellU{0})  \bigcap  \Gamma^{\ell_{\BLX}} (\DshellU{1})}
e^{-\BLX (H(W_{Y|X}|P_X)+R^{(\BLX)})+ (|{\cal Y}||{\cal X}| {\BLX^{3/4}}+ \ell_{\BLX} ) \ln \lambda} \\ 
\notag
&\mathop{\geq}\limits^{(b)}  
 (2\eta_{\BLX}-1)  e^{\BLX H(P_Y) + (|{\cal Y}||{\cal X}|{\BLX^{3/4}}+ \ell_{\BLX} ) \ln \lambda} 
e^{-\BLX (H(W_{Y|X}|P_X)+R^{(\BLX)})+ (|{\cal Y}||{\cal X}| {\BLX^{3/4}}+ \ell_{\BLX} ) \ln \lambda} \\ 
\label{eq:rewexp1}
&=(2\eta_{\BLX}-1)  e^{\BLX (I(P_X,W) -R^{(\BLX)}) + 2(|{\cal Y}||{\cal X}|{\BLX^{3/4}}+ \ell_{\BLX} ) \ln \lambda}
\end{align}
where $(a)$ follows from  equations (\ref{eq:worstprob1}) and (\ref{eq:worstprob2}) and $(b)$ follows from equation (\ref{eq:worstsize}).

Using Fano's inequality we get,
\begin{equation}
  \label{eq:rewexp2}
  \MI{\mes}{Y^{\BLX}}-\BLX R^{(\BLX)} \geq -\ln 2  -\BLX R^{(\BLX)} \Pe^{(\BLX)}
\end{equation}
where $  \MI{\mes}{Y^{\BLX}}$ is the mutual information between the message $\mes$ and channel output $Y^{\BLX}$. In addition we can upper bound $  \MI{\mes}{Y^{\BLX}}$ as follows,
\begin{align}
 \MI{\mes}{Y^{\BLX}} 
\notag
&=\sum_{i\in {\cal M}, y^{\BLX} \in {\cal Y}^{\BLX}}
\PX{i, y^{\BLX}} \ln \tfrac{\PCX{y^{\BLX}}{i}}{\PX{y^{\BLX}}}\\
\notag
&=\sum_{i\in {\cal M}, y^{\BLX} \in {\cal Y}^{\BLX}}
\PX{i, y^{\BLX}} \ln \tfrac{\PCX{y^{\BLX}}{i}}{\prod_{k=1}^{\BLX} P_Y(y_k)}
-\sum_{y^{\BLX} \in {\cal Y}^{\BLX}}
\PX{y^{\BLX}} \ln \tfrac{\PX{y^{\BLX}}}{\prod_{k=1}^{\BLX} P_Y(y_k)}\\
\notag
&\mathop{\leq}^{(a)}
\sum_{i\in {\cal M}} \tfrac{1}{|{\cal M}|} \sum_{k=1}^{\BLX} \sum_{y_k}
W_{Y|X}(y_k|\bx_{k}(i)) \ln \tfrac{W_{Y|X}(y_k|\bx_{k}(i))}{ P_Y(y_k)}\\
\label{eq:rewexp3}
&\mathop{=}^{(a)}\BLX I(P_X,W)
\end{align}
where $P_Y(\cdot)=\sum_{j \in {\cal X}}W_{Y|X}(\cdot) P_X(j)$. Step $(a)$ follows the non-negativity of KL divergence and step $(b)$ follows from the fact that all the code words are of type $P_X(\cdot)$.

Using equations (\ref{eq:rewexp1}), (\ref{eq:rewexp2}) and (\ref{eq:rewexp3})   we get
\begin{equation*}
  \PX{\hat{\mes_1} \neq \mes_1} \geq (2\eta_{\BLX}-1)  e^{-\ln 2  -\BLX R^{(\BLX)} \Pe^{(\BLX)} + 2(|{\cal Y}||{\cal X}|{\BLX^{3/4}}+ \ell_{\BLX} ) \ln \lambda}
\end{equation*}
Thus using  
$\displaystyle{\lim_{\BLX \rightarrow \infty} \Pe^{(\BLX)}=0}$, 
$\displaystyle{\lim_{\BLX \rightarrow \infty} \eta_{\BLX}=1}$ and 
$\displaystyle{\lim_{\BLX \rightarrow \infty} \tfrac{\ell_{\BLX}}{\BLX}=0}$
  we conclude that,
\begin{equation*}
  \lim_{\BLX \rightarrow \infty} \tfrac{ -\ln \SPX{\BLX}{\hat{\mes_1} \neq \mes_1}}{\BLX}=0
\end{equation*}
Now only think left, for proving $\Eb=0$, is  to establish inequality (\ref{eq:aux-lem}). One can write the error probability of the $\BLX^{th}$ code of $\SC$ as 
\begin{align}
  \Pe^{(\BLX)}
\notag 
&=\sum_{i \in {\cal M}^{(\BLX)}} \tfrac{1}{\cal M} \sum_{y^{\BLX} \in {\cal Y}^{\BLX}} (1-\IND{\hat{\mes}(y^{\BLX})={i}}) \PCX{y^{\BLX}}{\mes=i} \\  
\notag
&=\sum_{i \in {\cal M}^{(\BLX)}} e^{-\BLX R^{(\BLX)}}   \sum_{V} \sum_{y^{\BLX} \in \shell{V}{i}}  (1-\IND{\hat{\mes}(y^{\BLX})={i}}) e^{-\BLX(\CKLD{V_{Y|X}(\cdot|X)}{W_{Y|X}(\cdot|X)}{P_X}+H(V_{Y|X}|P_{X}))  } \\  
\notag  
&=\sum_{V}  e^{-\BLX(\CKLD{V_{Y|X}(\cdot|X)}{W_{Y|X}(\cdot|X)}{P_X}+H(V_{Y|X}|P_{X}) +R^{(\BLX)})  } \sum_{i \in {\cal M}^{(\BLX)}}  \sum_{y^{\BLX} \in \shell{V}{i}}   (1-\IND{\hat{\mes}(y^{\BLX})={i}})  \\
\label{eq:overcount}
&=\sum_{V}  e^{-\BLX(\CKLD{V_{Y|X}(\cdot|X)}{W_{Y|X}(\cdot|X)}{P_X}+H(V_{Y|X}|P_{X}) +R^{(\BLX)})  } (Q_{0,V}+Q_{1,V})
\end{align}
where  $Q_{k,V}=\displaystyle{\sum_{\substack{i=(k,j)\\ j \in {\cal M}_2}}   \sum_{y^{\BLX} \in \shell{V}{i}}  (1-\IND{\hat{\mes}(y^{\BLX})={i}})}$ for $k=0,1$.

Note that $Q_{k,V}$ is the sum, over the messages $i$ for which $\mes_1=k$,  of the number of the elements in $\shell{V}{i}$ that are not decoded to  message $i$. In a sense it is a measure of the contribution of  the $V$-shells of different codewords to the error probability. We will use equation (\ref{eq:overcount}) to establish lower bounds on $\PYIID{\Dshell{0}}$'s.

 Note that all elements of $\Dshell{0}$ have the same probability under $\PYIID{\cdot}$ and  
\begin{equation}
\label{eq:prbone}
  \PYIID{\Dshell{0}}=|\Dshell{0}|  e^{-\zeta \BLX} \qquad \mbox{ where } \quad \zeta=\sum_{x,y} P_X(x) V_{Y|X} (y|x)  \ln \tfrac{1}{P_Y(y)}.
\end{equation}
Note that
\begin{align*}
\zeta&=\sum_{x,y} P_X(x) V_{Y|X} (y|x)  \ln \tfrac{W_{Y|X}(y|x)}{P_Y(y)}+\sum_{x,y} P_X(x) V_{Y|X}(y|x)  \ln \tfrac{1}{W_{Y|X}(y|x)}\\
     &= I(P_X,W_{Y|X}) +\CKLD{ V_{Y|X}(\cdot|X)}{W_{Y|X}(\cdot|X)}{P_X}+H(V_{Y|X}|P_X)\\
&\qquad \hspace{4cm} +{ \small \sum_{x,y} P_X(x) (V_{Y|X} (y|x)-W_{Y|X}(y|x))  \ln \tfrac{W_{Y|X}(y|x)}{P_Y(y)}}
\end{align*}
Recall that $I(P_X,W_{Y|X}) \leq \CX$ and $\displaystyle{\min_{i,j}W_{Y|X}(i|j)=\lambda}$. Thus using the definition of $[W_{Y|X}]$ given in equation (\ref{eq:wbradef}) we get,
\begin{equation}
\label{eq:prbtwo}
\zeta
\leq \CX+\epsilon_{\BLX} +\CKLD{ V_{Y|X}(\cdot|X)}{W_{Y|X}(\cdot|X)}{P_X}+H(V_{Y|X}|P_X) \qquad \forall V_{Y|X} \in [W_{Y|X}]
\end{equation}
where $\epsilon_{\BLX}=\tfrac{|{\cal X}||{\cal Y}|}{\sqrt[4]{\BLX}} \ln \tfrac{1}{\lambda}$.

Note that
\begin{equation}
\label{eq:prbthree}
|\Dshell{0}|=|{\cal M}_2^{(\BLX)}| \cdot |\shell{V}{i}| - Q_{0,V} =\tfrac{1}{2}|\shell{V}{i}| e^{\BLX R^{(\BLX)}}- Q_{0,V}.
\end{equation}
Recalling that $\Dshell{0}$'s are disjoint and  using equations (\ref{eq:prbone}), (\ref{eq:prbtwo}) and (\ref{eq:prbthree}) we get
\begin{align}
\PYIID{\DshellU{0}}
\notag &\geq \sum_{V \in  [W] } \PYIID{\Dshell{0}}\\
\notag &\geq \sum_{V \in  [W] } e^{-\BLX (\CX +\epsilon_{\BLX})} \left(\tfrac{1}{2}|\shell{V}{i}| e^{\BLX R^{(\BLX)}}- Q_{0,V} \right)
 e^{-\BLX(\CKLD{V_{Y|X}(\cdot|X)}{W_{Y|X}(\cdot|X)}{P_X}+H(V_{Y|X}|P_{X}))}\\
\notag & \mathop{\geq}\limits^{(a)} 
 e^{\BLX(R^{(\BLX)} -(\CX +\epsilon_{\BLX}))} \left( \sum_{V \in  [W]} \tfrac{1}{2}|\shell{V}{i}|  e^{-\BLX(\CKLD{V_{Y|X}(\cdot|X)}{W_{Y|X}(\cdot|X)}{P_X}+H(V_{Y|X}|P_{X}))}  - \Pe \right) \\
\notag & = e^{\BLX(R^{(\BLX)}-(\CX +\epsilon_{\BLX}))} \left( \tfrac{1}{2}  \sum_{V \in  [W] } \sum_{y^{\BLX} \in \shell{V}{i}} \PCX{y^{\BLX}}{\mes =i}   - \Pe \right) \\
\notag & \mathop{\geq}\limits^{(b)} e^{\BLX(R^{(\BLX)} -(\CX +\epsilon_{\BLX}))} \left( \tfrac{1}{2}- \tfrac{|{\cal X}||{\cal Y}|}{8 \sqrt{\BLX}} - \Pe \right)
\end{align}
where $(a)$ follows the equation (\ref{eq:overcount}) and $(b)$ follows from the Chebyshev's inequality.\footnote{The claim in $(b)$  is identical to the one in \cite{CK}[Remark on page 34]}  
 \end{proof}

\subsection{Proof of Theorem \ref{thm:md}}
\subsubsection{Achievability:  $\Emd\ge\PCAP$}~

\begin{proof}
For each block length $\BLX$,  the special message is sent with the length $\BLX$ repetition sequence $\bx^{\BLX}(1)=(\xr,\xr\cdots,\xr)$ where $x_r$ is the input letter satisfying
\begin{equation*}
\KLD{P_{Y}^{*}(\cdot)}{W_{Y|X}(\cdot|x_r)}= \max_{i}\KLD{P_{Y}^{*}(\cdot)}{W_{Y|X}(\cdot|i)}.
\end{equation*}
The remaining $|{\cal M}|-1$ ordinary codewords are generated  randomly and independently of each other using capacity achieving input distribution $P_X^*$ i.i.d. over time.

Let us denote the empirical distribution of  a particular output sequence $y^{\BLX}$ by $\typey{y^{\BLX}}$. The receiver decodes to the special message only when the output distribution is not close to $P_Y^{*}$. Being more precise,  the set of output sequences close to $P_Y^{*}$,  $[P_Y^{*}]$, and decoding region of the special message, $\DEC{1}$, are given as follows,
  \begin{align*}
 [P_Y^{*}]
&=\{P_Y(\cdot):  \lVert P_Y(i)- P_{Y}^{*} (i) \rVert \leq  \sqrt[4]{1/\BLX}  \quad   \forall i \in {\cal Y} \}&  \DEC{1}&=\{y^{\BLX}:  \typey{y^{\BLX}} \in  [P_Y^{*}] \}.
  \end{align*}
 Since there are at most $(\BLX+1)^{|{\cal Y}|}$ different empirical output distribution for elements of ${\cal Y}^{\BLX}$ we get,
 \begin{equation*}
\SPCX{\BLX}{y^{\BLX} \notin \DEC{1}}{\mes=1} \leq (\BLX+1)^{|{\cal Y}|} e^{-\BLX \min_{Q_Y \in [P_Y^{*}]} \KLD{Q_{Y}(\cdot)}{W_{Y|X}(\cdot|\xr) }}
\end{equation*}
Thus $
\displaystyle{\lim_{\BLX \rightarrow \infty} \tfrac{- \ln \SPCX{\BLX}{y^{\BLX} \notin \DEC{1}}{\mes=1} }{\BLX}= \KLD{P_{Y}^{*}(\cdot)}{W_{Y|X}(\cdot|\xr) } =\PCAP}$.

Now the only thing we are left with to prove is that we can have low enough probability for the remaining messages. For doing that we will first calculate the average error probability of  the following  random code ensemble. 

Entries  of the codebook, other than the ones corresponding to the special message, are generated independently  using a capacity achieving input distribution $P_X^{*}$. Because of the symmetry  average error probability is same for all $i \neq 1 $ in ${\cal M}$. Let us calculate the error probability of the message $\mes=2$.

Assuming that the  second message  was transmitted, $\PCX{y^{\BLX} \in \DEC{1}}{\mes =2}$ is  vanishingly small. It is because, the output distribution for the random ensemble for ordinary codewords is  i.i.d. $P_Y^*$. Chebyshev's inequality guarantees that probability of the output type being outside a $\sqrt[4]{1/\BLX}$ ball around $P_Y^*$, i.e. $[P_Y^{*}]$,  is of the order $\sqrt{1/\BLX}$.

Assuming that the  second message  was transmitted, $\PCX{y^{\BLX} \in \cup_{i >2}\DEC{i}}{\mes =2}$ is vanishingly small  due to the standard random coding argument for achieving capacity \cite{shannon}.

Thus for any $\Pe>0$ for all large enough $\BLX$ average error probability of the code ensemble is smaller than $\Pe$ thus we  have at least one code with that $\Pe$. For that code at least half of the codewords have an error probability less then $2\Pe$.  
\end{proof}

\subsubsection{Converse:  $\Emd \leq \PCAP$}
In the  section \ref{subsec:fmdconverse}, we will prove that even with feedback and variable decoding time, the missed-detection exponent of a single special message is at most $\PCAP$. Thus $\Emd \leq \PCAP$.

\subsection{Proof of Theorem \ref{thm:mdmany}}
\subsubsection{Achievability:  $\Emdr \geq  E(r) $}~

\begin{proof}
{\noindent{\bf{Special codewords:}}} At any given block length 
$\BLX$, we start with a optimum codebook (say ${\cal C}_{\textrm special}$) for $\lceil e^{\BLX r}\rceil$ messages. Such optimum  codebook achieves error exponent $E(r)$ for every message in it. 
\begin{equation*}
\PX{\hat\mes\neq i|\mes=i} \doteq e^{-\BLX E(r)} \qquad\forall i\in{\cal M}_s \equiv\{1,2,\cdots,\lceil e^{\BLX r}\rceil\}  
\end{equation*}
Since there are at most $(\BLX+1)^{|{\cal X}|}$ different types, there is at least one type $\shellx{P_X}$ which has $\tfrac{\lceil e^{\BLX r}\rceil}{(1+\BLX)^{|{\cal X}|}}$ or more codewords. Throw away all other codewords from ${\cal C}_{\textrm special}$ and lets call the remaining fixed composition codebook as ${\cal C}'_{\textrm special}$. Codebook ${\cal C}'_{\textrm special}$ is used for transmitting the special messages.

As shown in Fig. \ref{fig:footballs}(a), let the noise ball around the codeword for the special message $i$ be ${\cal B}_i$. These balls need not be disjoint.  Let ${\cal B}$ denote the union of these balls of all special messages.
\begin{equation*}
{\cal B}=\bigcup_{i\in{\cal M}_s}{\cal B}_i  
\end{equation*}
If the output sequence $y^{\BLX}  \in {\cal B}$, the first stage of the decoder decides a special message was transmitted. The second stage then chooses the ML candidate amongst the messages in ${\cal M}_s$.

Let us define ${\cal B}_i$ precisely now.
\begin{equation*}
\label{eq:close}
  {\cal B}_i=\{y^{\BLX}: \ctype{y^{\BLX}}{i} \in {\cal W}(r+\slack, P_X)  \}
\end{equation*}
where ${\cal W}(r+\slack, P_X)= \{V_{Y|X}: \CKLD{V_{Y|X}(\cdot|X)}{W_{Y|X}(\cdot|X)}{P_X} \leq E_{\textrm{sp}}(r+\slack;P_X) \}$. Recall that the sphere-packing exponent for input type $P_X$ at rate $r$,  $E_{\textrm{sp}}(r;P_X)$ is given by,
\begin{equation*}
 E_{\textrm{sp}}(r;P_X)=\min_{V_{Y|X}: I(P_X,V_{Y|X}) \le r}\ \CKLD{V_{Y|X}(\cdot|X)}{W_{Y|X}(\cdot|X)}{P_X}
\end{equation*}

{\noindent{\bf{Ordinary codewords:}}} The ordinary codewords are generated randomly using a capacity achieving  input distribution  $P_X^*$. This is the same as Shannon's construction for achieving capacity. The random coding construction provides a simple way to show that in the cavity ${\cal B}^c$ (complement of ${\cal B}$), we can essentially fit enough typical noise-balls to achieve capacity. This avoids the complicated task of carefully choosing the ordinary codewords and their decoding regions in the cavity, ${\cal B}^c$.

If the output sequence $y^{\BLX} \in  {\cal B}^c$, the first stage of the decoder decides an ordinary message was transmitted. The second stage then chooses the ML candidate from ordinary codewords.

\vspace{.2cm}

{\noindent{\bf{Error analysis:}}} First, consider the case when a special codeword $\bx^{\BLX}(i)$ is transmitted. By Stein's lemma and definition of ${\cal B}_i$, the probability of $y^n\notin{\cal B}_i$ has exponent $E_{\textrm{sp}}(r+\slack;P_X)$. Hence the first stage error exponent is at least $E_{\textrm{sp}}(r+\slack;P_X)$.

Assuming correct first stage decoding,  the second stage error exponent for  special messages equals $E(r)$. Hence the effective error exponent for special messages is
\begin{equation*}
  \min\{E(r),E_{\textrm{sp}}(r+\slack;P_X)\}
\end{equation*}
Since $E(r)$ is at most the sphere-packing exponent $E_{\textrm{sp}}(r;P_X)$, \cite{gallager_fixed}, choosing arbitrarily  small $\slack$ ensures that missed-detection exponent of each special message equals  $E(r)$.

Now consider the situation of a uniformly chosen ordinary codeword being transmitted. We have to make sure that the error probability is vanishingly small now. In this case,  the output sequence distribution is i.i.d. $P_Y^*$ for the random coding ensemble. The first stage decoding error happens when $y^{\BLX} \in \bigcup{\cal B}_i$. Again by Stein's lemma, this exponent for any particular ${\cal B}_i$ equals $E_{\textrm{o}}$:
\begin{align*}
E_{\textrm{o}}
&= \min_{V_{Y|X} \in {\cal W}(r+\slack, P_X)}  \CKLD{V_{Y|X}(\cdot|X)}{P_{Y}^*(\cdot)}{P_X}\\
&\mathop{=}^{(a)} \min_{V_{Y|X} \in {\cal W}(r+\slack, P_X)} I(P_X,V_{Y|X}) + \KLD{(P V)_{Y}(\cdot)}{P_{Y}^*(\cdot)}\\
&\mathop{\geq}^{(b)} \min_{V_{Y|X} \in {\cal W}(r+\slack, P_X)} I(P_X,V_{Y|X}) \\
&\mathop{\geq}^{(c)} r+\slack  
\end{align*}
where in $(P V)_Y$ in $(a)$ is given by $(P V)_Y(j)=\sum_{i} P_X(i) V_{Y|X}(j|i)$, $(b)$ follows from the non-negativity of the KL divergence and $(c)$ follows from the definition of sphere-packing exponent and ${\cal W}(r+\slack, P_X)$.

Applying union bound over the  special messages, the probability of first stage decoding error after sending an ordinary message is at most $\doteq\exp(\BLX r-\BLX E_{\textrm{o}})$. We have already shown that $E_{\textrm{o}}\ge r+\slack$, which ensures that probability of first stage decoding error for ordinary messages is at most $\doteq e^{-\BLX\slack}$ for the random coding ensemble. Recall that for the random coding ensemble, average error probability of the second-stage decoding also vanishes below capacity. To summarize, we have shown these two properties of the random coding ensemble: 
\begin{enumerate}
\item Error probability of first stage decoding  vanishes as $a^{(\BLX)}\doteq\exp(-\BLX\slack)$ with $\BLX$ when a uniformly chosen ordinary message is transmitted.
\item Error probability of second stage decoding (say $b^{(\BLX)}$) vanishes with $\BLX$ when a uniformly chosen ordinary message is transmitted.
\end{enumerate}
Since the first error probability is at most $4a^{(\BLX)}$ for some $3/4$ fraction of codes in the random ensemble, and the second error probability is at most $4b^{(\BLX)}$ for some  $3/4$ fraction, there exists a particular code which satisfies both these properties. The overall error probability for ordinary messages is at most $4(a^{(\BLX)}+b^{(\BLX)})$, which vanishes with $\BLX$. We will use this particular code for the ordinary codewords. This de-randomization completes our construction of a reliable code for ordinary messages to be combined with the code  ${\cal C}_{\textrm special}$ for special messages. 
\end{proof}
\subsubsection{Converse:  $\Emdr \leq  E(r) $}
The converse argument for this result is obvious. Removing the ordinary messages from the code can only improve the error probability of the special messages. Even then, (by definition) the best missed detection exponent for the special messages equals $E(r)$.

\subsection{Proof of Theorem \ref{thm:mde}}
Let us now address the case with erasures. In this achievability result, the first stage of decoding remains unchanged from the no-erasure case.

\begin{proof}
We use essentially the same strategy as before. Let us start with a good code for $\lceil e^{\BLX r}\rceil$ messages allowing erasure decoding. Forney had shown in \cite{forney2} that, for symmetric channels an error exponent equal to $E_{\textrm{sp}}(r)+C-r$ is achievable while ensuring that erasure probability vanishes with  $\BLX$. We can use that code for these  $\lceil e^{\BLX r}\rceil$  codewords. As before, for $y^{\BLX}\in\bigcup_i{\cal B}_i$, the first stage decides a special codeword was sent. Then the second stage applies the erasure decoding method in \cite{forney2} amongst the special codewords.

With this decoding rule, when a special message is transmitted, error probability of the two-stage decoding is bottle-necked by the first stage: its error exponent $E_{\textrm{sp}}(r+\slack)$ is smaller than that of the second stage ($E_{\textrm{sp}}(r)+C-r$). By choosing arbitrarily small $\slack$, the special messages can achieve $E_{\textrm{sp}}(r)$ as their missed-detection exponent. 

The ordinary codewords are again generated i.i.d. $P_X^*$. If the first stage decides in favor of the ordinary messages,  ML decoding is implemented among  ordinary codewords. If an ordinary message was transmitted, we can ensure a vanishing error probability as before by repeating earlier arguments for no-erasure case.
\end{proof}
\section{Variable Length Block Codes with Feedback: Proofs}
\label{sec:proofsf}
In this section we will present a more detailed discussion of bit-wise and message wise \uep~ for variable length block codes with feedback by proving the Theorems \ref{thm:bitf}, \ref{thm:rf}, \ref{thm:many}, \ref{thm:mdf} and \ref{thm:msru}. In the proofs of converse results we need to discuss issues related with the conditional entropy of the messages given the observation of the receiver. In those discussion we use the following  notation for conditional entropy and conditional  mutual information,
\begin{align*}
\HX{{n}}
&=- \sum_{i \in \cal M} \PCX{\mes=i}{Y^n} \ln \PCX{\mes=i}{Y^n}\\
\CMI{\mes}{Y_{n+1}}{Y^n}
&= \HX{n}-\ECX{\HX{n+1}}{Y^n}.
\end{align*}
It is worth noting that this notation is different from widely used one, which includes a further expectation over the the conditioned variable.  ``$H(M|Y^n)$''  in the conventional notation, stands for the $\EX{\HX{n}}$ and ``$H(M|Y^{n}=y^{n})$'' stands for $\HX{n}$.
\subsection{ Proof  of Theorem \ref{thm:bitf}}
\subsubsection{Achievability:  $\FEb \geq \PCAP$}\label{sec:bitfach}~

This single  special bit exponent is achieved using the missed detection exponent of a single special message, indicating a decoding error for the special bit. The decoding error for the bit goes unnoticed when this special message is not detected. This shows how feedback connects bit-wise \uep~ to message-wise \uep~ in a fundamental manner.

\begin{proof}
We will prove that $\FEb \geq \PCAP$ by constructing a capacity achieving sequence with feedback, $\SC$, such that $\FEb_{,\SC}=\PCAP$. For that let $\SC'$ be a capacity achieving sequence such that $\Emd_{,\SC'}=\PCAP$. Note that existence of such a $\SC'$ is guaranteed as a result of Theorem \ref{thm:md}. We  first construct a two phase fixed length block code with feedback and erasures. Then using this we obtain the $\inx^{th}$ element of $\SC$.

In the first phase one of the two input symbols, $x_0$ and $x_1$, with distinct output distributions\footnote{Two input symbols $x_0$ and $x_1$ are such that  $W(\cdot|x_1)\neq W(\cdot|x_0)$ } is send for  $\lceil  \sqrt{\inx} \rceil$ time units depending on $\mes_1$. At time  $\lceil  \sqrt{\inx} \rceil$ receiver makes tentative decision $\tilde{\mes}_1$ on message $\mes_1$. Using Chernoff bound it can easily be shown that, \cite[Theorem 5]{sgb}
\begin{equation*}
  \PX{\tilde{\mes}_1 \neq \mes_1 } \leq e^{-\mu \sqrt{\inx}  }  \qquad \mbox{where~} \mu>0 
\end{equation*}
Actual value of $\mu$, however, is immaterial to us we are merely interested in  finding an upper bound on $\PX{\tilde{\mes}_1 \neq \mes_1 }$ which goes to zero as $\inx$ increases. 

In the second phase transmitter uses the $\inx^{\mbox{th}}$ member of $\SC'$. The message in the second phase, $\ames$,  is determined by $\mes_2$ depending on whether $\mes_1$ is decoded correctly  or not at the end of the first phase.
\begin{align*}
  \tilde{\mes}_1 \neq \mes_1  &\Rightarrow \ames=1\\
  \tilde{\mes}_1  = \mes_1 \mbox{ and } \mes_2=i  &\Rightarrow \ames=i+1 \quad \forall i
\end{align*}
At the end of the second phase decoder  decodes $\ames$ using the decoder of $\SC'$. If the decoded message is one, i.e. $\hat{\ames} =1$ then receiver declares an erasure, else $\hat{\mes}_1= \tilde{\mes}_1$ and  $\hat{\mes}_2=\hat{\ames}-1$.

Note that erasure probability of the two phase fixed length block code is upper bounded as
\begin{align}
\notag  \PX{\hat{\ames}=1}
&\leq   \PX{\tilde{\mes}_1 \neq \mes_1 }+ \PCX{\ames=1}{\ames \neq 1}\\
\label{eq:var-era}
&\leq   e^{-\mu \sqrt{\inx}} + \tfrac{{\cal M}^{'(\inx)}}{{\cal M}^{'(\inx)}-1} \Pe'^{(\inx)}
\end{align}
where $\Pe'^{(\inx)}$ is the error probability of the $\inx^{th}$ member of $\SC'$.

Similarly we can upper bound the probabilities of two error events associated with the two phase fixed length block code as follows
\begin{align}
\label{eq:var-ero1}
 \PX{\hat{\mes}_1 \neq \mes_1~,~  \hat{\ames}  \neq 1 }  &\leq   \Pe'^{(\inx)}(1)\\
\label{eq:var-ero}
\PX{  \hat{\mes} \neq \mes   ~,~ \hat{\ames}  \neq 1   }  &\leq    \tfrac{{\cal M}^{'(\inx)}}{{\cal M}^{'(\inx)}-1} \Pe'^{(\inx)} + \Pe'^{(\inx)}(1)
\end{align}
where $\Pe'^{(\inx)}(1)$ is the conditional error probability of the  $1^{st}$ message in the $\inx^{th}$ element of $\SC'$.

If  there is an erasure the transmitter and the receiver will  repeat what they have done again, until they get  $\hat{\ames}\neq 1$. If we sum the probabilities of all the error events, including error events in the possible repetitions we get;
\begin{align}
\label{eq:spbit-con1}    \PX{\hat{\mes}_1 \neq \mes_1} &=  \tfrac{ \PX{\hat{\mes}_1 \neq \mes_1~,~  \hat{\ames}  \neq 1 }} {1- \PX{\hat{\ames}=1}}\\
\label{eq:spbit-con2}   \PX{\hat{\mes} \neq \mes} &=  \tfrac{ \PX{  \hat{\mes} \neq \mes ~,~ \hat{\ames}  \neq 1   } } {1- \PX{\hat{\ames}=1}}
\end{align}
Note that expected decoding time of the code is
\begin{equation}
\label{eq:exp-dec1} \EX{\blx} = \tfrac{\inx+ \lceil\sqrt{\inx} \rceil }{1- \PX{\hat{\ames}=1}}\\
\end{equation}

Using equations (\ref{eq:var-era}), (\ref{eq:var-ero1}), (\ref{eq:var-ero}), (\ref{eq:spbit-con1}), (\ref{eq:spbit-con2}) and (\ref{eq:exp-dec1}) one can conclude that the resulting sequence of variable length block codes with feedback, $\SC$, is reliable. Furthermore $R_{\SC}=\CX$ and $\FEb_{,\SC}=\PCAP$.
\end{proof}
\subsubsection{Converse:  $\FEb \leq \PCAP$}~

We will use a converse result we have not proved yet, namely  converse part of Theorem \ref{thm:mdf}, i.e. $\FEmd \leq \PCAP$. 

  \begin{proof}
Consider a capacity achieving sequence, $\SC$, with message set sequence ${\cal M}^{(\inx)}=\{0,1 \} \times {\cal M}_2^{(\inx)}$. Using $\SC$ we construct another capacity achieving sequence $\SC'$ with a special message $0$,  with message set sequence   ${\amesX}^{(\inx)}=\{ 0 \} \cup {\cal M}_2^{(\inx)}$ such that $\FEmd_{,\SC'} = \FEb_{,\SC}$. This implies  $\FEb \leq \FEmd$, which together with Theorem \ref{thm:mdf}, $\FEmd \leq \PCAP$, gives us $\FEb \leq \PCAP$.

Let us denote the message of  $\SC$  by $\mes$ and that of $\SC'$ by $\ames$.  The $\inx^{\mbox{th}}$ code of $\SC'$ is as follow. At time $0$ receiver   chooses randomly an $\mes_1$ for $\inx^{\mbox{th}}$ element of $\SC$ and send its choice through feedback channel to transmitter. If the message of $\SC'$ is not $0$, i.e. $\ames \neq 0$ then the transmitter uses the codeword for $\mes=(\mes_1,\ames)$ to convey $\ames$. If $\ames=0$ receiver pick a $\mes_2$ with uniform distribution on ${\cal M}_2$ and uses the code word for $\mes=(1-\mes_1,\mes_2)$ to convey that $\ames=0$.

Receiver makes decoding using the decoder of $\SC$: if $\hat{\mes}= (\mes_1,i)$  then  $\hat{\ames}=i$, if $\hat{\mes}=(1-\mes_1,i)$ then  $\hat{\ames}=0$. One can easily show that expected decoding time and error probability of both of the codes are same. Furthermore error probability of $\mes_1$ in $\SC$ is equal to conditional error probability of message $\ames=0$ in $\SC'$ thus, $\FEmd_{,\SC'} = \FEb_{,\SC}$. 
\end{proof}
\subsection{Proof of  Theorem \ref{thm:rf}}
\subsubsection{Achievability: $\FEbr (r) \geq  \left( 1- \tfrac{r}{\CX}\right) \PCAP$}~

\begin{proof}
We will construct the capacity achieving sequence with feedback  $\SC$ using a capacity achieving sequence $\SC'$ satisfying  $\Emd_{,\SC'}=\PCAP$, as we did in the proof of theorem \ref{thm:bitf}.  We know that such a sequence exists, because of  Theorem \ref{thm:mdf}.

For $\inx^{\mbox{th}}$ member of $\SC$, consider the following two phase errors and erasures code. In the first phase transmitter  uses the  $\lfloor r \inx \rfloor ^{\mbox{th}}$ element of $\SC'$ to convey $\mes_1$. Receiver makes a tentative decision $\tilde{\mes}_1$.  In the second phase transmitter  uses  the $ \lfloor (\CX-r)\inx \rfloor ^{\mbox{th}}$ element of $\SC'$ to convey $\mes_2$ and whether $\tilde{\mes}_1=\mes_1$ or not, with a  mapping similar to the one we had in the proof of theorem \ref{thm:bitf}.
\begin{align*}
  \tilde{\mes}_1 \neq \mes_1  &\Rightarrow \ames=1\\
  \tilde{\mes}_1  = \mes_1 \mbox{ and } \mes_2=i  &\Rightarrow \ames=i+1 \quad \forall i
\end{align*}
Thus ${\cal M}_1^{(\inx)}={\amesX}^{(\lfloor r \inx \rfloor)} $ and ${\cal M}_2^{(\inx)} \cup \{|{\cal M}_2^{(\inx)}|+1 \}={\amesX}^{(\lfloor (\CX-r) \inx \rfloor)} $.  If we apply a decoding algorithm, like the one we had in  the proof of theorem \ref{thm:bitf}; going through essentially the same analysis with proof of Theorem \ref{thm:bitf}, we can conclude that $\SC$ is  a capacity achieving sequence and $\FEbr_{,\SC}=\left( 1- \tfrac{r}{\CX}\right) \PCAP$ and $r_{\SC}=r$.
\end{proof}
\subsubsection{Converse: $\FEbr  (r)\leq \left( 1- \tfrac{r}{\CX}\right) \PCAP$}\label{sec:proofthmsix}~\\
In establishing the converse we will use a technique that was used previously in \cite{burna2}, together with lemma \ref{lem:con1}  which we will prove in the converse part Theorem \ref{thm:mdf}.

\begin{proof}
 Consider  any variable length block code with feedback whose message set ${\cal M}$ is of the form ${\cal M}={\cal M}_1 \times {\cal M}_2$.  Let $t_{\substack{\delta}}$ be the first time instance that an   $i \in {\cal M}_1$ becomes more likely than $(1-\delta)$ and let $\ts =t_{\substack{\delta}} \wedge \blx $.

Recall that $\displaystyle{\min_{i,j} W_{Y|X}(j|i)=\lambda}$ consequently definition of $\ts$ implies that $\displaystyle{\min_{i \in {\cal M}_1} (1-\PCX{\mes_1=i}{y^{\ts}})\geq \lambda \delta}$. Thus using Markov inequality for $\Pe$ we get,
\begin{equation}
  \label{eq:tseq}
  \PX{\ts=\blx}\leq \tfrac{\Pe}{\lambda \delta}
\end{equation}
We use equation (\ref{eq:tseq})  to bound expected value of the entropy of first part of the message at time $\ts$ as follows,
\begin{align*}
  \EX{\HXa{\ts}}
&=  \EX{\HXa{\ts} \IND{\ts=\blx}}+  \EX{\HXa{\ts}\IND{\ts<\blx}}\\
&\leq \tfrac{\Pe}{\lambda \delta} \ln |{\cal M}_1|+ (\ln 2 + \delta  \ln |{\cal M}_1|)\\
&= \ln 2+ (\tfrac{\Pe}{\lambda \delta}+\delta) \ln |{\cal M}_1|)
\end{align*}
It has already been  established in, \cite{burna2},
\begin{equation}
 \label{eq:capacity-bound} \tfrac{\EX{\HXz-\HX{\ts}}}{\EX{\ts}} \leq \CX
\end{equation}
Thus,
\begin{align}
  \notag
\EX{\ts}
&\geq  \tfrac{1}{\CX} (\EX{\HXz-\HXa{\ts} -\htwo{\mes_1,Y^{\ts}}})\\
\label{eq:tsbound}
&\geq  \tfrac{1}{\CX} ( -\ln 2 + (1-\tfrac{\Pe}{\lambda \delta}-\delta) \ln |{\cal M}_1|)
\end{align}
Bound given in inequality (\ref{eq:tsbound}) specifies the time needed for getting a likely candidate, $\tilde{\mes}_1$. Like it was the case in  \cite{burna2}, remaining time is the time spend for confirmation. But unlike \cite{burna2} transmitter needs to convey also $\mes_2$ during this time. 

For each realization of $Y^{\ts}$  divide the message set into disjoint subsets, $\partit_0,\partit_1,\ldots,\partit_{|{\cal M}_2|}$ as follows,
\begin{align*}
  \partit_0
&=\{l:l\in {\cal M}, l=(i,j)~\mbox{ where}~i\neq \tilde{{\cal M}}_1(Y^{\ts}) \}&&\\
 \partit_j
&=\{l:l\in {\cal M}, l=(\tilde{{\cal M}}_1(Y^{\ts}),j)  \} && \forall j \in \{1,2,\ldots |{\cal M}_2| \}
\end{align*}
 where $\tilde{\mes}_1 (Y^{\ts})$ is the most likely message given $Y^{\ts}$. Furthermore let  the auxiliary-message, $\ames$, be the index of the set that $\mes$ belongs to, i.e. $\mes \in \partit_{\ames}$.

The decoder for the auxiliary message decodes the index of the decoded message at the decoding time $\tau$, i.e
\begin{equation*}
  \hat{\ames} (Y^{\blx})=j \Leftrightarrow  \hat{\mes} (Y^{\blx}) \in \partit_{j}.
\end{equation*}
With these definition  we have;
\begin{align*}
\PCX{\hat{\mes}(Y^{\blx}) \neq \mes}{Y^{\ts}} 
&\geq \PCX{\hat{\ames}(Y^{\blx}) \neq \ames}{Y^{\ts}} \\
\PCX{\hat{\mes}_1(Y^{\blx}) \neq \mes_1}{Y^{\ts}} 
&\geq  \PCX{\hat{\ames}(Y^{\blx}) \neq 0}{Y^{\ts},\ames=0}  \PCX{\ames=0}{ Y^{\ts}}.
\end{align*}
Now, we apply Lemma \ref{lem:con1}, which will be proved in  section \ref{subsec:fmdconverse}. To ease the notation we use following shorthand;
\begin{align*}
\APe{Y^{\ts}} 
&=\PCX{\hat{\ames}(Y^{\blx}) \neq \ames}{Y^{\ts}} \\
\APe{0,Y^{\ts}} 
&= \PCX{\hat{\ames}(Y^{\blx}) \neq 0}{Y^{\ts}, \ames=0} \\
\xi(Y^{\ts})
&= \PCX{\ames(Y^{\ts})=0}{ Y^{\ts}}.
\end{align*}
As a result of Lemma \ref{lem:con1}, for each realization of $y^{\ts} \in {\cal Y}^{\ts}$ such that $\ts <\blx$, we have
\begin{equation*}
(1-\xi(Y^{\ts})-\APe{Y^{\ts}}) \ln \tfrac{1}{\APe{0,Y^{\ts}}} \leq \ln 2+ \ECX{\blx-\ts}{Y^{\ts}}  \FX{\tfrac{\AHX{\ts}- \ln 2- \APe{Y^{\ts}} \ln |{\cal M}_2| }{\ECX{\blx-\ts}{Y^{\ts}}}}
\end{equation*}
By multiplying both sides of the inequality with $\IND{\ts < \blx}$, we get an expression that holds for all $Y^{\ts}$. 
\begin{align}
\notag
  \IND{\ts < \blx} (1-\xi(Y^{\ts})-\APe{Y^{\ts}}) 
& \ln \tfrac{1}{\APe{0,Y^{\ts}}}
\leq \\
&\IND{\ts < \blx}  \left[ \ln 2+ \ECX{\blx-\ts}{Y^{\ts}}
 \FX{\tfrac{\AHX{\ts}- \ln 2 - \APe{Y^{\ts}} \ln |{\cal M}_2| }{\ECX{\blx-\ts}{Y^{\ts}}}} \right]
\label{eq:compexp}
\end{align}
Now we take the expectation of both sides over $Y^{\ts}$. For the right hand side we have, 
\begin{align}
 R.H.S. 
\notag  &=\EX{ \left( \ln 2+ \ECX{\blx-\ts}{Y^{\ts}}  \FX{\tfrac{\AHX{\ts}- \ln 2 - \APe{Y^{\ts}} \ln |{\cal M}_2| }{\ECX{\blx-\ts}{Y^{\ts}}}} \right) \IND{\ts < \blx}}\\
\notag  
&\mathop{\leq}^{}  \ln 2+\EX{ \ECX{\blx-\ts}{Y^{\ts}}  \FX{\tfrac{\AHX{\ts}- \ln 2 - \APe{Y^{\ts}} \ln |{\cal M}_2| }{\ECX{\blx-\ts}{Y^{\ts}}}} \IND{\ts < \blx}}\\
\notag
&\mathop{\leq}^{(a)}  \ln 2+\EX{ {\blx-\ts}}  \FX{ \EX{\IND{\ts < \blx} \tfrac{\AHX{\ts}- \ln 2 - \APe{Y^{\ts}} \ln |{\cal M}_2| }{\EX{\blx-\ts}}}}\\
&\label{eq:bits-con-rhsold}
\mathop{\leq}^{(b)}   \ln 2+\EX{ {\blx-\ts}}  \FX{  \tfrac{\EX{\IND{\ts < \blx} \AHX{\ts}} - \ln 2- \Pe \ln |{\cal M}_2| }{\EX{\blx-\ts}}}
\end{align}
where $(a)$ follows the concavity of $\FX{\cdot}$ and Jensen's inequality when we interpret $\tfrac{\ECX{\blx-\ts}{Y^{\ts}}\IND{\ts < \blx}}{\EX{ {\blx-\ts}}}$ as probability distribution over ${\cal Y}^{\ts}$ and $(b)$ follows the fact that  $\FX{\cdot}$ is a decreasing function. 

Now we lower bound $\EX{\IND{\ts < \blx} \AHX{\ts}}$ in terms of $\EX{\HX{\ts}}$. Note that 
\begin{align*}
  \HX{\ts}
 &=\AHX{\ts}+ \PCX{\mes_1 \neq \tilde{\mes}_1 (Y^{\ts})}{Y^{\ts}} \h{\mes_1\neq \tilde{\mes}_1 (Y^{\ts}), Y^{\ts}}\\
&\leq \AHX{\ts}+\PCX{\mes_1 \neq \tilde{\mes}_1 (Y^{\ts})}{Y^{\ts}} \ln |{\cal M}_1| |{\cal M}_2|
\end{align*}
Furthermore   for all $Y^{\ts}$ such that $\tau>\ts$, $\PCX{\tilde{\mes}_1(Y^{\ts})\neq \mes_1}{Y^{\ts}}\leq \delta$. Thus
\begin{align}
  \EX{\IND{\ts < \blx} \AHX{\ts}}
\notag &\geq \EX{\IND{\ts < \blx} (\HX{\ts} - \delta \ln |{\cal M}_1| |{\cal M}_2|)}\\
\notag&= \EX{(1-\IND{\ts = \blx}) \HX{\ts}} - \delta \ln | {\cal M}_1| |{\cal M}_2|\\
\notag
&\geq \EX{\HX{\ts}} -\PX{\ts = \blx} \ln |{\cal M}_1| |{\cal M}_2|  - \delta \ln |{\cal M}_1| |{\cal M}_2|\\
\notag
&\mathop{\geq}^{(a)} \EX{\HX{\ts}} - (\tfrac{\Pe}{\lambda \delta}+ \delta) \ln |{\cal M}_1| |{\cal M}_2|\\
\label{eq:bits-con3}
&\mathop{\geq}^{(b)} (1 -\tfrac{\Pe}{\lambda \delta}- \delta) \ln |{\cal M}_1| |{\cal M}_2| -\CX \EX{\ts}
\end{align}
where $(a)$  follows from the inequality (\ref{eq:tseq}), $(b)$ follows from the inequality  (\ref{eq:capacity-bound}).  
Since $\FX{\cdot}$ is decreasing in its argument,  inserting  (\ref{eq:bits-con3}) in   (\ref{eq:bits-con-rhsold}) we get
\begin{equation}
\label{eq:bits-con-bound1}
 R.H.S.
 \leq
\ln 2+\EX{\blx-\ts}  \FX{\tfrac{\ln |{\cal M}_1| |{\cal M}_2|  \left(1- \tfrac{\Pe}{\lambda \delta} - \delta -\Pe \right) -
  \EX{\ts} \CX -\ln 2}
{\EX{\blx-\ts}}}
\end{equation}
Note that $\forall a>0,b>0,\CX>0$,
\begin{align*}
  \tfrac{d}{dx} (b-x) \FX{\tfrac{a-\CX x}{b-x}} \vert_{x=x_0}
&= -\FX{\tfrac{a-\CX x_0}{b-x_0}} - \left(\CX - \tfrac{a-\CX x_0}{b-x_0} \right) \tfrac{d}{dx} \FX{x} \vert_{x=\tfrac{a-\CX x_0}{b-x_0}}\\
&\mathop{\leq}^{(a)}  -\FX{\CX}
\end{align*}
where $(a)$ follows the concavity of $\FX{\cdot}$.
Thus upper bound given in equation (\ref{eq:bits-con-bound1}) is decreasing in  $\EX{\ts}$. Thus using the lower bound on $\EX{\ts}$, given in (\ref{eq:tseq}) we get,
\begin{equation}
R.H.S.
 \leq
\ln 2+\left( \EX{\blx} -   (1-\delta -\tfrac{\Pe}{\lambda \delta}) \tfrac{\ln |{\cal M}_1| }{\CX} +\tfrac{\ln 2}{\CX} \right)  \FX{
\tfrac{  \left(1- \tfrac{\Pe}{\lambda \delta} - \delta -\Pe \right) \ln |{\cal M}_2| -
\Pe \ln |{\cal M}_1| -2\ln 2}
{\EX{\blx}- (1-\delta- \tfrac{\Pe}{\lambda  \delta}) \tfrac{\ln |{\cal M}_1| }{\CX}+\tfrac{\ln 2}{\CX}}}
\label{eq:bits-con-rhs}
\end{equation}
Now let us consider the $L.H.S.$ we get by taking the expectation of the inequality given in (\ref{eq:compexp}).
\begin{align}
L.H.S.
\notag  &=
\EX{\IND{\ts < \blx} (1-\xi(Y^{\ts})-\APe{Y^{\ts}}) \ln \tfrac{1}{\APe{0,Y^{\ts}}} }\\
\notag   &\overset{(a)}{\geq}
\EX{\IND{\ts < \blx} (1-\xi(Y^{\ts})-\APe{Y^{\ts}})}
\ln \tfrac{ \EX{\IND{\ts < \blx} (1-\xi(Y^{\ts})-\APe{Y^{\ts}})} }
{\EX{\IND{\ts < \blx} (1-\xi(Y^{\ts})-\APe{Y^{\ts}})  \APe{0,Y^{\ts}}} } \\
\notag
&\overset{(b)}{\geq}
-e^{-1}+ \EX{\IND{\ts < \blx} (1-\xi(Y^{\ts})-\APe{Y^{\ts}})} \ln \tfrac{1}{\EX{\IND{\ts < \blx} (1-\xi(Y^{\ts})-\APe{Y^{\ts})}  \APe{0,Y^{\ts}}} } \\
\label{eq:lhsintexp1}
&\geq 
-e^{-1}+ \EX{\IND{\ts < \blx} (1-\xi(Y^{\ts})-\APe{Y^{\ts}})} \ln \tfrac{1}{\EX{\IND{\ts < \blx}  \APe{0,Y^{\ts}}}}
\end{align}
where $(a)$ follows log sum inequality and $(b)$ follows from the fact that $x\ln x \geq -e^{-1}$. 

Note that
\begin{align}
\EX{\IND{\ts < \blx} (1-\xi(Y^{\ts})-\APe{Y^{\ts}})}
\notag
&\geq \EX{\IND{\ts < \blx} (1-\xi(Y^{\ts}))}-\EX{\APe{Y^{\ts}}}\\
\notag
&\geq \EX{\IND{\ts < \blx}} (1- \delta)-\Pe\\
\label{eq:lhsintexp2}
&\geq 1- \tfrac{\Pe}{\lambda \delta }-\delta 
\end{align}
where in last step we have used the equation (\ref{eq:tseq}). Furthermore
\begin{align}
\EX{\IND{\ts < \blx}  \APe{0,Y^{\ts}}}  
\notag
&= \EX{\IND{\ts < \blx}\PCX{\hat{\mes}_1= \tilde{\mes}_1}{Y^{\ts}, \tilde{\mes}_1 \neq \mes_1}}\\
\notag
&\leq \tfrac{1}{\delta \lambda}
\EX{\IND{\ts < \blx}\PCX{\hat{\mes}_1= \tilde{\mes}_1}{Y^{\ts}, \tilde{\mes}_1 \neq \mes_1} \PCX{ \tilde{\mes}_1 \neq \mes_1}{Y^{\ts}}}\\
\label{eq:lhsintexp3}
&\leq  \tfrac{\Pe^{\mes_1}}{\delta \lambda}
\end{align}
Thus using equations (\ref{eq:lhsintexp1}), (\ref{eq:lhsintexp2}) and  (\ref{eq:lhsintexp3}) we get 
\begin{equation}
\label{eq:bits-con-lhs}
L.H.S.
\geq -e^{-1} - (1- \tfrac{\Pe}{\lambda \delta }-\delta)  \ln  \tfrac{\Pe^{\mes_1}}{\lambda \delta}
\end{equation}
Using the inequalities (\ref{eq:bits-con-rhs}), (\ref{eq:bits-con-lhs}) and choosing $\delta=\sqrt{\Pe}$ we get $\FEbr_{,\SC}  \leq \left( 1- \tfrac{r_\SC}{\CX}\right) \FX{\CX}$. Since $\FX{\CX}=\PCAP$ this  implies $\FEbr  (r)\leq \left( 1- \tfrac{r}{\CX}\right) \PCAP$.
\end{proof}
\subsection{Proof of  of Theorem \ref{thm:many}}
\subsubsection{Achievability}~

\begin{proof}
Proof is very similar to the achievability proof for Theorem \ref{thm:rf}. Choose a capacity achieving sequence $\SC'$ such that $\FEb_{,\SC'}=\PCAP$.  The capacity achieving sequence with feedback, $\SC$  uses $L$ elements of $\SC'$ as follows.

For the $\inx^{\mbox{th}}$ element of code $\SC$, transmitter uses the   $\lfloor \inx \cdot r_1 \rfloor  ^{\mbox{th}}$   element of $\SC'$ to send the first part of the message,  $\mes_1$.   In the remaining phases, $l \geq 2$ transmitter uses $\lfloor \inx \cdot r_l \rfloor  ^{\mbox{th}}$ element of $\SC'$. The special message of the code for phase $l$ is  allocated to the error event in previous phases.
\begin{align*}
  (\tilde{\mes}_1, \ldots, \tilde{\mes}_{(l-1)} ) \neq ( \mes_1, \ldots, \mes_{(l-1)} )   &\Rightarrow \ames_l=1 \qquad \forall l\\
  (\tilde{\mes}_1, \ldots, \tilde{\mes}_{(l-1)} ) = ( \mes_1, \ldots, \mes_{(l-1)} )   &\Rightarrow \ames_l=\mes_l+1 \qquad \forall l
\end{align*}
Thus ${\cal M}_1^{(\inx)}={\amesX}^{(\lfloor r \inx \rfloor)} $ and for all  $l \geq  1$ ${\cal M}_l^{(\inx)} \cup \{|{\cal M}_l^{(\inx)}|+1 \}={\amesX}^{(\lfloor r_l \inx \rfloor)} $.  If for all $l\in  \{2,3,\ldots,L\}$,  $\hat{\ames}_l\neq 1$, receiver decodes all parts of the information, else it declares an erasure. We skip the error analysis because it is essentially the same with   Theorem \ref{thm:rf}.
\end{proof}

\subsubsection{Converse}~

\begin{proof}
  We  prove the converse of Theorem \ref{thm:many} by contradiction. Evidently
\begin{equation*}
\max \{\Pe^{\mes_1}, \Pe^{\mes_2}, \ldots, \Pe^{\mes_j} \} \leq  \Pe^{\mes_1,\mes_2, \ldots,\mes_j} \leq \Pe^{\mes_1}+ \Pe^{\mes_2}+  \cdots +\Pe^{\mes_j} \qquad \forall j \in \{1,2,\ldots L\}
\end{equation*}
Thus if there exists a scheme that can reach an error exponent vector outside the region given in Theorem \ref{thm:many}, there is at least one $E_i$ such that  $E_i \geq (1 - \tfrac{\sum_{j=1}^i r_j} {\CX}) \PCAP$. Then we can have  two super messages as follows,
\begin{equation*}
  \ames_1=(\mes_1,\mes_2,\ldots,\mes_i)  \quad {\textrm{ and }}     \ames_2=(\mes_{i+1},\mes_{i+2},\ldots,\mes_l)
\end{equation*}
Recall that  $\Pe^{\mes_1} \leq \Pe^{\mes_2} \leq \cdots \leq \Pe^{\mes_l}$.  Thus this new code is a capacity achieving code, whose special bits have rate $r_{\SC'}$ and $\FEbr_{,\SC'} > \FEbr (r_{\SC'})$. This is contradicting with the Theorem \ref{thm:rf} we have already proved. Thus  all the achievable error exponent regions should lie in the region given in Theorem \ref{thm:many}.
\end{proof}

\subsection{Proof of  of Theorem \ref{thm:mdf}}
\subsubsection{Achievability:  $\FEmd \geq \PCAP$}~\\
Note that any fixed length block code without feedback, is also variable-length block code with feedback, thus $\FEmd \geq \Emd$. Using the capacity achieving sequence we have used in the achievability proof of Theorem \ref{thm:md}, we get $\FEmd \geq \PCAP$.

\subsubsection{Converse:  $\FEmd  \leq  \PCAP$}\label{subsec:fmdconverse}~\\
Now we  prove that even with feedback and variable decoding time, the best missed detection exponent of  a  single special message is less then or equal to $\PCAP$, i.e. $\FEmd \leq \PCAP$. Since the set of capacity achieving sequences is a subset of capacity achieving sequences with feedback and variable decoding time,  this also implies that $\Emd \leq \PCAP$.

Instead of directly proving the converse part of Theorem \ref{thm:mdf} we  first prove the following lemma.
\begin{lemma}
  \label{lem:con1}
For any variable length block code with feedback,  message set ${\cal M}$, initial entropy  $\HXz$ and  average error probability $\Pe$, the conditional  error probability of each message is lower bounded as follows,
\begin{equation}
\PCX{\hat{\mes} \neq i}{\mes=i}\geq  e^{-\tfrac{1}{1-\PX{\mes=i}-\Pe} \left( \FX{  \tfrac{\HXz-h(\Pe) - \Pe \ln (|{\cal M }|-1)} {\EX{\blx}}} \EX{\blx} +\ln 2 \right) }  \qquad \forall i
\end{equation}
where  $\FX{R}$ is given by the following optimization over probability distributions on ${\cal X}$
\begin{equation}
\label{eq:def-fx}
 \FX{R}
=\max_{\substack{ \alpha, x_1,x_2, P_X^{1} ,P_X^{2}:\\
 \alpha I(P_X^{1}, W_{Y|X})+ (1-\alpha) I(P_X^{2}, W_{Y|X})\geq R}}
 \alpha \KLD{ (P^{1}W)_Y(\cdot)}{ W(\cdot|x_1)} +(1- \alpha )\KLD{ (P^{2}W)_Y(\cdot)}{ W(\cdot|x_2)} 
\end{equation}
\end{lemma}
It is worthwhile remembering the notation we introduced previously that 
\begin{equation*}
  (P^{i}W)_Y(\cdot)=\sum_{j \in {\cal X}} P_X^{i}(j) W_{Y|X}(\cdot|j) \quad  \mbox{and} \quad   I(P_X^{i}, W_{Y|X})= \sum_{j \in{\cal X}, k \in {\cal Y}} P_{X}^{i}(j) W_{Y|X}(k|i) \ln \tfrac{W_{Y|X} (k|i)}{(P^{i}W)_Y(k)}
\end{equation*} 
First thing to note about Lemma \ref{lem:con1} is that it is not necessarily for  the case of uniform probability distribution on the message set ${\cal M}$. Furthermore as long as $\PX{\mes=i}<<1$ the lower bound on $\PCX{\hat{\mes}\neq i}{\mes=i}$ depends on the a priori probability distribution of the messages only through the entropy of it, $\HXz$.

In equation (\ref{eq:def-fx}) $\alpha$ is simply a time sharing variable, which allows us to use a $(x_i,P_X^{i})$ pair with low mutual information and high divergence together with another  $(x_i,P_X^{i})$ pair with high  mutual information and low divergence. As a result of Carath\'eodory's Theorem we see that time sharing between two points of the form $(x_i,P_X^{i})$ is sufficient for obtaining optimal performance, i.e. allowing time sharing between more than two points of the form  $(x_i,P_X^{i})$ will not improve the value of $\FX{R}$.

 Indeed for any $R \in [0,\CX]$ one can use the optimizing values of $\alpha$, $x_1$, $x_2$, $P_X^{1}$ and $P_X^{2}$ in a scheme like the one in Theorem \ref{thm:md} with  time sharing and prove that missed detection exponent of $\FX{R}$ is achievable for a reliable sequence of rate $R$. In that $\alpha$ determines  how long the input letter $x_1 \in {\cal X} $ is used  for the special message while $P_X^{1}$ is being used for the ordinary codewords. Furthermore arguments very similar to those of Theorem \ref{thm:mdf} can be used to prove no missed detection exponent higher than $\FX{R}$ is achievable for reliable sequences of rate $R$. Thus  $\FX{R}$ is the best exponent a message can get in a rate $R$ reliable sequence.

One can show that $\FX{R}$ is a concave function of $R$ over its support $[0,\CX]$. Furthermore  $\FX{0}=\DX$ and $\FX{\CX}=\PCAP$. Thus $\FX{R}$ is a concave strictly decreasing function of $R$ for $0\leq R\leq \CX$.

\begin{proofa}{ Lemma \ref{lem:con1}}
Recall that ${\cal G}(i)$ is the  decoding region for $\mes=i$ i.e. ${\cal G}(i)=\{y^{\blx}: \hat{\mes}(y^{\blx})=i \}$. Then as a result of data processing inequality for KL divergence we have
\begin{align}
  \EX{\ln \tfrac{\PX{Y^{\blx}}}{\PCX{Y^{\blx}}{\mes=i}}}
  \notag &\geq \PX{  {\cal G}(i)} \ln \tfrac{\PX{ {\cal G}(i)}}{ \PCX{{\cal G}(i) }{\mes=i}} + \PX{ \overline{ {\cal G}(i) }} \ln \tfrac{\PX{ \overline{{\cal G} (i)}}}{ \PCX{\overline{{\cal G} (i)}}{\mes=i}}\\
 \notag &\geq -h( \PX{  {\cal G} (i)})  + \PX{ \overline{ {\cal G} (i)}} \ln \tfrac{1}{ \PCX{\overline{{\cal G} (i)}}{\mes=i}}\\
  \label{eq:lowp}
&\geq -\ln 2 +  \PX{ \overline{ {\cal G} (i)}} \ln \tfrac{1}{ \PCX{\overline{{\cal G} (i)}}{\mes=i}}
\end{align}
where in the last step we have used, the fact that   $h(\PX{  {\cal G} (i)}) \leq \ln 2$.
In addition
\begin{align}
\notag  \PX{ \overline{ {\cal G} (i)}}
&\geq \PCX{\overline{ {\cal G} (i)}}{ \mes\neq i} \PX{ \mes\neq i}\\
\notag
&\geq \sum_{j\neq i} \PCX{ {\cal G} (j)}{ \mes= j } \PX{ \mes=j}\\
\label{eq:lowp1}
&\geq  \left( 1-\Pe- \PX{ \mes =i} \right).
\end{align}
Thus using the equations (\ref{eq:lowp}) and (\ref{eq:lowp1})  we get
\begin{equation}
\label{eq:lowp2}
 \PCX{\overline{{\cal G} (i)}}{\mes=i} \geq
e^{-\frac{1}{ 1-\Pe- \PX{ \mes =i}} \left(\ln 2 +\EX{\ln \frac{\PX{Y^{\blx}}}{\PCX{Y^{\blx}}{\mes=i}}} \right)}.
\end{equation}
Now we lower bound the error probability of  the special message by upper bounding $\EX{\ln \tfrac{\PX{Y^{\blx}}}{\PCX{Y^{\blx}}{\mes=i}}}$. For that let us consider the following stochastic sequence,
\begin{equation*}
  S_{n}=\ln \tfrac{\PX{Y^{n}}}{\PCX{Y^n}{\mes=i}}-\sum_{t=1}^{n} \ECX{\ln \tfrac{\PCX{Y_{t}}{Y^{t-1}}}{\PCX{Y_t}{\mes=i, Y^{t-1}}}}{Y^{t-1}}
\end{equation*}
Note that $\ECX{S_{n+1}}{Y^n}=  S_n$ and  since   $\min W_{i,j}=\lambda$ we have $\ECX{|S_{n+1}-S_n|}{Y^n}\leq  2  \ln  \tfrac{1}{\lambda} $. Thus $S_n$ is a martingale,  furthermore since $\EX{\blx}< \infty$ we can use \cite[Theorem 2 p 487]{shiryaev},   to get
\begin{equation*}
  \EX{S_{\blx}}=S_0=0.
\end{equation*}
Thus
\begin{equation}
\label{eq:explog1}   \EX{\ln \tfrac{\PX{Y^{\blx}}}{\PCX{Y^{\blx}}{\mes=1}}}
 =\EX{\sum_{t=1}^{\blx} \ECX{\ln \tfrac{\PCX{Y_{t}}{Y^{t-1}}}{\PCX{Y_t}{\mes=1, Y^{t-1}}}}{Y^{t-1}}}.
 \end{equation}
Note that 
\begin{equation}
\label{eq:explog2} 
  \ECX{\ln \tfrac{\PCX{Y_{t}}{Y^{t-1}}}{\PCX{Y_t}{\mes=1, Y^{t-1}}}}{Y^{t-1}}
\notag 
=\ECX{\ln \tfrac{\PCX{Y_{t}}{Y^{t-1}}}{W_{Y|X}(Y_t|\bx_t(1))}}{Y^{t-1}}.
\end{equation}
As a result of definition of $\FX{\cdot}$ given in equation (\ref{eq:def-fx}) we have,
\begin{equation}
\label{eq:explog3}
\ECX{\ln \tfrac{\PCX{Y_{t}}{Y^{t-1}}}{\PCX{Y_t}{\mes=1, Y^{t-1}}}}{Y^{t-1}}
\leq   \FX{ \CMI{X_t}{Y_t}{Y^{t-1}}}
\end{equation}
where $\CMI{X_t}{Y_t}{Y^{t-1}}$ is given by\footnote{ Note that unlike the conventional definition of conditional mutual information, $ \CMI{X_t}{Y_t}{Y^{t-1}}$  is not averaged over the conditioned random variable $Y^{t-1}$.}
  \begin{equation*}
    \CMI{X_t}{Y_t}{Y^{t-1}}=\ECX{\ln \tfrac{\PCX{X_t,Y_t}{Y^{t-1}}}{\PCX{X_t}{Y^{t-1}} \PCX{Y_t}{Y^{t-1}}}}{Y^{t-1}}
  \end{equation*}
Given $Y^{t-1}$ random variables  $\mes-X_t-Y_t$ forms a Markov chain. Thus 
\begin{equation}
  \label{eq:explog4}
  \CMI{X_t}{Y_t}{Y^{t-1}} \geq \CMI{\mes}{Y_t}{Y^{t-1}}.
\end{equation}
Since $\FX{\cdot}$ is a decreasing function, equations (\ref{eq:explog1}),  (\ref{eq:explog3}) and (\ref{eq:explog4})  lead to 
\begin{equation}
\label{eq:explog}   
\EX{\ln \tfrac{\PX{Y^{\blx}}}{\PCX{Y^{\blx}}{\mes=1}}}
 \leq \EX{\sum_{t=1}^{\blx} 
\FX{\CMI{\mes}{Y_t}{Y^{t-1}}}}
 \end{equation}
Note that
\begin{align}
\EX{\sum_{t=1}^{\blx} \FX{ \CMI{\mes}{Y_t}{Y^{t-1}}}}
\notag
&= \EX{\blx  \sum_{t=1}^{\blx}  \tfrac{1}{\blx}  \FX{ \CMI{\mes}{Y_t}{Y^{t-1}}}}\\
\notag
&\mathop{\leq}\limits^{(a)}  \EX{\blx \FX{ \sum_{t=1}^{\blx}  \tfrac{1}{\blx}   \CMI{\mes}{Y_t}{Y^{t-1}}}}\\
\notag
&=\EX{\blx}\EX{ \tfrac{\blx}{\EX{\blx}} \FX{ \sum_{t=1}^{\blx}  \tfrac{1}{\blx}   \CMI{\mes}{Y_t}{Y^{t-1}}}}\\
\notag
&\mathop{\leq}\limits^{(b)}   \EX{\blx}  \FX{\EX{  \tfrac{\blx}{\EX{\blx}} \sum_{t=1}^{\blx}  \tfrac{1}{\blx}   \CMI{\mes}{Y_t}{Y^{t-1}}}}\\
\label{eq:lowp5}
&=  \EX{\blx}  \FX{\EX{ \tfrac{\sum_{t=i}^{\blx}  \CMI{\mes}{Y_t}{Y^{t-1}}}{\EX{\blx}}}}
\end{align}
where in both $(a)$ and $(b)$ we use the the concavity of the $\FX{\cdot}$ function together with Jensen's inequality. Thus using equations (\ref{eq:lowp2}), (\ref{eq:explog}) and (\ref{eq:lowp5}) we get, 
\begin{equation*}
\PCX{\hat{\mes} \neq i}{\mes=i} \geq e^{-\tfrac{1}{ 1-\Pe- \PX{ \mes =i}} \left(\FX{ \tfrac{\EX{\sum_{t=i}^{\blx}  \CMI{\mes}{Y_t}{Y^{t-1}}}}{\EX{\blx}}} \EX{\tau} +\ln 2 \right) }
\end{equation*}
Since $\FX{R}$ is decreasing in $R$, the only thing we are left to show is that 
\begin{equation}
  \label{eq:lastcon}
  \EX{\sum_{t=i}^{\blx}  \CMI{\mes}{Y_t}{Y^{t-1}}} \geq \HXz- h(\Pe) - \Pe \ln (|{\cal M }|-1)
\end{equation}
For that consider the stochastic sequence,
\begin{equation*}
  V_{n}= \HX{n}+ \sum_{t=1}^{n}\CMI{\mes}{Y_t}{Y^{t-1}}.
\end{equation*}
Clearly $\ECX{V_{n+1}}{Y^n}=V_{n}$ and $\EX{|V_n|} < \infty$, thus $\{V_n\}$ is a martingale. Furthermore $\ECX{|V_{n+1}-V_n|}{Y^n} \leq K $ and $\EX{\blx} < \infty$ thus using a version of Doob's optional stopping theorem, \cite[Theorem 2 p 487]{shiryaev},  we get,
\begin{align}
 V_0
\notag &= \EX{V_{\blx}}  \\
\label{eq:veq}  &=\EX{\HX{\blx}} + \EX{\sum_{t=1}^{\blx} \CMI{\mes}{Y_t}{Y^{t-1}}}.
\end{align}
  One can write Fano's inequality as follows,
\begin{equation*}
\HX{\blx} \leq  h \left(\PCX{\hat{\mes} (Y^{\blx}) \neq \mes}{Y^{\blx}}  \right) +
\PCX{\hat{\mes} (Y^{\blx}) \neq \mes}{Y^{\blx}}  \ln (|{\cal M }|-1).
\end{equation*}
 Consequently
\begin{equation*}
\EX{\HX{\blx}} \leq  \EX{h \left( \PCX{\hat{\mes} (Y^{\blx}) \neq \mes}{Y^{\blx}}  \right)} +\EX{ \PCX{\hat{\mes} (Y^{\blx}) \neq \mes}{Y^{\blx}}} \ln (|{\cal M }|-1).
\end{equation*}
Using the concavity  of binary  entropy,
\begin{align}
\label{eq:fano} \EX{\HX{\blx}} \leq h(\Pe) + \Pe \ln (|{\cal M }|-1).
\end{align}
Using equation (\ref{eq:veq}) together with equation  (\ref{eq:fano}) we get the desired condition given in the equation (\ref{eq:lastcon}).
\end{proofa}
Above proof is for encoding schemes which does not have any randomization (time sharing), but same ideas can be used to establish the exact same result for  general variable length block codes with randomization. Now we are ready to prove the converse part of the Theorem \ref{thm:mdf}.\\
\begin{proofa}{Converse part of Theorem \ref{thm:mdf}}
In order to prove $\FEmd \leq \PCAP$, first note that for capacity achieving sequences we consider $\PX{ \mes =i}= \tfrac{1}{|{\cal M}^{(\inx)}|}$. Thus
\begin{equation}
-\tfrac{\ln (\Pe^{\mes} (i))^{(\inx)}}{\EX{\blx^{(\inx)}}}  \leq \tfrac{1}{ 1-\Pe^{(\inx)}-  \tfrac{1}{|{\cal M}^{(\inx)}|}}
 \left( \FX{ \tfrac{ \ln |{\cal M}^{(\inx)}| -  h(\Pe{(\inx)}) - \Pe^{(\inx)} \ln (|{\cal M}{(\inx)}|-1)} {\EX{\blx^{(\inx)}}}}+\tfrac{\ln 2} {\EX{\blx^{(\inx)}}} \right).
\end{equation}
Thus for any capacity achieving sequence with feedback,
\begin{equation*}
\lim_{\inx \rightarrow \infty} -\tfrac{\ln (\Pe^{\mes} (i))^{(\inx)}}{\EX{\blx^{(\inx)}}}  \leq \FX{\CX} =\PCAP.
\end{equation*}
\end{proofa}

\subsection{Proof of  of Theorem \ref{thm:msru}}
In this subsection we will show how the strategy for sending a special bit can be combined with the Yamamoto-Itoh strategy when many special messages demand a missed-detection exponent. However unlike previous results about capacity achieving sequences, Theorems \ref{thm:bitf}, \ref{thm:rf}, \ref{thm:many}, \ref{thm:mdf}, we will have and additional uniform delay assumption.

We will restrict ourself to uniform delay capacity achieving sequences.\footnote{ Recall that for any reliable variable length block code with feedback $\Gamma$ is defined as $\Gamma= \tfrac{\max_{i \in {\cal M}} \ECX{\tau}{\mes=i}}{\EX{\tau}} $ and uniform delay reliable sequences  are the ones    that satisfy  $\lim_{\substack{\inx \rightarrow \infty}} \Gamma_{\SC}^{(\inx)} =1$. } Clearly capacity
achieving sequences in general need not to be uniform delay. Indeed many messages, $i \in {\cal M}$, can get an expected delay, $\ECX{\blx}{\mes=i}$ much larger than the average delay, $\EX{\blx}$. This in return can decrease the error probability of these messages.  The potential drawback of such codes, is that their average delay is sensitive to assumption of messages being chosen according to a uniform probability distribution. Expected decoding time, $\EX{\blx}$, can increase a lot if  the code is used in a system in which the  messages are not chosen uniformly.

It is worth emphasizing that all previously discussed  exponents (single message exponent $\FEmd$,  single bit exponent $\FEb$, many bits exponent $\FEb(r)$ and achievable multi-layer exponent regions) remain unchanged whether or not this uniform delay constraint is imposed. Thus the flexibility to provide different expected delays to different messages does not improve those exponents.

However, this is not true for the message-wise \uep~ with exponentially many messages. Removing the uniform delay constraint can considerably enhance the protection of special messages at rate higher than $(1-\tfrac{\PCAP}{\DX})\CX$. Indeed one can make the exponent of all special messages, $\PCAP$. The flexibility of providing more resources (decoding delay) to special messages achieves this enhancement. However, we will not discuss those cases in this article and stick to uniform delay codes.

\subsubsection{Achievability: $\FEmd (r) \geq  \min  \{\PCAP, (1-\tfrac{r}{\CX}) \DX  \} $}\label{sec:femdach}~\\
The optimal scheme here reverses the trick for achieving $\FEb$: first a special bit tells to the receiver whether the message being transmitted is special one or not. After the decoding of this bit the message itself is transmitted. This further emphasizes how feedback connects bit-wise and message-wise \uep, when used with variable decoding time.

\begin{proof}
Like all the previous achievability results, we construct a capacity achieving sequence, $\SC$, with the desired asymptotic behavior. A sequence of multi phase fixed length errors and erasures codes, $\SC'$ is used as the building block of $\SC$. Let us consider the $\inx^{\mbox{th}}$ member of  $\SC'$. In the first phase transmitter sends one of the two input symbols with distinct output distributions for $\lfloor  \sqrt{\inx} \rfloor$ time units in order to tell whether $\mes  \in {\cal M}_s^{(\inx)}$ or not. Let $b$ be $b=\IND{\mes \in {\cal M}_s^{(\inx)}}$.Then, as it was mentioned in subsection \ref{sec:bitfach}, with a threshold decoding we can  achieve
\begin{equation}
\label{eq:exp-manymes1}
 \PCX{\hat{b} \neq 1}{b=1} =  \PCX{\hat{b} \neq 0}{b=0} \leq e^{-\sqrt{\inx} \mu} \qquad \mbox{where } \mu>0.
\end{equation}
Actual value of $\mu$ is not important for us, we are  merely interested in an upper bound vanishing with increasing $\inx$.

In the second phase one of two length $\inx$ codes is  used depending on $\hat{b}$.
\begin{itemize}
\item {If $\hat{b}=0$, in the second phase, transmitter  uses the $\inx^{\mbox{th}}$ member of a capacity achieving sequence, $\SC''$ such that $\Eb_{,\SC''} =\PCAP$. We know that such a sequence exists because of Theorem \ref{thm:md}. The message, $\ames$ of the $\SC''$ is  determined  using the following mapping
\begin{align*}
  \mes \in {\cal M}_s &\Rightarrow \ames=1 \\
  \mes \notin {\cal M}_s &\Rightarrow \ames=\mes-|{\cal M}_s|+1
\end{align*}
At the end of the second phase, receiver  decodes $\ames$. If $\hat{\ames}=1$, then receivers declares an erasure, $\tilde{\mes}=\era$. If $\hat{\ames}\neq 1$,  then  $\hat{\mes}=\tilde{\mes}=\hat{\ames}+|{\cal M}_s|-1$.}

\item{If $\hat{b}=1$, transmitter uses a two phase code with errors and erasures in the second phase, like the one described by Yamamoto and Itoh in \cite{itoh}. The two phases of this code are called communication and control phases, respectively. 

In communication phase transmitter uses $\lceil r \inx \rceil^{\mbox{th}}$ member of a capacity achieving sequence, $\SC''$ with $\Eb_{,\SC''} =\PCAP$, to convey its message, $\ames$. The auxiliary message $\ames$ is determined as follows,
\begin{align*}
  \mes \notin {\cal M}_s &\Rightarrow \ames=1\\
  \mes \in {\cal M}_s &\Rightarrow \ames=\mes+1
\end{align*}
The decoded message of the  $\lceil r \inx \rceil^{\mbox{th}}$ member of $\SC''$ is called the tentative decision of communication phase and denoted by $\tilde{\ames}$.
In the control phase,
\begin{itemize}
\item if $\tilde{\ames}=\ames$  tentative decision is confirmed by  sending accept symbol $\xa$ for $\ell(\inx)=\inx-\lceil  \tfrac{r}{\CX}  \inx \rceil $ time units.
\item if $\tilde{\ames}\neq \ames$ tentative decision is  rejected by sending reject symbol $\xd$ for  $\ell(\inx)=\inx-\lceil  \tfrac{r}{\CX}  \inx \rceil $ time units. 
\end{itemize}
where $\xa$ and $\xd$ are the maximizers in the following optimization problem.
\begin{equation*}
  \DX=\max_{i,j} \KLD{W_{Y|x}(\cdot|i)}{W_{Y|X} (\cdot|j)}=  \KLD{W_{Y|x}(\cdot|\xa)}{W_{Y|X} (\cdot|\xd)}
\end{equation*}
If the output sequence in last $\inx-\lceil  \tfrac{r}{\CX}  \inx \rceil $ time steps is typical with $W_{Y|X}(\cdot|\xa)$ then $\hat{\ames}=\tilde{\ames}$ else erasure is declared for $\ames$. Note that the total probability of  $W_{Y|X}(\cdot|\xa)$ typical sequences are less than $e^{-\ell(\inx)(\DX-\delta_{\ell(k)})}$ when $\tilde{\ames}\neq \ames$ and more than $1-\delta_{\ell(\inx)}$ when  $\tilde{\ames}= \ames$ where $\displaystyle{\lim_{\ell(\inx) \rightarrow \infty} \delta_{\ell(\inx)}=0}$, \cite[Corrollary 1.2, p19]{CK}.

If $\hat{\ames}=\era$ or if $\hat{\ames}=1$ then receiver declares  erasure for $\mes$, $\tilde{\mes}=\era$. If $\hat{\ames}\in \{2,3,\ldots,|{\cal M}_s|+1\}$, then $\hat{\mes}=\tilde{\mes}=\hat{\mes}-1$.}
\end{itemize}
Now we can calculate the error and erasure probabilities of the two phase fixed length block code. Let us denote the erasures by $\tilde{\mes}=\era$ for each $\inx$.

For $i \in {\cal M}_s$ using the equation (\ref{eq:exp-manymes1}) and Bayes rule we get
\begin{align}
\label{eq:manymes1}  
\PCX{\tilde{\mes}= \era}{\mes=i} 
&\leq e^{-\mu \sqrt{\inx}}+(\Pe_{,\SC'}^{( \inx-\ell(\inx))}+\delta_{\ell(\inx)})\\
\label{eq:manymes2}
\PCX{\tilde{\mes} \neq i , \tilde{\mes} \neq \era}{\mes=i} 
&\leq e^{-\mu \sqrt{\inx}}  \Pe_{\SC'}^{\inx}(1)+ \Pe_{,\SC'}^{(\inx -\ell(\inx))} e^{-\ell(\inx)(\DX-\delta_{\ell(\inx)})} .
\end{align}
For $i \notin {\cal M}_s$ using the equation (\ref{eq:exp-manymes1}) and Bayes rule we get
\begin{align}
\label{eq:manymes3}  
\PCX{\tilde{\mes}= \era}{\mes=i} &\leq  e^{-\mu \sqrt{\inx}} + \Pe_{,\SC'}^{(\inx)} \\
\label{eq:manymes4}
\PCX{\tilde{\mes} \neq i , \tilde{\mes} \neq \era}{\mes=i} &\leq e^{-\mu \sqrt{\inx}} +   \Pe_{,\SC'}^{(\inx)}  .
\end{align}
Whenever $\tilde{\mes}=\era$ than transmitter and receiver  try to send the message once again from scratch using same strategy. Then  for any $i \in {\cal M}$
\begin{align}
\label{eq:manymes5} 
\PCX{\hat{\mes} \neq i } {\mes=i} &=\tfrac{\PCX{\tilde{\mes} \neq i , \tilde{\mes} \neq \era} {\mes=i}}{1-\PCX{\tilde{\mes}= \era}{\mes=i}}  \\
\label{eq:manymes6} 
\ECX{\blx}{\mes=i} &=\tfrac{ \inx+\sqrt{\inx} }{1-\PCX{\tilde{\mes}= \era}{\mes=i}}
\end{align}
Using equations (\ref{eq:manymes1}), (\ref{eq:manymes2}), (\ref{eq:manymes3}), (\ref{eq:manymes4}), (\ref{eq:manymes5}) and (\ref{eq:manymes6})  we conclude that that $\SC$ is capacity achieving sequence such that 
\begin{align*}
  \lim_{\inx \rightarrow \infty} -\tfrac{\ln \max_{i \in {\cal M}_s}\PCX{\tilde{\mes} \neq i , \hat{\mes} \neq \era}{\mes=i} }{\EX{\blx} }&= \min\{\PCAP, (1-\tfrac{r}{\CX})\DX \}\\
  \lim_{\inx \rightarrow \infty} \tfrac{\ln |{\cal M}_s^{(\inx)}|}{\EX{\blx}} &=r
\end{align*}
\end{proof}

\subsubsection{Converse: $\FEmd (r) \leq  \min \{\PCAP, (1-\tfrac{r}{\CX}) \DX  \} $}~

\begin{proof}
  Consider any uniform delay capacity achieving sequence, $\SC$.  Note that by excluding all $i \notin {\cal M}_s^{(\inx)}$ we get a reliable sequence, $\SC'$ such that
\begin{align*}
\Pe^{'(\inx)} &\leq \SPCX{\inx}{\hat{\mes} \neq \mes} {\mes \in {\cal M}_s } \\
\EX{\tau^{'(\inx)}} &\leq \Gamma^{(\inx)}\EX{\tau^{(\inx)}}
\end{align*}
Thus
\begin{equation*}
  \tfrac{- \ln \PCX{\hat{\mes} \neq \mes} {\mes \in {\cal M}_s }^{(\inx)}}{\EX{\tau^{(\inx)}}} \leq -  \tfrac{\ln \Pe^{'(\inx)}}{\EX{\tau^{'(\inx)}}}   \Gamma^{(\inx)}
\end{equation*}
Consequently  $\FEmd (r) \leq  (1-\tfrac{r}{\CX}) \DX $. Similarly by excluding all but one of the elements of ${\cal M}_s$ we can prove that  $\FEmd (r) \leq  \PCAP$, using Theorem \ref{thm:mdf} and uniform delay condition.
\end{proof}

\section{Avoiding False Alarms: Proofs}\label{sec:proofsfa}

\subsection{Block Codes without Feedback: Proof of  Theorem \ref{thm:fa}}

\subsubsection{Lower Bound: $\Efa \geq \Efal$}~

\begin{proof} 
As a result of the coding theorem \cite[Ch. 2 Corollary 1.3, page 102   ]{CK} we know that there exits a reliable sequence $\SC'$ of fixed composition codes whose rate is $\CX$ and whose $\BLX^{\mbox{th}}$ elements composition $P_X^{(\BLX)}$ satisfies,
\begin{equation*}
  \sum_{i \in {\cal X}} |P_X^{(\BLX)}(i)-P_X^{*}(i)| \leq \sqrt[4]{\tfrac{1}{\BLX}}.
\end{equation*}
We  use the codewords of the $\BLX^{\mbox{th}}$ element of $\SC'$ as the codewords of the ordinary messages in the $\BLX^{\mbox{th}}$  code in  $\SC$. For the special message we  use a length-$\BLX$ repetition sequence $\bx^{\BLX}(1)=(\xfl,\xfl,\cdots,\xfl)$.

The decoding region for the special message is essentially the bare minimum. We include the typical channel outputs within the decoding region of the special message to ensure small missed detection probability for the special message, but we exclude all other output sequence $y^{\BLX}$.
\begin{equation*}
\DEC{1}=\{y^{\BLX}:  \sum_{i \in {\cal Y}} |\typey{y^{\BLX}}(i)-W_{Y|X}(i|\xfl)| \leq \sqrt[4]{1/\BLX} \}  
\end{equation*}
Note that this definition of $\DEC{1}$ itself ensures that special message is transmitted reliably whenever it is sent, $ \displaystyle{ \lim_{\BLX \rightarrow \infty } \SPCX{\BLX}{\hat{\mes} \neq 1}{\mes = 1}=0}$.

The decoding regions of the ordinary messages, $j=\{2,3,\ldots {\cal M}^{(\BLX)}\}$, is the intersection of the corresponding decoding region in $\SC'$ with the complement of $\DEC{1}$. Thus the fact that $\SC'$ is a reliable sequence  implies that, 
\begin{equation*}
 \lim_{\BLX \rightarrow \infty } \SPCX{\BLX}{y^{\BLX} \in  \displaystyle{\bigcup_{j \notin  \{1,i\}} \DEC{j}} }{\mes=i} =0
\end{equation*}
Consequently we  have reliable communication for ordinary messages as long as $\displaystyle{\lim_{\BLX \rightarrow \infty} \SPCX{\BLX}{\DEC{1}}{\mes =j } =0}$, $\forall j \neq 1$. But we prove a much stronger result to ensure that $\SPCX{\BLX}{\hat{\mes} =1}{\mes \neq 1}$ is  decaying fast enough. Before doing that let us note that in the second stage of the decoding, when we are choosing a message among the ordinary ones,   ML decoder can be used instead of  the decoding rule of the original code. Doing that will only decrease the average error probability.

 Note the probability of a $V$-shell of a  message $i$ is equal to,  
\begin{equation*}
  \SPCX{\BLX}{\shell{V}{i}}{\mes=i}= e^{-\BLX\CKLD{V_{Y|X}(\cdot|X)}{W_{Y|X}(\cdot|X)}{P_X^{(\BLX)}}}
\end{equation*}
Note that also that $\DEC{1}$ can be written as the union of $V$-shells of a message $i$ as follow.
\begin{equation*}
\DEC{1}=\bigcup_{V_{Y|X} \in {\cal V}^{(\BLX)} }  \shell{V}{i}   \qquad \forall  i \neq 1
\end{equation*}
where ${\cal V}^{(\BLX)}= \{V_{Y|X}:  \sum_{j} |\sum_{k}V_{Y|X} (j|k) P_{X}^{\BLX} (k) -W_{Y|X}(j|\xfl)| \leq  \sqrt[4]{1/\BLX} \}$. Note that since there are at most $(1+\BLX)^{|{\cal X}| |{\cal Y}|}$ different conditional types.
\begin{equation*}
\SPCX{\BLX}{\DEC{1}}{\mes=i} \leq  (1+\BLX)^{|{\cal X}| |{\cal Y}|}  \max_{V_{Y|X} \in {\cal V}^{(\BLX)}}  \PCX{\shell{V}{i}}{\mes=i}
\end{equation*}
  Thus for all $i>1$
 \begin{equation*}
 \lim_{\BLX \rightarrow \infty} \tfrac{-\ln \SPCX{\BLX}{\DEC{1}}{\mes=i}}{\BLX}= \min_{
V_{Y|X}:\ \sum_j P_X^*(j)V_{Y|X}(\cdot|j)=W_{Y|X}(\cdot|\xfl)}
\CKLD{V_{Y|X}(\cdot|X)}{W_{Y|X}(\cdot|X)}{P_X^*}
\end{equation*}
\end{proof}

\subsubsection{Upper Bound: $\Efa \leq \Efau$}~

\begin{proof} 
As a result of data processing inequality for KL divergence we have
\begin{align}
\notag \sum_{y^{\BLX} \in {\cal Y}^{\BLX}}
 \PCX{y^{\BLX}}{\mes=1}  \ln \tfrac{\PCX{y^{\BLX}}{\mes=1}}{\PCX{y^{\BLX}}{\mes \neq 1}}
&\geq
\PCX{\DEC{1}}{M=1}  \ln \tfrac{\PCX{\DEC{1}}{M=1}}{\PCX{\DEC{1}}{M \neq 1}}
\PCX{\overline{{\DEC{1}}}}{M=1}  \ln \tfrac{ \PCX{\overline{{\DEC{1}}}}{M=1}}{\PCX{\overline{\DEC{1}}}{M \neq 1}}\\
&\geq - \ln 2 - \PCX{\DEC{1}}{M=1}  \ln \PCX{{\DEC{1}}}{M \neq 1}
\label{eq:fal-dat}
\end{align}
Using the convexity of the KL divergence we  get
\begin{align}
\sum_{y^{\BLX} \in {\cal Y}^{\BLX}}
 \PCX{y^{\BLX}}{\mes=1}  \ln \tfrac{\PCX{y^{\BLX}}{\mes=1}}{\PCX{y^{\BLX}}{\mes \neq 1}}
\notag &\leq
\sum_{i =2}^{|{\cal M}|} \tfrac{1}{|{\cal M}|-1} \sum_{y^{\BLX} \in {\cal Y}^{\BLX}}
 \PCX{y^{\BLX}}{\mes=1}  \ln \tfrac{\PCX{y^{\BLX}}{\mes=1}}{\PCX{y^{\BLX}}{\mes =i}}\\
\notag &=
\sum_{i =2}^{|{\cal M}|} \tfrac{1}{|{\cal M}|-1} \sum_{y^{\BLX} \in {\cal Y}^{\BLX}}
 \PCX{y^{\BLX}}{\mes=1}  \sum_{k=1}^{\BLX}  \ln \tfrac{\PCX{y_k}{\mes=1, y^{k-1}}}{\PCX{y_k}{\mes =i,y^{k-1}}}\\
\label{eq:false-aleq}
&=\sum_{k=1}^{\BLX} \sum_{i=2}^{|{\cal M}|} \tfrac{1}{|{\cal M}|-1}
\KLD{W_{Y|X}(\cdot|\bx_k(1))}{W_{Y|X}(\cdot|\bx_k(i))} 
\end{align}
where $\bx_k(i)$ denotes the input letter for codeword of message $i$, at time $k$. 

Let us denote the empirical distribution of the $\bx_k(i)$ for time $k$, by $P_{X_k}$. 
\begin{equation*}
  P_{X_k} (i) =\tfrac{\sum_{j \in {\cal M}} \IND{\bx_k(j)=i}}{|{\cal M}|} \quad \forall i \in {\cal X}
\end{equation*}
Using equation (\ref{eq:fal-dat}) and (\ref{eq:false-aleq}) we get
\begin{equation}
 \PCX{{\DEC{1}}}{M \neq 1} \geq e^{-\tfrac{1}{\PCX{\DEC{1}}{M=1}} \left(\tfrac{|{\cal M}|}{|{\cal M}|-1}\sum_{k} \CKLD{W_{Y|X}(\cdot|\bx_k(1))}{W_{Y|X}(\cdot|X_k)}{P_{X_k}} +\ln 2 \right) } \label{eq:false-alarm-fin}
\end{equation}
We show below that for all capacity achieving codes, almost all of the $k$'s has a $P_{X_k}$ which is essentially equal to $P_X^{*}$. For doing that let us first define the set $\DIST{\epsilon}$ and $\delta(\epsilon)$
\begin{equation*}
\DIST{\epsilon} \DEF \{P_X: I(P_X,W_{Y|X}) \geq \CX- \epsilon \} \quad \mbox{and}   \quad \delta(\epsilon) \DEF \max_{P_X \in \DIST{\epsilon}} \sum_{i} |P_X(i)-P_X^{*}(i)| 
\end{equation*}
Note that $\displaystyle{\lim_{\epsilon \rightarrow 0} \delta(\epsilon) =0}$. As a result of Fano's inequality we have,
\begin{equation}
\label{eq:fa-fano1}
 \MI{\mes}{Y^{\BLX}}\geq \BLX R^{(\BLX)} (1-\Pe)-\ln 2 
\end{equation}
On the other hand using standard manipulations on mutual information we  get
\begin{align}
\notag
\MI{\mes}{Y^{\BLX}} 
&=\sum_{k=1}^{\BLX} I(P_{X_k}, W_{Y|X}) \\
\label{eq:fa-fano2}
&\leq \CX \BLX - \epsilon \sum_{k=1}^{\BLX} \IND{P_{X_k} \notin \DIST{\epsilon}}
  \end{align}
Using equation (\ref{eq:fa-fano2})  in equation (\ref{eq:fa-fano1}) we get,
\begin{equation*}
 \sum_{k=1}^{\BLX}  \IND{P_{X_k} \notin \DIST{\epsilon}} \leq  \BLX \tfrac{(\CX-R^{(\BLX)} (1-\Pe) - \ln 2 /\BLX)}{\epsilon} 
\end{equation*}
Let $\epsilon{(\BLX)}$ be $\epsilon{(\BLX)} =\sqrt{\CX-R^{(\BLX)} (1-\Pe) - \tfrac{\ln 2}{\BLX}}$, then $\displaystyle{\lim_{\BLX \rightarrow \infty} \epsilon{(\BLX)}=0}$ and  
\begin{equation}
 \label{eq:false-alarm-fin1}
 \sum_{k=1}^{\BLX}  \IND{P_{X_k} \notin \DIST{\epsilon^{(\BLX)}}} \leq  \BLX \epsilon^{(\BLX)}.
\end{equation} 
Note for any $P_X \in \DIST{\epsilon^{(\BLX)}}$ we have
\begin{align}
\notag
\CKLD{W_{Y|X}(\cdot|x_k(1))}{W_{Y|X}(\cdot|X_k)}{P_{X}}
&\leq  \CKLD{W_{Y|X}(\cdot|x_k(1))}{W_{Y|X}(\cdot|X)}{P_{X}^*} + \delta(\epsilon^{(\BLX)}) \DX\\
\label{eq:false-alarm-fin2}
&\leq  \Efau+ \delta(\epsilon^{(\BLX)}) \DX 
  \end{align}
where $\Efau= \max_{i \in {\cal X}}  \CKLD{W_{Y|X}(\cdot|i)}{W_{Y|X}(\cdot|X)}{P_{X}^*} $

Using equations (\ref{eq:false-alarm-fin1}) and (\ref{eq:false-alarm-fin2}) 
\begin{equation*}
\sum_{k} \CKLD{W_{Y|X}(\cdot|x_k(1))}{W_{Y|X}(\cdot|X_k)}{P_{X_k}}
\leq \BLX (\Efau+ \delta(\epsilon^{(\BLX)}) \DX +\epsilon^{(\BLX)} \DX)
\end{equation*}
Inserting this in equation (\ref{eq:false-alarm-fin}) we get 
\begin{equation*}
\lim_{\BLX \rightarrow \infty} \left(\tfrac{- \ln \SPCX{\BLX}{{\DEC{1}}}{M \neq 1}}{\BLX} \right)\leq \Efau
\end{equation*}
\end{proof}

\subsection{Variable Length Block Codes with Feedback: Proof of Theorem \ref{thm:faf}}
\subsubsection{Achievability: $\FEfa \geq \DX$ }~

\begin{proof}
 We  construct a capacity achieving sequence with feedback, $\SC$, by using a construction like the one we have for $\FEmdr{r}$. In fact, this scheme achieves the false alarm exponent simultaneously with the best missed detection exponent, $\PCAP$, for the special message. 

We use a  fixed length multi-phase errors and erasure code as the building block for the $\inx^{\mbox{th}}$ member of $\SC$. In the first phase, $b=\IND{\mes=1}$ is conveyed using a length $\lceil  \sqrt{\inx} \rceil$ repetition code, like we did in subsections \ref{sec:bitfach} and \ref{sec:femdach}. Recall that
 \begin{equation}
\label{eq:fa-exp}
 \PCX{\hat{b} \neq 1}{b=1} =  \PCX{\hat{b} \neq 0}{b=0} \leq e^{-\mu \sqrt{\inx}}
  \qquad  \mu>0
\end{equation}
In the second phase one of the two  length $\inx$ codes is  used depending on $\hat{b}$.
\begin{itemize}
\item {If $\hat{b}=0$, transmitter uses the $\inx^{\mbox{th}}$ member of a capacity achieving sequence, $\SC'$ such that $\Emd_{,\SC'} =\PCAP$ to convey the message. We know that such a sequence exists because of Theorem \ref{thm:md}. Let the message of $\SC$ be the message of $\SC'$, i.e. the auxiliary message,
    \begin{equation*}
      \ames=\mes.
    \end{equation*}
If at the end of the second phase $\hat{\ames}=1$, receiver declares an erasure, $\tilde{\mes}=\era$,  else $\mes$ is decoded  $\hat{\mes}=\tilde{\mes}=\hat{\ames}$.}
\item{If $\hat{b}=1$, transmitter uses a length $\inx$ repetition code to convey whether $\mes=1$ or not.
    \begin{itemize}
    \item  If $\mes=1$,     $\ames=1$ and transmitter sends  the codeword $(\xa,\xa,\ldots,\xa)$.
    \item  If $\mes\neq 1$, $\ames=0$ and transmitter sends  the codeword $(\xd,\xd,\ldots,\xd)$.
    \end{itemize}
       where $\xa$ and $\xd$ are the maximizers achieving $\DX$:
    \begin{equation*}
        \DX=\max_{i,j} \KLD{W_{Y|x}(\cdot|i)}{W_{Y|X} (\cdot|j)}=  \KLD{W_{Y|x}(\cdot|\xa)}{W_{Y|X} (\cdot|\xd)}
      \end{equation*} 
Receiver decodes $\hat{\ames}=1$ only when output sequence is typical with $W_{Y|X}(\cdot|\xa)$. Evidently as before we have, \cite[Corrollary 1.2, p19]{CK}.
\begin{align}
\label{eq:fa-asym-exp1}
  \PCX{\hat{\ames}=0}{\mes=1}
&\leq \delta_{\inx}
\\
\label{eq:fa-asym-exp2}
  \PCX{\hat{\ames}=1}{\mes=0}
&\leq e^{- \inx (\DX-\delta_{\inx})}
\end{align}
where  $\displaystyle{\lim_{\inx \rightarrow \infty} \delta_{\inx}=0}$.
 
If $\hat{\ames}=1$ then  $\hat{\mes}=1$, else receiver declares erasure for the whole block, i.e. $\tilde{\mes}=\era$.}
\end{itemize}

Now we can calculate the error and erasure probabilities for $(\lceil \inx \rceil+ \inx)$ long block code. Using the equations (\ref{eq:fa-exp}), (\ref{eq:fa-asym-exp1}), (\ref{eq:fa-asym-exp2}) and Bayes' rule we get
\begin{align}
\label{eq:fa-manymes1}  
\PCX{\tilde{\mes}= \era}{\mes=1} 
&\leq e^{-\mu \sqrt{\inx}} + \delta_{\inx}&&\\
\label{eq:fa-manymes2}  
\PCX{\tilde{\mes}= \era}{\mes=i} 
&\leq e^{-\mu \sqrt{\inx}}+  \Pe_{\SC'}^{(\inx)} && i\neq 1 \\
\label{eq:fa-manymes3}  
\PCX{\tilde{\mes} \in {\cal M} \setminus \{1\}}{\mes=1} 
&\leq e^{-\mu \sqrt{\inx}} \Pe_{\SC'}^{(\inx)}(1) &&\\
\label{eq:fa-manymes4}  
\PCX{\tilde{\mes}  \in {\cal M} \setminus \{1,i\}}{\mes=i} 
&\leq \Pe_{\SC'}^{(\inx)}  && i \neq 1 \\
\label{eq:fa-manymes5}  
\PCX{\tilde{\mes}= 1}{\mes = i} 
&\leq  e^{-\mu \sqrt{\inx}} e^{- \inx (\DX-\delta_{\inx})} && i \neq 1
  \end{align}

Whenever $\tilde{\mes}=\era$ than transmitter tries  to send the message again from scratch, using same strategy. Consequently  all of the above error probabilities are scaled by a factor of $\tfrac{1}{1-\PCX{\tilde{\mes}= \era}{\mes=i}}$ when we consider the corresponding error probabilities for the variable decoding time code. Furthermore
\begin{equation}
  \label{eq:fa-manymes6}
 \ECX{\blx}{\mes=i} =\tfrac{ \inx+\sqrt{\inx} }{1-\PCX{\tilde{\mes}= \era}{\mes=i}}
\end{equation}
Using equations (\ref{eq:fa-manymes1}), (\ref{eq:fa-manymes2}), (\ref{eq:fa-manymes3}), (\ref{eq:fa-manymes4}), (\ref{eq:fa-manymes5}) and (\ref{eq:fa-manymes6})  we conclude that  $\SC$ is a capacity achieving code with $\FEmd_{,\SC}=\PCAP$ and $\FEfa_{,\SC}=\DX$.
\end{proof}

\subsubsection{Converse: $\FEfa \leq \DX$ }~

\begin{proof}
Note that as result of convexity of KL divergence we have
\begin{align}
  \notag
  \ECX{ \ln \tfrac{\PCX{Y^{\blx}}{\mes = 1}}{\PCX{Y^{\blx}}{\mes \neq 1}}}{\mes = 1}
&\geq
 \PCX{{\cal G}(1)}{\mes = 1} \ln \tfrac{ \PCX{{\cal G}(1)}{\mes = 1} } { \PCX{{\cal G}(1)}{\mes \neq 1} }+
 \PCX{\overline{{\cal G}(1)}}{\mes =1} \ln \tfrac{ \PCX{\overline{{\cal G}(1)}}{\mes = 1} } { \PCX{\overline{{\cal G}(1)}}{\mes \neq 1} }\\
&\geq
 -\ln 2+ \PCX{{\cal G}(1)}{\mes = 1} \ln \tfrac{ 1 }{ \PCX{{\cal G}(1)}{\mes \neq 1} }
\label{eq:fa-con1}
\end{align}

It has already been proved in \cite{burna2} that,
\begin{equation}
\label{eq:fa-con2} \ECX{ \ln \tfrac{\PCX{Y^{\blx}}{\mes = 1}}{\PCX{Y^{\blx}}{\mes \neq 1}}}{\mes = 1} \leq \DX \ECX{\blx}{\mes=1}
\end{equation}
Note that as a result of definition of ~$\Gamma$~ we have $\ECX{\blx}{\mes=1} \leq \EX{\blx} \Gamma$ using this together with equations (\ref{eq:fa-con1}) and (\ref{eq:fa-con2})  the  we get,
\begin{equation*}
   \PCX{{\cal G}(1)}{\mes \neq 1} \geq e^{-\tfrac{\ln 2+\Gamma \DX  \EX{\blx}}{\PCX{{\cal G}(1)}{\mes = 1}}}
\end{equation*}
Thus for any uniform delay reliable sequence, $\SC$, we have $\FEfa_{,\SC} \leq \DX$.
\end{proof}

\appendix
\subsection{Equivalent definitions of  \uep~  exponents}
We could have defined all the \uep~ exponents in this paper without using the notion of capacity achieving sequences.  As an example in this section  we define the single-bit exponent in this alternate manner and show that both definitions leads to identical results. In this alternative  first $\bar{\Eb}(R)$ is defined as the best exponent for the special bit at a given data-rate $R$ and then it is minimized over all $R<C$ to obtain $\bar{\Eb}$.

\begin{definition}
  For any $R\geq 0$, ${\cal Z}(R)$ is the set of  sequence of codes, $\SC$, with message sets ${\cal M}^{(n)}$ such that
  \begin{equation*}
    |{\cal M}^{(n)}|\geq e^{R\BLX} \qquad \mbox{and} \qquad {\cal M}^{(n)}= {\cal M}_1 \times {\cal M}_2^{(n)}
  \end{equation*}
where ${\cal M}_1=\{0,1\}$.
\end{definition}
\begin{definition}
 For a sequence of codes, $\SC$, such that $\displaystyle{ \lim_{\BLX \rightarrow \infty}\SPX{\BLX}{\hat{\mes}\neq \mes}=0}$, singe bit exponent $\Eb_{,\SC}$ equals
\begin{equation}
\Eb_{,\SC} \DEF \liminf_{\BLX \rightarrow \infty} \tfrac{ -\ln \SPX{\BLX}{\hat{\mes_1}\neq\mes_1}}{\BLX} .
\end{equation}
 \end{definition}
\begin{definition}
$\bar{\Eb}(R)$ and  the single bit exponent $\bar{\Eb}$ are defined as
\begin{align*}
\bar{\Eb}(R) 
&\DEF \sup_{\SC \in {\cal Z}(R)}\Eb_{,\SC}\\
\bar{\Eb} 
&\DEF \inf_{R<C}\bar{\Eb}(R).  
\end{align*}
\end{definition}
Note that according to this definition the special bit can achieve the exponent $\bar\Eb$, no matter how close the rate is to capacity. We now show  why this definition is equivalent to the earlier definition in terms of capacity achieving sequences given in section \ref{sec:nofeed}.
\begin{lemma}
\ $\bar\Eb=\Eb$
\end{lemma}
\begin{proof}
{\textbf{$\Eb\le\bar\Eb$}:}\\
By definition of $\Eb$, for any given $\delta>0$, there exists a capacity-achieving sequence $\SC$ such that $\Eb_{\SC}=\Eb$ and  for large enough $\BLX$, $R^{(\BLX)}\geq \CX -\delta$.  If we replace first $\BLX$ members of $\SC$ with codes whose rate are $(\CX-\delta)$ or higher we get another sequence  $\SC'$ such that $\SC' \in {\cal Z}(\CX -\delta)$ where $\Eb_{\SC'} =\Eb$. Thus $\bar{\Eb}(\CX -\delta) \geq \Eb$ for all $\delta>0$. Consequently 
\begin{equation*}
  \bar{\Eb}\geq \Eb
\end{equation*}
 {\textbf{$\Eb\ge \bar\Eb$}:}\\
 Let us first fix an arbitrarily small $\delta>0$. In the table in Figure
\ref{fig:table}, row $k$ represents a code-sequence $\bar\SC_k \in {\cal Z}(C-1/k)$, whose single-bit exponent
\[\Eb_{,\bar\SC_k}\ge\bar{\Eb}(R)-\delta\] 
Let $\bar\SC_k(l)$ represent length-$l$ code in this sequence. We construct a capacity achieving sequence $\SC$ from this table by sequentially choosing  elements of $\SC$ from rows $1,2,\cdots$  as follows .

   \begin{figure}[!h]
   \setlength{\unitlength}{1cm}
   \centering
   \begin{picture}(15,6)(-.5,1.5)

   \multiput(0,6.5)(0,-.7){7}{\line(1,0){14.5}}
   \multiput(0,6.5)(1.4,0){10}{\line(0,-1){5}}
   \put(-0.5,6.05){$\bar\SC_1$}
   \put(-0.5,5.35){$\bar\SC_2$}
   \put(-0.5,4.65){$\bar\SC_3$}
   \put(-0.5,3.95){$\bar\SC_4$}
   \multiput(-0.3,3.70)(0,-.2){11}{$\cdot$}

\put(0.20,6.05){{$ \bar\SC_1(1)$}}
\put(1.60,6.05){{$\bar\SC_1(2)$}}
\put(3.00,6.05){{$\bar\SC_1(3)$}}
\put(4.40,6.05){{$\bar\SC_1(4)$}}

\put(0.20,5.35){{$\bar\SC_2(1)$}}
\put(1.60,5.35){{$\bar\SC_2(2)$}}
\put(3.00,5.35){{$\bar\SC_2(3)$}}

\put(0.20,4.65){{$\bar\SC_3(1)$}}
\put(1.60,4.65){{$\bar\SC_3(2)$}}
\put(0.20,3.95){{$\bar\SC_4(1)$}}

\put(7,7.0){Block Length}
\put(0.55,6.6){$1$}
\put(1.95,6.6){$2$}
\put(3.35,6.6){$3$}

\multiput(3.75,6.6)(.27,0){14}{$\cdot$}
\put(7.55,6.6){$\BLX_1$}
\multiput(8.05,6.6)(.27,0){3}{$\cdot$}
\put(8.95,6.6){$\BLX_2$}
\multiput(9.45,6.6)(.27,0){8}{$\cdot$}
\put(11.75,6.6){$\BLX_3$}
\multiput(12.25,6.6)(.27,0){6}{$\cdot$}

\multiput(5.4,6.05)(.27,0){8}{$\cdot$}
\linethickness{.45mm}
\put(7.55,6.15){\line(1,0){1.40}}
\multiput(8.95,6.05)(.27,0){18}{$\cdot$}
\put(8.95,6.15){\line(0,-1){.70}}

\multiput(4.05,5.35)(.27,0){18}{$\cdot$}
\put(8.95,5.45){\line(1,0){2.8}}
\multiput(11.85,5.35)(.27,0){7}{$\cdot$}
\put(11.75,5.45){\line(0,-1){.70}}

\multiput(2.7,4.65)(.27,0){34}{$\cdot$}
\put(11.75,4.75){\vector(1,0){2}}

\multiput(1.2,3.95)(.33,0){38}{$\cdot$}
\multiput(0.3,3.25)(.33,0){41}{$\cdot$}
\multiput(0.3,2.55)(.33,0){41}{$\cdot$}
\multiput(0.3,1.85)(.33,0){41}{$\cdot$}

   \end{picture}
\caption{Row $k$ denotes a reliable code sequence at rate $C-1/k$. Bold path shows capacity achieving sequence $\SC$.}
\label{fig:table}
   \end{figure}

\begin{itemize}
\item For each sequence $\bar\SC_i$, let $\BLX_i$ denote the smallest block length $\BLX$ at which,
  \begin{enumerate}
  \item The single bit error probability satisfies 
\begin{equation*}
\SPX{\BLX}{\hat{\mes}_1 \neq\mes_1} \le e^{-\BLX(\bar{\Eb}(R)-2\delta))}
\end{equation*}
 \item The over all error probability satisfies
\begin{equation*}
\SPX{\BLX}{\hat{\mes} \neq\mes} \le 1/i 
\end{equation*}
\item $\BLX_{i}\geq \BLX_{i-1}$
  \end{enumerate}
\item  Given the sequence, $\BLX_1,\BLX_2,\cdots$, we choose the members of our capacity achieving code from the  code-table shown in Figure \ref{fig:table} as follows.
  \begin{itemize}
  \item \textbf{Initialize:} We use first $\BLX_2-1$ members of $\bar\SC_{1}$ as the first $\BLX_2-1$ members of the new code.
  \item \textbf{Iterate:}  We choose codes of length $\BLX_i$ to $\BLX_{i+1}-1$ from the  code sequence $\bar\SC_{i+1}$, i.e.,
\[\(\bar\SC_i(\BLX_i),\bar\SC_i(\BLX_i+1)\cdots,\bar\SC_i(\BLX_{i+1}-1)\) \] 
  \end{itemize}
\end{itemize}
Thus $\SC$ is a sampling of the code-table as shown by the bold path in Figure \ref{fig:table}. Note that this choice of $\SC$ is a capacity achieving sequence, moreover it will also achieve a single bit exponent
\[\Eb_{,\SC}=\inf_{R<C}\{\bar{\Eb}(R)-2\delta\} =\bar{\Eb}-2\delta\]
Choosing arbitrarily small $\delta$ proves $\Eb\ge\bar{\Eb}$.
\end{proof}

\section*{Acknowledgment}
The authors are indebted to Bob Gallager for his insights and encouragement for this work in general. In particular, Theorem \ref{thm:mdmany} was mainly inspired from his remarks. Helpful discussions with David Forney  and Emre Telatar are also gratefully acknowledged.

\bibliographystyle{plain}

\end{document}